%% file: article.tex
\documentclass[a4paper,UKenglish]{lipics-v2016}
 
\usepackage{microtype}

\title{Randomized Rumor Spreading Revisited}

\author[1]{Benjamin Doerr}
\author[2]{Anatolii Kostrygin}
\affil[1]{Laboratoire d'Informatique (LIX), \'Ecole Polytechnique, Palaiseau, France\\
{\tt doerr@lix.polytechnique.fr}}
\affil[1]{Laboratoire d'Informatique (LIX), \'Ecole Polytechnique, Palaiseau, France\\
{\tt kostrygin@lix.polytechnique.fr}}
\authorrunning{B. Doerr and A. Kostrygin} 

\Copyright{Benjamin Doerr and Anatolii Kostrygin}


\usepackage{amsmath}
\usepackage{amsthm}
\usepackage{amssymb}
\usepackage{mathrsfs}
\usepackage{accents}

\newcommand{\N}{\mathbb{N}}
\newcommand{\Z}{\mathbb{Z}}
\newcommand{\R}{\mathbb{R}}

\usepackage{graphicx}
\usepackage{color}

\usepackage{caption}

\renewcommand{\theenumi}{(\roman{enumi})}

\newtheorem*{recall*}{Recall}

\theoremstyle{plain}
	\newtheorem*{claim*}{Claim}
	\newtheorem*{lemma*}{Lemma}
	\newtheorem*{observation*}{Observation}

	\newtheorem{observation}{Observation}

\theoremstyle{definition}
	\newtheorem*{def*}{Definition}
	\newtheorem{defn}{Definition}
	
\theoremstyle{remark}
	\newtheorem*{remark*}{Remark}

\renewcommand{\Pr}{\mathbb{P}}

\newcommand{\eps}{\varepsilon}
\newcommand{\Expect}{\mathbb{E}}
\newcommand{\E}{\Expect}
\newcommand{\Pk}{p_k}
\newcommand{\Ck}{c_k}
\newcommand{\Cov}{\operatorname{Cov}}
\newcommand{\Var}{\operatorname{Var}}

\newcommand{\Geom}{\operatorname{Geom}}
\newcommand{\Bin}{\operatorname{Bin}}
\newcommand{\dominated}{\preceq}

\newcommand{\fun}{\tfrac{u}{n}}

\newcommand{\AND}{\wedge}

\newcommand{\NOT}{\neg}

\date{}

\begin{document}
\maketitle

\begin{abstract}
  We develop a simple and generic method to analyze randomized rumor spreading processes in fully connected networks. In contrast to all previous works, which heavily exploit the precise definition of the process under investigation, we only need to understand the probability and the covariance of the events that uninformed nodes become informed. This universality allows us to easily analyze the classic push, pull, and push-pull protocols both in their pure version and in several variations such as messages failing with constant probability or nodes calling a random number of others each round. Some dynamic models can be analyzed as well, e.g., when the network is a $G(n,p)$ random graph sampled independently each round [Clementi et al.\ (ESA 2013)].
  
  Despite this generality, our method determines the expected rumor spreading time precisely apart from additive constants, which is more precise than almost all previous works. We also prove tail bounds showing that a deviation from the expectation by more than an additive number of  $r$ rounds occurs with probability at most $\exp(-\Omega(r))$. 

  We further use our method to discuss the common assumption that nodes can answer any number of incoming calls. We observe that the restriction that only one call can be answered leads to a significant increase of the runtime of the push-pull protocol. In particular, the double logarithmic end phase of the process now takes logarithmic time. This also increases the message complexity from the asymptotically optimal $\Theta(n \log\log n)$ [Karp, Shenker, Schindelhauer, V\"ocking (FOCS 2000)] to $\Theta(n \log n)$. We propose a simple variation of the push-pull protocol that reverts back to the double logarithmic end phase  and thus to the $\Theta(n \log\log n)$ message complexity.
\end{abstract}

\section{Introduction}

Randomized rumor spreading is one of the core primitives to disseminate information in distributed networks. It builds on the paradigm that nodes call random neighbors and exchange information with these contacts. This gives highly robust dissemination algorithms belonging to the broader class of gossip-based algorithms that, due to their epidemic nature, are surprisingly efficient and scalable. Randomized rumor spreading has found numerous applications, among others, maintaining the consistency of replicated databases~\cite{Demers87}, disseminating large amounts of data in a scalable manner~\cite{MatosSFOR12}, and organizing any kind of communication in highly dynamic and unreliable networks like wireless sensor networks and mobile ad-hoc networks~\cite{IwanickiS10}. Randomized rumor spreading processes are also used to model epidemic processes like viruses spreading over the internet~\cite{BergerBCS05}, news spreading in social networks~\cite{DoerrFF11}, or opinions forming in social networks~\cite{Kleinberg08}.

The importance of these processes not only has led to a huge body of experimental results, but, starting with the influential works of Frieze and Grimmett~\cite{FriezeG85} and Karp, Shenker, Schindelhauer, and V\"ocking~\cite{KarpSSV00} also to a large number of mathematical analyses of rumor spreading algorithms giving runtime or robustness guarantees for existing algorithms and, based on such findings, proposing new algorithms.

Roughly speaking, two types of results can be found in the literature, general bounds trying to give a performance guarantee based only on certain graph parameters and analyses for specific graphs or graph classes. In the domain of general bounds, there is the classic maximum-degree-diameter bound of~\cite{FeigePRU90} and more recently, a number of works bounding the rumor spreading time in terms of conductance or other expansion properties~\cite{MoskAoyamaS08,ChierichettiLP10stoc,Giakkoupis11,Giakkoupis14}, which not only greatly helped our understanding of existing processes, but could also be exploited to design new dissemination algorithms~\cite{CensorHillelHKM12,HillelS10,CensorHillelS11,Haeupler13}.
The natural downside of such general results is that they often do not give sharp bounds. It seems that among the known graph parameters, none captures very well how suitable this network structure is for randomized rumor spreading. Also, it has to be mentioned that these results mostly apply to the push-pull protocol.

The other research direction followed in the past is to try to prove sharper bounds for specific graph classes. This led, among others, to the results that the push-protocol spreads a rumor in a complete graph in time $\log_2 n + \ln n \pm \omega(1)$ with high probability $1-o(1)$ (whp.)~\cite{Pittel87} (and in time $\log_{2-p} n + \frac 1p \ln n \pm o(\log n)$ when messages fail independently with probability $p$), whereas the push-pull protocol does so in time $\log_3 n + O(\log\log n)$~\cite{KarpSSV00}. The push protocol spreads rumors in hypercubes in time $O(\log n)$ whp.~\cite{FeigePRU90}, determining the leading constant is a major open problem. For Erd\H os-R\'eny random graphs with edge probability asymptotically larger than the connectivity threshold, again a runtime of $\log_{2-p} n + \frac 1p \ln n \pm o(\log n)$ was shown for the push protocol allowing transmission errors with rate $p$~\cite{FountoulakisHP10}. For preferential attachment graphs, which are often used as model for real-world networks, it was proven that the push-protocol needs $\Omega(n^\alpha)$ rounds, $\alpha>0$ some constant, whereas the push-pull protocol takes time $\Theta(\log n)$ and $\Theta((\log n)/\log\log n)$ when nodes avoid to call the same neighbor twice in a row~\cite{ChierichettiLP09,DoerrFF11}. Even faster rumor spreading times were shown on Chung-Lu power-law random graphs~\cite{FountoulakisPS12}.

One weakness of all these results on specific graphs is that they very much rely on the particular properties of the protocol under investigation. Even in fully connected networks (complete graphs), the existing analyses for the basic push protocol~\cite{FriezeG85,Pittel87,DoerrK14}, the push protocol in the presence of transmission failures~\cite{DoerrHL13}, the push protocol with multiple calls~\cite{PanagiotouPS15}, and the push-pull protocol~\cite{KarpSSV00} all uses highly specific arguments that cannot be used immediately for the other processes. This is despite the fact that the global behavior of these processes is often very similar. For example, all processes mentioned have an exponential expansion phase in which the number of informed node roughly grows by a constant factor until a constant fraction of the nodes is informed. Clearly, this hinders a faster development of the field. Note that the typical analysis of a rumor spreading protocol in the papers cited above needs between six and eight pages of proofs.

\subsection*{Our Results}\label{sec:results}

In this work, we make a big step forward towards overcoming this weakness. We propose a \emph{general analysis method} for all symmetric and memoryless rumor spreading processes in complete networks. It allows to easily analyze all rumor spreading processes mentioned above and many new ones. The key to this generality is showing that the rumor spreading times for these protocols are determined by the probabilities $p_k$ of a new node becoming informed in a round starting with $k$ informed nodes together with a mild bound on the covariance on the indicator random variables of the events that new nodes become informed. Consequently, all other particularities of the protocol can safely be ignored.

Despite this generality, our method gives bounds for the expected rumor spreading time that are \emph{tight apart from an additive constant number of rounds}. Such tight bounds so far have only been obtained once, namely for the basic push protocol~\cite{DoerrK14}. 

Our method also gives \emph{tail bounds} stating that deviations from the expectation by an additive number of at least $r$ of rounds occur with probability at most $A' \exp(-\alpha'r)$, where $A',\alpha'>0$ are absolute constants. Such a precise tail bound was previously given only for the push protocol in~\cite{DoerrK14}. Note that our tail bounds imply the usual whp-statements, e.g., that overshooting the expectation by any $\omega(1)$ term happens with probability $o(1)$ only, and that a rumor spreading time of $O(\log n)$ can be obtained with probability $1-n^{-c}$, $c$ any constant, by making the implicit constant in the time bound large enough.

We use our method to obtain the following particular results. We only state the expected runtimes. In all cases, the above tail bounds are valid as well.

\textbf{Classic protocols, robustness:} We start by analyzing the three basic push, pull, and push-pull protocols. In the \emph{push protocol}, in each round each informed node calls a random node and sends a copy of the rumor to it. In the \emph{pull protocol}, in each round each uninformed node calls a random node and tries to obtain the rumor from it. In the \emph{push-pull protocol}, all nodes contact a random node and in each such contact the informed nodes send rumor to the communication partner. 

For these three protocols, both in the fault-free setting and when assuming that calls fail independently with probability $1-p$, our method easily yields the expected rumor spreading times given in Table~\ref{tab:results}. Note that all previous works apart from~\cite{DoerrK14} did not state explicitly a bound for the expected runtime. Note further that for half of the settings regarded in Table~\ref{tab:results} no previous result existed. In particular, we are the first to find that the double logarithmic shrinking phase observed by Karp et al.~\cite{KarpSSV00} for the push-pull protocol disappears when messages fail with constant probability $p$, and is instead replaced by an ordinary shrinking regime with the number of uninformed nodes reducing by roughly a factor of $(1-p)e^{-p}$ each round. This observation is not overly deep, but has the important consequence that the message complexity of the push-pull protocol raises from the theoretically optimal $\Theta(n \log\log n)$ value proven in~\cite{KarpSSV00} to an order of magnitude of $\Theta(n \log n)$ in the presence of a constant rate of transmission errors. Hence the significant superiority of the push-pull protocol over the push protocol in the fault-free setting reduces to a constant-factor advantage in the faulty setting.

\begin{table}
{\small
	\centering
		\begin{tabular}{|p{1.4cm}|p{5.0cm}|p{5.8cm}|}\hline
		& no transmission failures & calls fail indep.\ with prob.\ $1-p \in (0,1)$\\\hline
		push  \newline protocol & $\E[T] = \log_2 n + \ln n \pm O(1)$ \newline $\lfloor \log_2 n \rfloor + \ln n - 1.116 \le \E[T] \le$\newline$\lceil \log_2 n \rceil + \ln n + 2.765 + o(1)$ \cite{DoerrK14} & $\E[T] = \log_{1+p} n + \tfrac 1p \ln n \pm O(1)$ \newline $T = \log_{1+p} n + \tfrac 1p \ln n \pm o(\log n)$ whp. \cite{DoerrHL13} \\\hline
		pull \newline protocol & $\E[T] = \log_2 n + \log_2 \ln n \pm O(1)$ \newline  & $\E[T] = \log_{1+p} n + \frac{1}{\ln\frac{1}{1-p}} \ln n \pm O(1)$ \newline  \\\hline
		push-pull protocol & $\E[T] = \log_3 n + \log_2 \ln n \pm O(1)$ \newline $T = \log_3 n \pm O(\log\log n)$ whp. \cite{KarpSSV00} & $\E[T] = \log_{1+2p} n + \frac{1}{p+\ln\frac{1}{1-p}} \ln n \pm O(1)$ \newline  \\\hline
					\end{tabular}
\caption{New and previous-best results for rumor spreading time $T$ of the classic rumor spreading protocols in complete graphs on $n$ vertices. The first line of each table entry contains the result that follows  from the method proposed in this work, the second line states the best previous result (if any). For all new bounds on the expected rumor spreading time, a tail bound of type $\Pr[|T \ge \Expect[T]| \ge r] \le A' \exp(-\alpha'r)$ with $A', \alpha' >0$ suitable constants follows as well from this work. In~\cite{DoerrK14}, such a bound was given for the rumor spreading time of the push protocol without transmission failures.}\label{tab:results}
}
\end{table}

\textbf{Multiple calls:} Panagiotou, Pourmiri, and Sauerwald~\cite{PanagiotouPS15} proposed a variation of the classic protocols in which the number of calls (always to different nodes) each node performs when active is a positive random variable $R$. They mostly assume that for each node, this random number is sampled once at the beginning of the process. For the case that $R$ has constant expectation and variance, they show that the rumor spreading time of the push protocol is $\log_{1+\E[R]} n + \frac 1 {\E[R]} \ln n  \pm o(\log n)$ with high probability and that the rumor spreading time of the push-pull protocol is $\Omega(\log n)$ with probability $1 - \eps$, $\eps > 0$. When $R$ follows a power law with exponent $\beta=3$, the push-pull protocol takes $\Theta(\frac{\log n}{\log\log n})$ rounds, and when $2 < \beta < 3$, it takes $\Theta(\log\log n)$ rounds.

The model of~\cite{PanagiotouPS15} makes sense when assuming that nodes have generally different communication capacities. To model momentarily different capacities, e.g., caused by being occupied with other communication tasks, we assume that the random variable is resampled for each node in each round. We also allow $R$ to take the value $0$. Again for the case $\E[R] = \Theta(1)$ and $\Var[R] = O(1)$, we show that the expected rumor spreading time of the push protocol is $\log_{1+\E[R]} n + \frac 1 {\E[R]} \ln n \pm O(1)$. The rumor spreading time of the push-pull protocol depends critically on the smallest value $\ell$ which $R$ takes with positive probability. If $\ell = 0$, that is, with constant probability nodes contact no other node, then there is no double exponential shrinking and the expected rumor spreading time is $\log_{1+2\E[R]}n + \tfrac1{\E[R]-\ln\Pr[R=0]}  \ln n \pm O(1)$. If nodes surely perform at least one call, then we have a double exponential shrinking regime and an expected rumor spreading time of $\log_{1+2\E[R]}n + \log_{1+\ell}\ln n \pm O(1)$.

\textbf{Dynamic networks:} We also show that our method is capable of analyzing dynamic networks when the dynamic is memory-less. Clementi et al.~\cite{ClementiCDFPS16} have shown that when the network in each round is a newly sampled $G(n,p)$ random graph, then for any constant $c$ the rumor spreading time of the push protocol is $\Theta(\log(n) / \min\{p,1/n\})$ with probability $1 - n^{-c}$. We sharpen this result for the most interesting regime that $p =a/n$, $a$ a positive constant. For this case, we show that the expected rumor spreading time is $\log_{2-e^{-a}} n + \frac{1}{1-e^{-a}} \ln(n) + O(1)$. 
Our tail bound $\Pr[|T - \Expect[T]| \ge r] \le A' \exp(-\alpha'r)$ for suitable constants $A', \alpha>0$ implies also the large deviation statement of~\cite{ClementiCDFPS16} (where for $\Theta(\log n)$ deviations in the lower tail the trivial $\log_2(n)$ lower bound holding with probability $1$ should be used). 

\textbf{Answering single calls only:} We finally use our method to discuss an aspect mostly ignored by previous research. While in all protocols above (apart from the one of~\cite{PanagiotouPS15}) it is assumed that each node can call at most one other node per round, it is tacitly assumed in the pull and push-pull protocols that nodes can answer all incoming calls. For complete graphs on $n$ vertices, the classic balls-into-bins theory immediately gives that in a typical round there is at least one node that receives $\Theta(\frac{\log n}{\log\log n})$ calls. So unlike for the outgoing traffic, nodes are implicitly assumed to be able to handle very different amounts of incoming traffic in one round.

The first to discuss this issue are Daum, Kuhn, and Maus~\cite{DaumKM15} (also the SIROCCO 2016 best paper). Among other results, they show that if only one incoming call can be answered and if this choice is taken adversarially, then there are networks where a previously polylogarithmic rumor spreading time of the pull protocol becomes $\tilde \Omega(\sqrt n)$. If the choice which incoming call is answered is taken randomly, then things improve and the authors show that for any network, the rumor spreading times of the pull and push-pull protocol increase by at most a factor of $O(\frac{\Delta(G)}{\delta(G)} \log n)$ compared to the variant in which all incoming calls are answered. Subsequently, Ghaffari and Newport~\cite{GhaffariN16} showed that with the restriction to accept only one incoming call, the general performance guarantees for the push-pull protocol in terms of vertex expansion or conductance~\cite{Giakkoupis11,Giakkoupis14} do not hold. Kiwi and Caro~\cite{KiwiC17} showed that solving the problem of multiple incoming calls via a FIFO queue can lead to extremely long rumor spreading times.

With our generic method, we can easily analyze this aspect of rumor spreading on complete graphs. While for the pull protocol only the growth phase mildly slows down, giving a total expected rumor spreading time of $\E[T] = \log_{2-1/e} n + \log_2 \ln n \pm O(1)$, for the push-pull protocol also the double logarithmic shrinking phase breaks down and we observe a total runtime of $\E[T] = \log_{3-2/e} n + \frac 12 \ln n \pm O(1)$ and, similarly as for the push-pull protocol with transmission failures, an increase of the message complexity to $\Theta(n \log n)$. The reason, as our proof reveals,  is that when a large number of nodes are informed, then their push calls have little positive effect (as in the classic push-pull protocol), but they now also block other nodes' pull calls from being accepted. This problem can be overcome by changing the protocol so that informed nodes stop calling others when the rumor is $\log_{3-2/e} n$ rounds old. The rumor spreading time of this modified push-pull protocol is $\E[T] = \log_{3-2/e} n + \log_2 \ln n \pm O(1)$ and, when halted at the right moment, this process takes $\Theta(n \log\log n)$ messages.

\section{Outline of the Analysis Method}

As just discussed, the main advantages of our approach are its universality and the very tight bounds it proves. We now briefly sketch the main new ideas that lead to this progress. Interestingly, they are rather simpler than the ones used in previous works.

\subsection{Tight Bounds via a Target-Failure Calculus}

We first describe how we obtain estimates for the rumor spreading time that are \emph{tight apart from additive constants}. Let us take as example the classic push protocol. It is easy to compute that in a round starting with $k$ informed nodes, the expected number of newly informed nodes is $E(k) = k - \Theta(k^2/n)$. Hence roughly speaking the number of informed nodes doubles each round (which explains the $\log_2 n$ part of the $\log_2 n + \ln n \pm O(1)$ rumor spreading time), but there is a growing gap to truly doubling caused by (i) calls reaching already informed nodes and (ii) several calls reaching the same target. This weakening of the doubling process was a main difficulty in all previous works.

The usual way to analyze this weakening doubling process is to partition the rumor spreading process in phases and within each phase to uniformly estimate the progress. For example, Pittel~\cite{Pittel87} considers 7 phases. He argues first that with high probability the number if informed nodes doubles until $n_1 = o(\sqrt{n})$ nodes are informed. Then, until $n_2 = n / \log^2(n)$ nodes are informed, with high probability in each round the number of informed nodes increases by at least a factor of $2 (1 - \frac{1}{\log^2(n)})$. Consequently, this second phase lasts at most $\log_{2 (1 - \frac{1}{\log^2(n)})}(n_2/n_1)$ rounds. While this type of argument gives good bounds for phases bounded away from the middle regime with both $\Theta(n)$ nodes informed and uninformed, we do not see how this ``estimating a phase uniformly'' argument can cross the middle regime without losing a number $\omega(1)$ of rounds.

For this reason, we proceed differently. To prove upper bounds on rumor spreading times, for each number $k$ of informed nodes, we formulate a pessimistic round target $E_0(k)$ that is sufficiently below the expected number $E(k)$ of newly informed nodes. Here ``sufficiently below'' means that the probability $q(k)$ to fail reaching this target number of informed nodes is small, but not necessarily $o(1)$ as in all previous analyses. Using a restart argument, we observe that the random time needed to go from $k$ informed nodes to at least $E_0(k)$ informed nodes is stochastically dominated by $1$ plus a geometric random variable with parameter $1 - q(k)$, where all our geometric random variables count the number of failures until success (this is one of the two definitions of geometric distributions that are in use). In particular, the expected time to go from $k$ to at least $E_0(k)$ informed nodes is at most $1 + \frac{q(k)}{1-q(k)}$.

The second, again elementary, key argument is that when we define a sequence of round targets by $k_0 := 1$, $k_1 := E_0(k_0)$, $k_2 := E_0(k_1), \dots$ with suitably defined $E_0(\cdot)$, then the $k_i$ grow almost like $2^i$ (in the example of the classic push protocol). More precisely, there is a $T = \log_2 n \pm O(1)$ such that $k_T = \Theta(n)$. Hence together with the previous paragraph we obtain that the number of rounds to reach $k_T$ informed nodes is dominated by $T$ plus a sum of independent geometric random variables. This sum has expectation $\sum_{i=0}^{T-1}   \frac{q(k_i)}{1-q(k_i)} = O(\sum_{i=0}^{T-1} q(k_i))$, so it suffices that the sum of the failure probabilities $q(k_i)$ is a constant (unlike in previous works, where it needed to be $o(1)$). A closer look at this sum also gives the desired tail bounds.

Similarly, to prove matching lower bounds, we define optimistic round targets $E_0(k)$ such that a round starting with $k$ informed nodes finds it unlikely to reach $E_0(k)$ informed nodes. Since again we want to allow failure probabilities that are constant, we now have to be more careful and also quantify the probability to overshoot $E_0(k)$ by larger quantities. This will then allow to argue that when defining a sequence of round targets recursively as above, then the expected number of targets overjumped (and thus the expected number of rounds saved compared to the ``one target per round'' calculus), is only constant. 

We remark that a target-failure argument similar to ours was used already in~\cite{DoerrK14}, there however only to give an upper bound for the runtime of the push protocol in the regime from $n^s$, $s$ a small constant, to $\Theta(n)$ informed nodes, that is, the later part of the exponential growth regime of the push process, in which via Chernoff bounds very strong concentration results could be exploited. Hence the novelty of this work with respect to the target-failure argument is that this analysis method can be used (i)~also from the very beginning of the process on, where we have no strong concentration, (ii)~also for the exponential and double exponential shrinking regimes of rumor spreading processes, and (iii)~also for lower bounds.

\subsection{Uniform Treatment of Many Rumor Spreading Processes}

As discussed earlier, the previous works regarding different rumor spreading processes on complete graphs all had to use different arguments. The reason is that the processes, even when looking similar from the outside, are intrinsically different when looking at the details. As an example, let us consider the first few rounds of the push and the pull protocol. In the push protocol, we just saw that while there are at most $o(\sqrt n)$ nodes informed, then a birthday paradox type argument gives that with high probability we have perfect doubling in each round. For the pull process, in which each uninformed node calls a random node and becomes informed when the latter was informed, we also easily compute that a round starting with $k$ informed nodes creates an expected number of $(n-k)\frac kn = k - \frac{k^2}n$ newly informed nodes. However, since these are binomially distributed, there is no hope for perfect doubling. In fact, for the first constant number of rounds, we even have a constant probability that not a single node becomes informed.

The only way to uniformly treat such different processes is by making the analysis depend only on general parameters of the process as opposed to the precise definition. Our second main contribution is distilling a few simple conditions that (i)~subsume essentially all symmetric and time-invariant rumor spreading processes on complete graphs and (ii)~suffice to prove rumor spreading times via the above described target-failure method. All this is made possible by the observation that the target-failure method needs much less in terms of failure probabilities than previous approaches, in particular, it can tolerate constant failure probabilities. Consequently, instead of using Chernoff and Azuma bounds for independent or negatively correlated random variables (which rely on the precise definition of the process), it suffices to use Chebyshev's inequality as concentration result.

Consequently, to apply our method we only need to (i)~understand (with a certain precision) the probability $p_k$ that an uninformed node becomes informed in a round starting with $k$ informed nodes; recall that we assumed symmetry, that is, this probability is the same for all uninformed nodes, and (ii)~we need to have a mild upper bound on the covariance of the indicator random variables of the events that two nodes become informed. 

The probabilities $p_k$ usually are easy to compute from the protocol definition. Also, we do not know them precisely. For example, for the growth phase of the push protocol discussed above, it suffices to know that there are constants $a<2$ and $a'$ such that for all $k < n/2$ we have $\frac kn (1 - a \frac kn) \le p_k \le \frac kn (1 + a' \frac kn)$.  This (together with the covariance condition) is enough to show that the rumor spreading process takes $\log_2 n \pm O(1)$ rounds to inform $n/2$ nodes or more. The constants $a, a'$ have no influence on the final result apart from the additive constant number of rounds hidden in the $O(1)$ term. The covariances are also often easy to bound with sufficient precision, among others, because many in processes the events that two uniformed nodes become informed are independent or negatively correlated.

In our general analysis method, we profit from the fact that seemingly all reasonable rumor spreading processes in complete networks can be described via three regimes:

\emph{Exponential growth:} Up to a constant fraction $fn$ of informed nodes, $p_k = \gamma_n \frac kn (1 \pm O(\frac kn))$. The number of informed nodes thus increases roughly by a factor of $(1+\gamma_n)$ in each round, hence the expected time to reach $fn$ informed nodes or more is $\log_{1+\gamma_n} n \pm O(1)$.

\emph{Exponential shrinking:} From a certain constant fraction $u = n-k = gn$ of uninformed nodes on, the probability of remaining uninformed satisfies $1-p_{n-k} = e^{-\rho_n} \pm O(\frac un)$. This leads to a shrinking of the number of uninformed nodes by essentially a factor of $e^{-\rho_n}$ per round. Hence when starting with $gn$ informed nodes, it takes another $\frac 1 {\rho_n} \ln n \pm O(1)$ rounds in expectation until all are informed.

\emph{Double exponential shrinking:} From a certain constant fraction $u = n-k = gn$ of uninformed nodes on, the probability of remaining uninformed satisfies $1-p_{n-k} = \Theta((\frac un)^{\ell-1})$. Now the expected time to go from $gn$ uninformed nodes to no uninformed node is $\log_\ell \ln n \pm O(1)$.

Due to their different nature, we cannot help treating these three regimes separately, however all with the target-failure method. Hence the main differences between these regimes lie in defining the pessimistic estimates for the targets, computing the failure probabilities, and computing the number of intermediate targets until the goal is reached. All this only needs computing expectations, using Chebyshev's inequality, and a couple of elementary estimates.

\section{Precise Statement of the Technical Results}\label{sec:tech}

In this work, we consider only \emph{homogeneous rumor spreading processes} characterized as follows. We always assume that we have $n$ nodes. Each node can be either \emph{informed} or \emph{uninformed}. We assume that the process starts with exactly one node being informed. Uninformed nodes may become informed, but an informed node never becomes uninformed. We consider a discrete time process, so the process can be partitioned into \emph{rounds}. In each round each uninformed node can become informed. Whenever a round starts with $k$ nodes being informed, then the probability for each uninformed node to become informed is some number $p_k$, which only depends on the number $k$ of informed nodes at the beginning of the round.

The main insight of this work is that for such homogeneous rumor spreading processes we can mostly ignore the particular structure of the process and only work with the \emph{success probabilities} $p_k$ defined above and the \emph{covariance numbers} $c_k$ defined as follows. 

\begin{defn}[Covariance numbers]
For a given homogeneous rumor spreading process and $k \in [1..n-1]$ let $c_k$ be the smallest number such that whenever a round starts with $k$ informed nodes and for any two uninformed nodes $x_1, x_2$, the indicator random variables $X_1, X_2$ for the events that these nodes become informed in this round satisfy \[\Cov[X_1,X_2] \le c_k.\]
\end{defn}
Upper bound for these covariances imply upper bounds on the variance of the number of nodes newly informed in a round. If the latter is small, Chebyshev's inequality yields that the actual number of newly informed nodes deviates not a lot from its expectation (which is determined by $p_k$).

Our main interest is studying after how many round all nodes are informed.

\begin{defn}[Rumor spreading times]
  Consider a homogeneous rumor spreading process. For all $t = 0, 1, \dots$ denote by $I_t$ the number of informed nodes at the end of the $t$-th round ($I_0 := 1$). Let $k \le m \le n$. By $T(k,m)$ we denote the time it takes to increase the number of informed nodes from $k$ to $m$ or more, that is,
    \[
        T(k,m) = \min\{ t-s | I_s = k \operatorname{ and } I_t \ge m \}.
    \]
    We call $T(1,n)$ the rumor spreading time of the process.
\end{defn}

As it turns out, almost all homogeneous rumor spreading processes can be analyzed via three regimes.

\subsection{Exponential Growth Regime} 

When not too many nodes are informed, in most rumor spreading processes we observe roughly a constant-factor increase of the number of informed nodes in one round, however, this increase becomes weaker with increasing number of informed nodes. 

\begin{defn}[Exponential growth conditions]
 Let $\gamma_n$ be bounded between two positive constants. Let $a, b, c \ge 0$ and $0 < f < 1$.
    We say that a homogeneous rumor spreading process satisfies the \emph{upper (respectively lower) exponential growth conditions} in $[1,fn[$ if for any $n \in \N$ big enough the following properties are satisfied for any $k < fn$.
    \begin{enumerate}
        \item $\Pk \ge \gamma_n \tfrac{k}{n} \cdot \left(1- a\tfrac{k}{n} - \tfrac{b}{\ln n}\right)$ (respectively $\Pk \le \gamma_n \tfrac{k}{n} \cdot \left(1+ a\tfrac{k}{n} + \tfrac{b}{\ln n}\right)$).
        \item $\Ck \le c\tfrac{k}{n^2}$.
    \end{enumerate}
    In the case of the upper exponential growth condition, we also require $af<1$.
\end{defn}

These growth conditions suffice to prove that in an expected time of at most (respectively at least) $\log_{1+\gamma_n} n \pm O(1)$ rounds a linear number of nodes becomes informed. Consequently, the decrease of the dissemination speed when more nodes are informed (quantified by the term $-a\frac kn$ in the upper exponential growth condition), which was a main difficulty in previous analyses, has only an $O(1)$ influence on the rumor spreading time.  

\begin{theorem}
  If a homogeneous rumor spreading process satisfies the upper (lower) exponential growth conditions in $[1,fn[$, then there are constants $A', \alpha'>0$ such that 
    \begin{align*}
    	&\Expect[T(1, fn)] \underset{(\ge)}{\le} \log_{1+\gamma_n} n \underset{(-)}{+} O(1), \\
    	&\Pr[T(1, fn) \underset{(\le)}{\ge} \log_{1+\gamma_n} n \underset{(-)}{+} r] \le A' \exp(-\alpha'r) \, \mbox{ for  all $r \in \N$}. 
    \end{align*}
  When the lower exponential growth conditions are satisfied, then also there is an $f' \in ]f,1[$ such that with probability $1-O\left(\tfrac1n\right)$ at most $f'n$ nodes are informed at the end of round $T(1,fn)$.
\end{theorem}

We note that the upper tail bound is tight apart from the implicit constants. This is witnessed, for example, by the pull protocol, where rounds starting with only a constant number of informed nodes have a constant probability of not informing any new node.

\subsection{Exponential Shrinking Regime} 

In a sense dual to the previous regime, in some rumor spreading processes (e.g., the push protocol as well as the pull and push-pull protocols in the presence of transmission failures) we observe that the number of uninformed nodes shrinks by a constant factor once sufficiently many nodes are informed. Again, the weaker shrinking at the beginning of this regime has only an $O(1)$ influence on the resulting rumor spreading times.

\begin{defn}[Exponential shrinking conditions]
Let $\rho_n$ be bounded between two positive constants. Let $0 < g < 1$, and $a, c \in \R_{\ge0}$.
	We say that a homogeneous rumor spreading process satisfies the \emph{upper (respectively lower) exponential shrinking conditions} if for any $n \in \N$ big enough, the following properties are satisfied for all $u = n-k \le gn$.
	\begin{enumerate}
		\item
			$1-p_k = 1-p_{n-u} \le e^{-\rho_n} + a\frac{u}{n}$\, (respectively $1-p_k = 1-p_{n-u} \ge e^{-\rho_n} - a\frac{u}{n}$).
		\item
			$c_k = c_{n-u} \le \frac{c}{u}$.
	\end{enumerate}
	For the upper exponential shrinking conditions, we also assume that $e^{-\rho_n} + ag < 1$.
\end{defn}

\begin{theorem}
	If a homogeneous rumor spreading process satisfies the upper (lower) exponential shrinking conditions, then there are $A' \alpha' > 0$ such that 
	\begin{align*}
		&\Expect[T(n-\lfloor g n \rfloor, n)] \underset{(\ge)}{\le} \tfrac1{\rho_n}\ln n \underset{(-)}{+} O(1),\\
		&\Pr[T(n-\lfloor g n \rfloor, n) \underset{(\le)}{\ge} \tfrac1{\rho_n}\ln n \underset{(-)}{+} r] \le A'\exp(-\alpha'r)\, \mbox{ for all $r \in \N$}.
  \end{align*}		
\end{theorem}

Again, the upper tail bound is tight apart from the constants as shown by the push protocol. Here, a round starting with $n-1$ informed nodes has a constant chance to not inform the remaining node.

\subsection{Double Exponential Shrinking Regime}

Protocols using pull operations in the absence of transmission failures display a faster reduction of the number if uninformed nodes.

\begin{defn}[Double exponential shrinking conditions]
	Let $g \in ]0,1]$, $\ell > 1$, and $a, c \in \R_{\ge0}$ such that $ag^{\ell-1} < 1$.
	We say that a homogeneous rumor spreading process satisfies the \emph{upper (respectively lower) double exponential shrinking conditions}
	if for any $n$ big enough the following properties are satisfied for all $u = n-k \in [1, gn]$.
	\begin{enumerate}
		\item $1-p_{n-u} \le a\left(\tfrac{u}{n}\right)^{\ell-1}$\, (respectively $1-p_{n-u} \ge a\left(\tfrac{u}{n}\right)^{\ell-1}$).
		\item $c_{n-u} \le c \tfrac{n}{u^2}$.
	\end{enumerate}
\end{defn}

\begin{theorem}
	If a homogeneous rumor spreading process satisfies the upper (lower) double exponential shrinking conditions, then there are $A', \alpha' > 0$ and $R$ (depending on $\alpha$) such that 
	\begin{align*}
		&\Expect[T(n-\lfloor g n \rfloor, n)] \underset{(\ge)}{\le} \log_\ell \ln n \underset{(-)}{+} O(1),\\
		&\Pr[T(n-\lfloor g n \rfloor, n) \ge \log_\ell \ln n + r] \le O(n^{-\alpha'r+A'})\, \mbox{ for all $r \in \N$},\\
		\mbox{(}&\Pr[T(n-\lfloor g n \rfloor, n) \le \log_\ell \ln n - R] \le O(n^{-1+2\ell\alpha})\mbox{)}.
  \end{align*}		
\end{theorem}

The last rounds of the push-pull protocol show that the upper tail bound is tight apart from the constants. The lower tail bound is clearly not best possible, but most likely good enough for most purposes.

\subsection{Connecting Regimes}

While often these above described three regimes suffice to fully analyze a rumor spreading process, occasionally it is necessary or convenient to separately regard a constant number of rounds between the growth and the shrinking regime. This is achieved by the following two lemmas.

\begin{lemma}
	Consider a homogeneous rumor spreading process. Let $0 < \ell < m < n$ and $0 < p < 1$. Suppose for any number $\ell \le k < m$, we have $p_{k} \ge p$. Then
  \begin{align*}
   	&\Expect[T(\ell,m)] \le \tfrac{n-\ell}{n-m} \cdot \tfrac1{p}\,,\\
   	&\Pr[T(\ell,m) > r] \le \tfrac{n-\ell}{n-m} \cdot (1-p)^r\, \mbox{ for all $r \in \N$}.
  \end{align*}
\end{lemma}

\begin{lemma}
    Let $f, p \in ]0,1[$ and $c > 0$.
    Suppose that for any $k < fn$ we have $p_k \le p$ and $c_k \le \tfrac{c}{n}$.
    Then there exists $f' \in ]f,1[$ such that with probability $1-O\left(\tfrac1n\right)$ at the end of some round the number of informed nodes will be between $fn$ and $f'n$.
\end{lemma}

\section{Applying the Above Technical Results}

In this section, we sketch how to use the above tools to obtain some of the results described in Section~\ref{sec:results}. Since it does not make a difference, to ease the notation we always assume that nodes call random nodes, that is, including themselves. The main observation is that computing the $p_k$ is usually very elementary. For the covariance conditions, often we easily observe a negative or zero covariance, but when this is not true, then things can become technical.

For the \emph{basic push, pull, and push-pull protocols}, we easily observe that all covariances to be regarded are negative or zero: Knowing that one uninformed node $x_1$ becomes informed in the current round has no influence on the pull call of another uninformed node $x_2$. When the protocol has push calls and $x_1$ was informed via a push call, then this event makes it slightly less likely that $x_2$ becomes informed via a push call, simply because at least one informed node is occupied with calling $x_1$.

The success probabilities $p_k$ are easy to compute right from the protocol definition. When $k$ nodes are informed, then the probabilities that an uninformed node becomes informed are 
\begin{equation*}
p_k = 
\begin{cases}
1 - (1-1/n)^k\, &\mbox{ for the push protocol,}\\
k/n\, &\mbox{ for the pull protocol,}\\
p_k = 1 - (1-1/n)^k \frac{n-k}n\, &\mbox{ for the push-pull protocol.}
\end{cases}
\end{equation*}
Using elementary estimates like $1 - k/n \le (1-1/n)^k \le 1 - k/n + k^2/2n^2$, we see that the push and pull protocols satisfy the exponential growth conditions with $\gamma_n = 1$, whereas the push-pull protocol does the same with $\gamma_n = 2$. The push protocol satisfies the exponential shrinking conditions with $\rho_n=1$. The pull and push-pull protocols satisfy the double exponential shrinking conditions with $\ell = 2$. All growth conditions are satisfied at least up to $k = n/2$ informed nodes and all shrinking conditions are satisfied at least for $u \le n/2$ uninformed nodes, so we do not need the intermediate lemmas. This proves our results given in Table~\ref{tab:results} for the fault-free case.


\emph{Faulty communication:} The same arguments (with different constants $\gamma_n$ and $\rho_n$) suffice to analyze these protocols when messages get lost independently with probability $1-p$. The only structural difference is that now for the pull and push-pull protocols uninformed nodes remain uninformed with at least constant probability. For this reason, now all three protocols have an exponential shrinking phase.

The push-pull protocol with the restriction that nodes \emph{answer only a single incoming call} randomly chosen among the incoming calls is an example where the exponential growth and shrinking conditions are harder to prove. To compute the $p_k$ we assume that all $n$ calls have a random unique priority in $[1..n]$ and that the call with lowest priority number is accepted. For fixed priority, the probability of being accepted is easy to compute, and this leads to the success probability of a pull call. For the probability to become informed via a push call, the simple argument that the first incoming call is from an informed node with probability $k/n$ solves the problem. When showing the covariance conditions, we face the problem that it is indeed not clear if we have negative or zero covariance. The event that some node becomes informed increases the chance that this node received a push call. This push call cannot interfere with another node's pull call to an informed node. So it does have some positive influence on the probability of another uninformed node to become informed. Fortunately, for our covariance conditions allow some positive correlation. Because of this, very generally speaking, we can ignore certain difficulties to handle situations when they occur rare enough.

\emph{Dynamic communication graphs:} Being maybe the result where it is most surprising that bounds sharp apart from additive constants can be obtained, we now regard in more detail a problem regarded in~\cite{ClementiCDFPS16}. There, the performance of push rumor spreading in a group of $n$ agents was investigated when the actual communication network is changing in each round. As one such dynamic models, it was assumed that the communication graph in each round is a newly sampled $G(n,p)$ random graph, that is, there is an edge independently with probability $p$ between any two vertices. 

For the ease of presentation, we assume that the edge probability equals $p = a/n$ for some constant $a>0$. This is clearly the most interesting case. For such (and larger) $p$, a rumor spreading time of $\Theta(\log n)$ was shown to hold with inverse-polynomial failure probability. Recalling that for $p=a/n$ the graph $G(n,p)$ is not connected and has vertex degrees ranging from $0$ to $\Theta(\log(n)/\log\log(n))$, this result is not obvious (as the proof in~\cite{ClementiCDFPS16} also indicates). Also, observe that the random graph is not newly sampled for each action of a node, so there are dependencies that have to be taken into account. 

For this setting, we now conduct a very precise analysis, which in particular makes precise the influence of the graph density parameter $a$.
\begin{theorem}\label{thm:random}
  Let $T$ be the time the push protocol needs to inform $n$ nodes when in each round a newly sampled $G(n,p)$, $p = a/n$, random graph represents the communication network. Then \[\Expect[T] = \log_{2-e^{-a}} n + \tfrac{1}{1-e^{-a}} \ln(n) \pm O(1)\] and there are constants $A',\alpha'>0$ such that $\Pr[|T - \Expect[T]| \ge r] \le A' \exp(\alpha'r)$ holds for all $r \in \N$.
\end{theorem}

Recall that a vertex is isolated with probability $e^{-a} + O(1/n)$. Clearly, an informed vertex when isolated necessarily fails to inform another vertex in this round. The rumor spreading time proven above is the same as the one for the case that the communication network is always a complete graph, but calls fail independently with probability $e^{-a}$. Hence in a sense the changing topology (with low vertex degrees) is not harmful apart from the effect that it creates isolated vertices with constant rate. We did not expect this. 

To prove Theorem~\ref{thm:random}, we first observe that the covariance properties are fulfilled. By symmetry, we can assume that in a round starting with $k$ informed nodes, we first sample the random graph and decide for each node which neighbor it potentially calls in this round, and only then decide randomly which $k$ nodes are informed and have these call the random neighbor determined before. Conditioning on the outcome of random graph, neighbor choice, and on that nodes $x$ and $y$ are not informed, in the remaining random experiment the events ``$x$ becomes informed'' and ``$y$ becomes informed'' clearly are negatively correlated. 

Estimating the probability $p_k$ for an uninformed node to become informed in a round starting with $k$ informed nodes, is slightly technical. Since it is unlikely that two neighbors of an uninformed node $x$ are connected by an edge, the main contribution to $p_k$ stems from the case that the informed neighbors of $x$ form an independent set. Conditioning on this outcome of the edges in $\{x\} \cup N(x)$, each informed neighbor of $x$ has an independent probability of roughly $(1-e^{-a})/a$ of calling $x$, giving (again taking care of the dependencies) a probability of roughly $1-p_k \approx (1 - \frac an \frac{1 - e^{-a}}{a})^k$ for the event that no informed node calls $x$. From this, we estimate $\frac kn (1-e^{-a}) (1 - \frac {k+O(1)}{2n}(1-e^{-a})) \le p_k \le \frac kn (1-e^{-a}+O(1/n))$, showing that the exponential growth conditions are satisfied with $\gamma_n = 1 - e^{-a}$. Similar arguments, again taking some care for the dependencies that the random graph imposes on the actions of informed neighbors, show that the upper exponential shrinking conditions are satisfied for $\rho_n = 1-e^{-a}$, whereas the lower exponential shrinking conditions are satisfied with $\rho_n = 1 - e^{-a} + O(\log(n)^2 /n)$. 

\section{Summary, Outlook}

In this work, we presented a general, easy-to-use method to analyze homogeneous rumor spreading processes on complete networks (including memoryless dynamic settings). Such processes are important in many applications, among others, due to the use of random peer sampling services in many distributed systems. Such processes also correspond to the fully mixed population model in mathematical epidemiology.

The two main strengths of our method are (i)~that it builds only on estimates for the probability and the covariance of the events that new nodes become informed---consequently, many processes can be analyzed with identical arguments (as opposed to all previous works), and (ii)~that it determines the expected rumor spreading time precise apart from additive constants (with tail bounds giving in most cases that deviations by an additive number $r$ of rounds occur with probability $\exp(-\Omega(r))$ only). The key to our results is distilling the right growth and shrinking conditions, which allow to describe essentially all previously regarded homogeneous processes, and to show, based on these conditions, that the usually present mild deviations from a perfect exponential growth or shrinking in total cost only a constant number of rounds.
%

From a broader perspective, this work shows that the traditional approach to randomized processes of splitting the analysis in several phases and then trying to understand each phase with uniform arguments might not be the ideal way to capture the nature of processes with a behavior changing continuously over time. While we demonstrated that the more careful round-target approach is better suited for homogeneous rumor spreading processes, one can speculate if similar ideas are profitable for other randomized algorithms or processes regarded in computer science.

\newpage
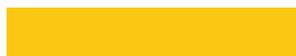
\appendix

\section*{APPENDIX}

This appendix contains material to be read at the reviewers' discretion. Since this appendix is much longer than the paper itself, to ease reading we not only give the parts left out in the body of the submission, but repeat (sometimes mildly reformulated) the technical parts of the body of the submission.

\input{introduction.tex}

\section{Main Analysis Technique}

As outlined earlier, in this work we attempt to develop a general analysis technique that covers a large class of rumor spreading problems in perfectly connected networks (complete graphs). To this aim, we define a general class of rumor spreading processes and then distill three regimes such that most rumor spreading processes regarded in the literature are covered by these regimes. For each regime, we prove rumor spreading times sharp apart from additive constants. We shall treat upper and lower bounds separately, so that in cases where only estimates in one direction are known, we still obtain this type of bound.

%
%
%
%
%

\subsection{Homogeneous Rumor Spreading Processes}

We now characterize the class of rumor spreading processes we aim at analyzing.

\begin{defn} [Homogeneous rumor spreading process]
  We always assume that we have $n$ nodes. Each node can be either \emph{informed} or \emph{uninformed}. We assume that the process starts with exactly one node being informed. Uninformed nodes may become informed, but an informed node never can become uninformed. We consider a discrete time process, so the process can be partitioned into \emph{rounds}. In each round each uninformed node can become informed. Whenever a round starts with $k$ nodes being informed, then the probability for each uninformed node to become informed is some number $p_k$, which only depends on the number $k$ of the informed nodes at the beginning of the round.
\end{defn}

The above definition is relatively abstract and, in principle, could be simply phrased as a Markov process on the number $k \in [1..n]$ of informed nodes. We still find it natural to use the language of rumor spreading. We will discuss many rumor spreading processes covered by this definition in Sections~\ref{sec:classics},~\ref{sec:more examples},~and~\ref{sec:single incoming call}, so let us for the moment only remark that the definition covers all processes regarded in the literature as long as they are memoryless (the events in the current round depend only on which nodes are informed) and symmetric (only the numbers of informed and uninformed nodes is relevant, but not which nodes these are). We remark that our methods can be applied to suitable processes that are not memoryless, see Section~\ref{sec:single call-fast} for an example that is not memoryless due to the use of a time counter. 

The main insight of this work is that we can mostly ignore the particular structure of a rumor spreading process and only work with the \emph{success probabilities} $p_k$ and the \emph{covariance numbers} $c_k$ defines as follows.

\begin{defn}[Covariance numbers]
For a given homogeneous rumor spreading process and $k \in [1..n-1]$ let $c_k$ be the smallest number such that whenever a round starts with $k$ informed nodes and for any two uninformed nodes $x_1, x_2$, the indicator random variables $X_1, X_2$ for the events that these nodes become informed in this round satisfy \[\Cov[X_1,X_2] \le c_k.\]
\end{defn}

It turns out that essentially all homogeneous rumor spreading processes have an \emph{exponential growth phase}, which is roughly characterized by the fact that for suitable constants $f \in (0,1]$, $c \in \R$ and $\gamma_n > 0$ we have for all $k \in [1..fn-1]$ both $p_k = \gamma_n \frac k n (1 \pm O(\frac kn))$ and $c_k \le c \frac k{n^2}$.

This growth phase is followed by one of the following two shrinking regimes. (i) Exponential shrinking regime: For suitable constants $g > 0$, $c > 0$, and $\rho_n > 0$, we have for all $u \le gn$ that $1-p_{n-u} = e^{-\rho_n} \pm \Theta(\frac un)$ and $c_{n-u} \le \frac cu$. In particular, in a round starting with $u \le gn$ uninformed nodes, we expect the number of uninformed nodes to shrink by a factor of roughly $e^{-\rho_n}$. (ii) Double exponential shrinking regime: For suitable constants $g > 0$ and $\ell>1$, we have that for all $u \le gn$ both $1-p_{n-u} = \Theta((\frac un)^{\ell-1})$ and $c_{n-u} \le c \frac n {u^2}$. In particular, we expect the fraction of uninformed nodes to be raised to some positive power $\ell-1$.

In the following subsections, we shall analyze each of these regimes, treating separately upper and lower bound guarantees. The very rough analysis idea is the same in each subsection, so we present and discuss it in more detail in the following subsection and then are more brief in the remaining ones.

Before doing so, we define the rumor spreading time and show an elementary fact that will be convenient several times in the following.

\begin{defn}[Rumor spreading times]
  Consider a homogeneous rumor spreading process. For all $t = 0, 1, \dots$ denote by $I_t$ the number of informed nodes at the end of the $t$-th round ($I_0 := 1$). Let $k \le m \le n$. By $T(k,m)$ we denote the time it takes to increase the number of informed nodes from $k$ to $m$ or more, that is,
    \[
        T(k,m) = \min\{ t-s | I_s = k \operatorname{ and } I_t \ge m \}.
    \]
    We call $T(1,n)$ the rumor spreading time of the process.
\end{defn}

Most homogeneous rumor spreading processes have the property that when a constant fraction of the nodes is informed, then each uninformed node has a constant positive probability of becoming informed in one round. In this situation, the following lemma allows to argue that an expected constant number of rounds suffices to go from any constant fraction of informed nodes to any constant fraction of uninformed nodes. This will be convenient in some the following proofs of upper bounds for rumor spreading times, namely when the growth or shrinking conditions are not strong enough near to the middle point of $n/2$ informed nodes.

\begin{lemma}\label{lem:general-connect}
	Consider a homogeneous rumor spreading process. Let $0 < \ell < m < n$ and $0 < p < 1$. Suppose for any number $\ell \le k < m$, we have $p_{k} \ge p$.
    Then
    \renewcommand{\theenumi}{(\roman{enumi})}%
    \begin{enumerate}
	    \item $\Pr[T(\ell, m) > r] \le \tfrac{n-\ell}{n-m} \cdot (1-p)^r$.
    	\item $\Expect[T(\ell,m)] \le \frac{n-\ell}{n-m} \cdot \frac1{p}\,.$
    \end{enumerate}
\end{lemma}

\begin{proof}
  Let $q := 1-p$. We regard a dummy process which coincides with the given process until the number of informed nodes is at least $m$. If there are at least $m$ nodes informed, then the dummy process shall be such that each uniformed node in each round independently becomes informed with probability $p$. Obviously, $T(\ell,m)$ is the same for both processes, so we consider the dummy process in the following.

In this dummy process, by the memorylessness of our rumor spreading process, an uninformed node remains uninformed for $r$ rounds with probability at most $q^r$. Hence the expected number $U_r$ of uninformed nodes after $r$ rounds is $\Expect[U_r] \le  (n-\ell) q^r$ and Markov's inequality gives
    \[
        \Pr[T(\ell,m) > r] = \Pr[ U_r > (n-m) ] < \frac{n-\ell}{n-m} \cdot q^r.
    \]
    Hence
    \begin{align*}
    	\Expect[T(\ell,m)]
    	& = \sum_{r=0}^\infty \Pr[ T(\ell,m) > r ]
			< \frac{n-\ell}{n-m} \sum_{r=0}^\infty q^r
			= \frac{n-\ell}{n-m} \cdot \frac1{1-q}\,\,.
    \end{align*}
%
%
%
%
\end{proof}

Similarly to the lemma above, the following lemma will be convenient in some of the proofs of lower bounds for rumor spreading times, again when the growth and shrinking conditions do not cover the whole process. In this case, the following lemma allows to argue that an arbitrarily small, but still constant fraction of uninformed nodes will be reached at some time.

\begin{lemma}\label{lem:general-connect-lower}
    Let $f, p \in ]0,1[$ and $c > 0$.
    Suppose that for any $k < fn$ we have $p_k \le p$ and $c_k \le \tfrac{c}{n}$.
    Then there exists $f' \in ]f,1[$ such that with probability $1-O\left(\tfrac1n\right)$ at the end of some round the number of informed nodes will be between $fn$ and $f'n$.
\end{lemma}
\begin{proof}
    Suppose $k < fn$.
    Denote by $X(k)$ the number of newly informed nodes in a round starting with $k$ informed nodes.
    Since $p_k \le p$, we have $\E[X(k)] \le pn(1-f) \le pn$.
    Then by Lemma~\ref{lem:prelim:variance} we have $\Var[X(k)] \le (p+c)n$.
    Let $f' \in ]f+p(1-f),1[$.
    Applying Chebyshev's inequality, we compute
    \begin{align*}
        & \Pr[k+X(k) \ge f'n] \le \Pr[X(k) \ge (f'-f)n] \\
        & \le \Pr[X(k) \ge \E[X(k)] + n(f'-f-p(1-f))] \\
        & \le \tfrac{\Var[X(k)]}{n^2(f'-f-p(1-f))^2}
            = \tfrac{p+c}{n(f'-f-p(1-f))^2} = O\left(\tfrac1n\right).
    \end{align*}
    Therefore, the probability that the process ``jumps over'' the interval $[fn,f'n]$ is $O\left(\tfrac1n\right)$.
\end{proof}

\input{exp_growth_new.tex}

{\sloppy
 \input{exp_shrinking_new.tex}
}


 \input{double_exp_shrinking_new.tex}

 \input{basic_push_pull.tex}

 \input{push_pull_ignorance.tex}
\bibliography{bibrumors}

\end{document}

%% file: introduction.tex
\section{Preliminaries}

In this section, for the sake of completeness, we collect some elementary facts which are well-known.

\subsection{Variance. Chebyshev's and Cantelli's Inequalities}
We recall that the \emph{variance} of a discrete random variable $X$ is $\Var[X] = \Expect[X^2]-\Expect[X]^2$.
By definition it is a measure of how well $X$ is concentrated around its mean.
The two following inequalities gives the bounds for the ``tail'' probabilities for any random variable $X$.
\begin{lemma}[Chebyshev's inequality]\label{lem:prelim:Chebyshev}
    For all $\lambda>0$,
    \[
        \Pr\left[|X-\Expect[X]|] \ge \lambda\sqrt{\Var[X]}\right] \le \tfrac1{\lambda^2}.
    \]
\end{lemma}
There is a one-sided version of the Chebyshev inequality attributed to Cantelli, replacing $\tfrac1{\lambda^2}$ by $\tfrac1{\lambda^2+1}$.
\begin{lemma}[Cantelli's inequality]\label{lem:prelim:Cantelli}
    For all $\lambda > 0$,
    \[
        \Pr\left[X-\Expect[X] \ge \lambda\sqrt{\Var[X]}\right] \le \tfrac1{1+\lambda^2}.
    \]
\end{lemma}
We remark that Cantelli's inequality gives the bound which is less than one for any positive $\lambda$.

In addition we provide a simple method to bound a variance of a sum of indicator random variables.
We recall that the \emph{covariance} of two discrete random variables $X$ and $Y$ is $\Cov[X,Y] = \E[(X-\E[X])(Y-\E[Y])]$.
\begin{lemma}\label{lem:prelim:variance}
    Let a random variables $X = \sum_{i=1}^n X_i$, where $X_i$ are indicator random variables.
    Suppose, for any $i \ne j$ we have $\Cov[X_i,X_j] \le c$ for some constant $c$.
    Then $\Var[X] \le \E[X] + cn^2$.
\end{lemma}
\begin{proof}
    Since $X_i$ is a binary random variable, $\Var[X_i] \le \E[X_i]$.
    Therefore,
    \[
        \Var[X] \le \sum_{i=1}^n \Var[X_i] + \sum_{i \ne j} \Cov[X_i,X_j]
        \le \E[X] + cn^2.
    \]
\end{proof}

%

\subsection{Geometric Distribution and Stochastic Domination}
    \begin{def*}
        We say that a random integer variable $G$ has a \emph{geometric distribution} with success probability $p$ and write $G \sim \Geom(p)$ if $\Pr[G=k] = p(1-p)^k$ for any $k \ge 0$.
    \end{def*}
    The geometric distribution corresponds the number of failed Bernoulli trials until the first success.
    Recall that if $G \sim \Geom(p)$, then we have $\E[G] = \tfrac{1-p}{p}$ and $\Var[G] = \tfrac{1-p}{p^2}$.

    Another important concept is the stochastic domination.
    Informally, a random variable $X$ dominates a random variable $Y$ if $X$'s distribution is ``to the right'' of the $Y$'s distribution.
    \begin{def*}\label{definition:stochastic domination}
        Let a pair of random variables $X, Y$ be given.
        We say that $X$ \emph{stochastically dominates} $Y$, and write $Y \dominated X$, if
        $\Pr[X \ge x] \ge \Pr[Y \ge x]$
        for all $x$.
    \end{def*}
    The stochastic domination satisfies the following elementary properties.
    \begin{itemize}
        \item if $X \dominated Y$ and $Z \dominated T$, then $X+Y \dominated Z+T$.
        \item if $X \dominated Y$ then $\E X \le \E Y$.
    \end{itemize}
    \begin{lemma}[\cite{DoerrK14}]\label{lemma:sum geometrical}
        Let $G_1, \ldots, G_n$ be independent random variables with $G_i \sim \Geom(1-q_i)$.
        Then $\sum_{i=1}^n G_i$ is stochastically dominated by a random variable $G$ with $G \sim \Geom(1-\sum_{i=1}^n q_i)$
    \end{lemma}
    

The following lemma contains a high probability bound for the sum of geometrically distributed variables in the case when $\sum_i q_i = O(1)$, but not necessarily less than 1.

\begin{lemma}~\label{lemma:sum geometrical-2}
	Let $\eps, \delta \in ]0,1[$ and $s > 0$.
	Let $q_j := \min\{1-\eps, s\delta^j\}$, for any $j$.
	Let $G$ be stochastically dominated by $\sum_{j=0}^{J-1} G_j$, where $G_j \sim \Geom(1-q_j)$.
	Then there exist constant $A, \alpha > 0$ such that for any integer $r>0$ we have $\Pr[G > r] \le Ae^{-\alpha r}$.
\end{lemma}
\begin{proof}
	Let $j_0$ is the smallest such that $\sum_{j\ge j_0} q_j < 1-\eps$.
	By construction, $j_0 = O(1)$.
	By Lemma~\ref{lemma:sum geometrical}, $\sum_{j=j_0}^{J-1} G_j$ is stochastically dominated by a random variable with distribution $\Geom(\eps)$.
	Therefore, for any integer $r>0$ we have
	$$\Pr\left[\sum_{j=j_0}^{J-1} G_j > \tfrac{r}{j_0+1}\right] \le (1-\eps)^{r/(j_0+1)}.$$
	Similarly, for any $j < j_0$ we have $\Pr[G_j > \tfrac{r}{j_0+1}] \le (1-\eps)^{r/(j_0+1)}$.
	We conclude,
	$$\Pr[G>r] \le (j_0+1) \cdot (1-\eps)^{r/(j_0+1)}.$$
\end{proof}

Finally, the following lemma will be used to argue that in the Erd\H{o}s-R\'enyi graph with $n$ vertices and edge probability $\tfrac an$, $a >0$ a constant, the maximum vertex degree at most $O(\log n)$ with high probability. This follows immediately from a simple Chernoff bound argument (as would the sharp $O(\log(n) / \log\log(n))$ bound, which we do not need).
\begin{lemma} \label{lemma: log degree}
	For any $a > 0$ there exists $c > 0$ such that
	$\Pr\left[\Bin\left(n,\tfrac an\right) \ge c \log n\right] \le \tfrac1n$.
\end{lemma}

\subsection{First Order Bounds}


\begin{lemma}\label{lem:prelim:(1-1/n)^n}
    For any $n > 0$ we have $\tfrac1e - \tfrac1{en} \le \left(1-\tfrac1n\right)^n \le \tfrac1e$.
\end{lemma}
\begin{lemma}\label{lem:prelim:(1-1/n)^k}
    For any $k < n$ we have
    $1-\tfrac{k}{n} \le \left(1-\tfrac1n\right)^k \le 1-\tfrac{k}{n} + \tfrac{k^2}{2n^2}$.
\end{lemma}
Replacing $n$ by $n/p$ for some $p > 0$ we get the following.
\begin{corollary}\label{cor:prelim:(1-p/n)^k}
    For any $k < n/p$ we have
    $1-p\tfrac{k}{n} \le \left(1-\tfrac{p}{n}\right)^k \le 1-p\tfrac{k}{n} + \tfrac{p^2k^2}{2n^2}$.
\end{corollary}
\begin{lemma}\label{lem:prelim:1/(1-x)}
    For any $0 \le x < 1$ we have $\tfrac1{1-x} \ge 1+x$.\\
    For any $0 \le x \le \tfrac12$ we have $\tfrac1{1-x} \le 1+2x$.
\end{lemma}
Combining the three lemmas above we obtain the following corollary.
\begin{corollary}\label{cor:prelim:(1-1/n)^(n-u)}
    For any $u < n$ we have
    $\tfrac1e \le \left(1-\tfrac1n\right)^{n-u} \le \tfrac1e + \tfrac{2u}{en}$.
\end{corollary}

Again, we have the similar estimates for some $p > 0$.
\begin{corollary}\label{cor:prelim:(1-p/n)^(n-u)}
    For any $u < n/p$ we have
    $e^{-p} \le \left(1-\tfrac pn\right)^{n-u} \le e^{-p}\left(1 + 2p\tfrac{u}{n}\right)$.
\end{corollary}

%% file: exp_growth_new.tex
\subsection{Exponential Growth Regime. Upper Bound}
\label{section:exp-growth-upper}

In this section and the following, we analyze the runtime of a homogeneous rumor spreading process in the regime where the number of informed nodes roughly grows by a constant factor until a linear number $fn$ of nodes is informed. Not surprisingly, this implies that the process takes a logarithmic time to inform a linear number of nodes.

The challenge in the following analysis, which was also faced by previous works, is that in most rumor spreading processes the dissemination speed reduces when more nodes are informed. So it is not true that for all $k \in [1,fn]$, a round starting with $k$ informed nodes ends with an expected number of $k + \gamma k$ nodes, where $\gamma$ is some constant, but rather that we only expect $E_k = \gamma k (1 - \Theta(k/n))$ newly informed nodes. This non-linearity also implies that a round starting with an \emph{expected} number of $k$ nodes does not end with an expected number of $k + E_k$ informed nodes, but less. So we also need to argue that the number of newly informed nodes a round ends with is strongly concentrated around its expectation, and that thus, we can assume that with sufficiently high probability we end up not too far below the expectation (which gives another small loss over the idealized multiplicative increase of the number of informed nodes).

We overcome these difficulties as follows. (i) We formulate an \emph{exponential growth condition} that is satisfied by essentially all homogeneous rumor spreading processes showing an exponential growth regime. The key observation, which allows us to treat many protocols with this single analysis is that it is not necessary that the actions of the nodes show particular independences. It suffices that a relatively mild covariance condition is satisfied. (ii) We then use (throughout the whole regime from the first informed node to a linear number of informed nodes) a simple phase-target argument. (a) We define for each number $k$ of initially informed nodes a \emph{round target} $E_0(k)$ such that a round starting with $k$ informed nodes with (sufficiently high) probability $1-q_k$ ends with $E_0(k)$ informed nodes. Hence the expected time to go from $k$ to $E_0(k)$ or more informed nodes is $t_k = 1 + \frac{q_k}{1-q_k}$. (b) From this, we define a sequence of target $k_0 = 1, k_1 = E_0(k_0), k_2 = E_0(k_1), \dots, k_J = \Theta(n)$ and argue that the time to reach $k_J$ informed nodes is just the sum of the expected times $t_{k_j}$. By defining the round targets in a suitable manner, we ensure that $J = \log_{1+\gamma}(n) + \Theta(1)$ and that the sum of the $t_{k_j}$ is $J + \Theta(1)$. We note that the phase-target argument was also used in~\cite{DoerrK14}, there however only for the push-protocol and only in the regime from $n^s$, $s$ a small constant, to $\Theta(n)$ informed nodes. Consequently, due to the large number of active nodes acting independently, the phase failure probabilities where ignorable small.

In principle, all the arguments outlined above are very elementary and use nothing more advanced than expectations and Chebyshev's inequality. Hence the main technical progress of this work is formulating an exponential growth condition (including the covariance condition) that allows these elementary arguments in a way that the deviations from the idealized ``multiply-by-$\gamma$'' world in the end all disappear in the $\Theta(1)$ term of the dissemination time. These technicalities also appear in some of the following calculations, which therefore, while all not difficult, are at times slightly lengthy. Since arguments similar to the ones in this section are used throughout this work, we give all details in this section and will be more brief in the following ones.

We start in this section with proving an upper bound for the runtime given that we have suitable lower bounds for the probability that an uninformed node becomes informed. In the following section, we prove a lower bound for the runtime given that we have suitable upper bounds on the speed of the progress. These bounds will match apart from additive constants if the growth factor $\gamma$ is identical.

\subsubsection{Exponential Growth Conditions}

Throughout this section, we assume that we regard a homogeneous epidemic protocol which satisfies the following \emph{upper exponential growth conditions} including a covariance condition.



\begin{defn} [upper exponential growth conditions]\label{def:upper-exp-growth-conditions}
    Let $\gamma_n$ be bounded between two positive constants.
    Let $a, b, c \ge 0$ and $0 < f < 1$ with $af<1$.
    We say that a homogeneous epidemic protocol satisfies the \emph{upper exponential growth conditions} in $[1,fn[$ if for any $n \in \N$ big enough the following properties are satisfied for any $k < fn$.
	\renewcommand{\theenumi}{(\roman{enumi})}%
    \begin{enumerate}
        \item $\Pk \ge \gamma_n \tfrac{k}{n} \cdot \left(1- a\tfrac{k}{n} - \tfrac{b}{\ln n}\right)$.
        \item $\Ck \le c\tfrac{k}{n^2}$.
    \end{enumerate}
\end{defn}

The main result of this section is that the upper exponential growth conditions imply that the number of informed nodes multiplies by, essentially, $1+\gamma_n$ in each round, and that the expected number of rounds until $fn$ nodes are informed, is at most $\log_{1+\gamma_n}n + O(1)$.

\begin{theorem}[upper bound for the spreading time]\label{th:exp-growth-upper}
    Consider a homogeneous epidemic protocol satisfying the upper exponential growth conditions in $[1,fn[$.
    Then there exist constant $A', \alpha'$ such that
    \begin{eqnarray*}
    	&& \Expect[T(1, fn)] \le \log_{1+\gamma_n} n + O(1), \\
		&& \Pr[T(1,fn) > \log_{1+\gamma_n} n + r] \le A' e^{-\alpha' r} \, \mbox{for any $r \in N$}.
    \end{eqnarray*}
\end{theorem}

\subsubsection{Round Targets and Failure Probabilities}\label{subsection:exp-growth-upper - round targets}
Let us introduce the random variable $X(k)$ being equal to the number of newly informed nodes in a round having $k$ informed nodes at the beginning.
Since $\Expect[X(k)] = p_k (n-k)$, the exponential growth conditions imply $\Expect[X(k)] \ge E(k)$, where
\begin{equation}
	E(k) := \gamma_n k \left(1 - (a+1)\tfrac{k}{n} - \tfrac{b}{\ln n}\right). \notag
\end{equation}

Using Chebyshev's inequality we can show that the value of $X(k)$ is concentrated around its expected value.
Lemma~\ref{lem:exp-growth-failure} hence claims that with good probability, $X(k)$ attains at least the \emph{target value}
\begin{equation}
    E_0(k) := E(k) - Ak^B, \label{eq:def-E0-upper}
\end{equation}
where $A > 0$ and $B \in ]0.5, 1[$ are some constants chosen uniformly for all values of $k$ and $n$.
There are no special conditions on $B$, so we suppose that $B$ is fixed from now on, e.g., to 3/4.
We will, in the following, choose $A$ small enough to ensure that the $-Ak^B$ term has a sufficiently small influence on the general bevahior of  $E_0(k)$.

\begin{lemma}\label{lem:exp-growth-E0-increase}
There exist $f' > 0$ and $A'>0$ such that for $n$ big enough, the following conditions are satisfied.
\begin{itemize}
	\item $E(\cdot)$ is increasing up to $f'n$, that is, for all $i < j \le f'n$ we have $E(i) < E(j)$;
	\item When $A$ in equation~\eqref{eq:def-E0-upper} satisfies $0 < A < A'$, then also $E_0(\cdot)$ is increasing up to $f'n$;
	\item $E_0(k) > 0$ for all $k \in [1,f'n[$.
\end{itemize}
\end{lemma}
\begin{proof}
	The first claim follows from the second, so let us regard the derivative of $E_0(k)$,
    \[
    	E_0'(k) = \gamma_n - 2\gamma_n(a+1)\tfrac{k}{n} - \gamma_n\tfrac{b}{\ln n} - ABk^{-1+B}.
    \]
    We see that, for any $f' < \tfrac1{2(a+1)}$, any $A > 0$ small enough, and any $n$ large enough, $E_0'(k)$ is positive for all $k \in [1, f'n[$. Therefore, to satisfy the first two parts of the claim, we pick any $f' \in ]0, \tfrac1{2(a+1)}[$ and then any $A' < \tfrac1B \gamma_n(1-2(a+1)f')$.

    To show that $E_0(k) > 0$ for all $k \in [1,f'n[$, it suffices to check this for $k=1$.
    By possibly lowering $A'$ further, we obtain for $n$ large enough that
    \begin{equation}
    	E_0(1) = \gamma_n \left(1 - \tfrac{a+1}{n} - \tfrac{b}{\ln n}\right) - A > 0 \notag.
    \end{equation}
\end{proof}
We assume in the following that $f$ in Definition~\ref{def:upper-exp-growth-conditions} satisfies $f < f'$ and that $A$ in~\eqref{eq:def-E0-upper} was chosen in $]0,A'[$.

\begin{lemma}\label{lem:exp-growth-failure}
    For any $k < fn$,
    \begin{equation}
        \Pr[X(k) \le E_0(k)] \le \min\left\{q(k), \tfrac1{1+1/q(1)}\right\} \notag,
    \end{equation}
    where $q(k) := \tfrac{\gamma_n+c}{A^2} \cdot k^{-2B+1}$.
\end{lemma}
\begin{proof}
    By the exponential growth conditions, $\Expect[X(k)] \ge E(k)$.
    Applying Chebyshev's inequality, we compute
    \begin{align*}
        & \Pr[X(k) \le E_0(k)]
    	   = \Pr\left[X(k) \le E(k) \cdot \left(1-\tfrac{Ak^B}{E(k)}\right)\right] \\
        & \le \Pr\left[ X(k) \le \Expect[X(k)]\cdot\left(1-\tfrac{Ak^B}{E(k)}\right) \right] \\
    	& = \Pr\left[ X(k) \le \Expect[X(k)] - Ak^B \cdot \tfrac{\Expect[X(k)]}{E(k)} \right] \\
        & \le \tfrac{\Var[X]}{(Ak^B)^2} \cdot \tfrac{E(k)^2}{\Expect[X(k)]^2}.
    \end{align*}
    From the covariance condition, it follows that $\Var[X(k)] \le \Expect[X(k)] + ck$.
	Using $E(k) / \Expect[X(k)] \le 1$ once, we obtain
    \begin{align*}
        \Pr[X(k) \le E_0(k)]
        & \le \left(1 + \tfrac{ck}{\Expect[X(k)]}\right) \cdot \tfrac{E(k)}{\Expect[X(k)]}
    	   		\cdot \tfrac{E(k)}{A^2k^{2B}} \\
    	& \le \left(1 + \tfrac{ck}{E(k)}\right) \cdot \tfrac{E(k)}{A^2k^{2B}} \\
        & = \tfrac{E(k)+ck}{A^2k^{2B}} \le \tfrac{\gamma_nk + ck}{A^2k^{2B}}.
    \end{align*}
    One can see that for small values of $k$, $q(k)$ might be more than one.
	To avoid such a trivial bound for the failure probability, it suffices to replace Chebyshev's inequality in the proof by the Cantelli's inequality (see Lemma~\ref{lem:prelim:Cantelli}) and bound the probability by $\tfrac1{1+1/q(k)}$.
	To finish the proof we note that $q(k)$ is decreasing in $k$, so
	$\Pr[X(k) \le E_0(k)] \le \min\left\{q(k), \tfrac1{1+1/q(1)}\right\}$.
\end{proof}


\subsubsection{The Phase Calculus}\label{subsection:exp-growth-upper - phases}

Having just defined round targets for all numbers $k$ of initially informed nodes and the probabilities that these targets are not achieved within a round, we now proceed to define the sequence $k_j$ of round targets which we aim at satisfying one after the other, ideally within one round per target.
%
%
%

We define recursively
\begin{equation}
    k_0 = 1, \quad k_{j+1} := k_j + E_0(k_j) \notag.
\end{equation}

\begin{lemma}\label{lem:exp-growth-expk-upper}
	After possibly lowering $A'$ from Lemma~\ref{lem:exp-growth-E0-increase}, there exist $\alpha > 0$ and $J = \log_{1+\gamma_n} n + O(1)$ such that
	\begin{equation}
		fn > k_j \ge \alpha(1+\gamma_n)^j \notag,
	\end{equation}
	for all $j \le J$.
	In particuar, $k_J = \Theta(n)$.
\end{lemma}
\begin{proof}
	By definition of $k_j$,
	\begin{equation}
		k_j = k_{j-1} + E_0(k_{j-1})
		= k_{j-1}
			\left(1 + \gamma_n - \gamma_n(a+1)\tfrac{k_{j-1}}{n}-\gamma_n\tfrac{b}{\ln n} - Ak_{j-1}^{-1+B}\right)\notag.
	\end{equation}
	Let $\Gamma_n := 1+\gamma_n-\gamma_n\tfrac{b}{\ln n}$.
	Then,
	\begin{equation}
		k_j = \Gamma_n k_{j-1} \left(1 - \gamma_n\tfrac{a+1}{\Gamma_n}\cdot\tfrac{k_{j-1}}{n}
			- \tfrac{A}{\Gamma_n}\cdot k_{j-1}^{-1+B}\right) \notag.
	\end{equation}
	Clearly, $\Gamma_n \ge (1+\gamma_n)(1-\tfrac{b}{\ln n})$.
	By our assumption on $\gamma_n$, $\Gamma_n$ is bounded from above by a constant and is at least $1+\gamma_n/2$ for $n$ big enough.
	Let hence $\tilde{a} := \gamma_n\tfrac{a+1}{1+\gamma_n/2}$ and $\tilde{A} := \tfrac{A}{1+\gamma_n/2}$.
	Then, for any big $n$,
	\begin{equation}
		k_j \ge (1+\gamma_n) \left(1 -
			\tfrac{b}{\ln n}\right)
			k_{j-1} \left(1 - \tilde{a}\tfrac{k_j-1}{n} - \tilde{A}k_{j-1}^{-1+B}\right) \notag.
	\end{equation}
    We assume that $A$ (resp. $\tilde{A}$) and $f$ are small enough such that the expression in the brackets is positive.
    Since $k_0 = 1$, by induction we obtain for all $j$ that
	\begin{equation}
        k_j \ge (1+\gamma_n)^j (1-\tfrac{b}{\ln n})^j
        	\prod_{i=0}^{j-1}\left(1-\tilde{a}\tfrac{k_i}{n} - \tilde{A}k_i^{-1+B}\right) \notag.
    \end{equation}
    By choosing $f$ and $A$ small enough, we can assume that $k_i > 0$ for all $i < j$.
    \begin{equation}
        k_j \ge (1+\gamma_n)^j (1-\tfrac{b}{\ln n})^j
        	\left(1 - \tilde{a}\sum_{i=0}^{j-1} \tfrac{k_i}{n}
        	- \tilde{A}\sum_{i=0}^{j-1} k_i^{-1+B}\right) \notag.
    \end{equation}


    Let $J := \log_{1+\gamma_n}(fn) - \Delta r$ for some positive $\Delta r = O(1)$ determined later.
    For $j \le J$ we have $k_j \le (1+\gamma_n)^j$ by construction, and thus $k_j \le fn$.
    Also we have $(1-\tfrac{b}{\ln n})^j = \Theta(1)$.
    In particular this term is at least $2\alpha$ for some $\alpha > 0$ and all $n$ big enough.

    We show by induction on $j$ that $k_j \ge \alpha(1+\gamma_n)^j$ for all $j \le J$.
    The base for $j=0$ and $k_0=1$ is obvious.
    Let $1 \le j \le J$  and let $k_i \ge \alpha(1+\gamma_n)^i$ for all $i<j$.
    By construction, we have $k_i \le (1+\gamma_n)^i$.
    Therefore,
    \begin{align}
    	k_j & \ge 2\alpha (1+\gamma_n)^j
    		\left(1-\tfrac{\tilde{a}}{n}\sum_{i=0}^{j-1}(1+\gamma_n)^i
    		- \tilde{A}\alpha^{-1+B}\sum_{i=0}^{j-1} (1+\gamma_n)^{i(-1+B)}\right) \notag \\
    	& \ge 2\alpha (1+\gamma_n)^j
    		\left(1 - \tilde{a}\cdot\tfrac{(1+\gamma_n)^{-\Delta r}}{\gamma_n}
    		- \tilde{A}\alpha^{-1+B}\cdot\tfrac1{1-(1+\gamma_n)^{B-1}}\right) \notag.
    \end{align}
    By choosing $\Delta r$ large enough and $\tilde{A}$ (resp. $A$) small enough, we can bound the last two expressions by $1/4$, and obtain
    $$k_j \ge 2\alpha(1+\gamma_n)^j (1-1/4-1/4) = \alpha(1+\gamma_n)^j.$$
\end{proof}

By Lemma~\ref{lem:exp-growth-E0-increase}, the $k_j$ form a non-decreasing sequence.
We say that our homogeneous rumor spreading process is in phase $j$ for $j \in \{0,\ldots,J-1\}$, if the number of informed nodes is in $[k_j, k_{j+1}[$.


\begin{lemma}\label{lem:exp-growth-ETj-upper}
	If our process is in phase $j<J$, then the number of rounds to leave phase $j$ is stochastically dominated by $1 + \Geom(1-Q_j)$, where $Q_j := \min\left\{q(k_j), \tfrac1{1+1/q(1)}\right\}$.
\end{lemma}
\begin{proof}
	By Lemma~\ref{lem:exp-growth-E0-increase} we have $k + E_0(k) \ge k_j + E_0(k_j) = k_{j+1}$ for any $k_j \le k < fn$.
	By Lemma~\ref{lem:exp-growth-failure},
	\begin{equation}
		\Pr[k+X(k) \le k_{j+1}]
		< \Pr[k+X(k) \le k+E_0(k)]
		< \min\left\{q(k), \tfrac1{1+1/q(1)}\right\} \notag.
	\end{equation}
	Since $q(k)$ is decreasing,
	\begin{equation}
    	\max_{k_{j+1} > k \ge k_j} \Pr[k + X(k) < k_{j+1}] \le Q_j \notag,
	\end{equation}
	and this is an upper bound for the probability to stay in phase $j$ for one round.
	We can thus bound the number of rounds taken to leave phase $j$ by a random variable with geometric distribution $\Geom(1-Q_j)$.
\end{proof}

\begin{lemma}\label{lem:exp-growth-sum-Qj}
	$\sum_{j=0}^{J-1} Q_j = O(1)$.
\end{lemma}
\begin{proof}
    We apply the estimate for $q(k_j)$ from Lemma~\ref{lem:exp-growth-failure} and the bounds for $k_j$ from Lemma~\ref{lem:exp-growth-expk-upper}.
    Therefore,
    \begin{align}
        \sum_{j=0}^{J-1} Q_j
        & \le \sum_{j=0}^{J-1} q(k_j)
        	\le \tfrac{\gamma_n+c}{A^2} \cdot \sum_{j=0}^{J-1} k_j^{-2B+1} \notag \\
        & \le \tfrac{\gamma_n+c}{A^2} \cdot \alpha^{-2B+1} \cdot \sum_{j=0}^{J-1} (1+\gamma_n)^{j(-2B+1)} \notag.
    \end{align}
    The last sum is a decreasing geometric series as $B > 0.5$.
    So, $\sum_j Q_j = O(1)$.
\end{proof}

Now we can prove the main result of this section.

\begin{proof}[Proof of Theorem~\ref{th:exp-growth-upper}]
	By Lemma~\ref{lem:exp-growth-expk-upper}, there exists $J = \log_{1+\gamma_n} n + O(1)$ such that $k_J = \Theta(n)$.
	In the following we assume that $J \le \log_{1+\gamma_n} + \tau$ for some constant $\tau$.
	The phase method allows us to bound the number of rounds until at least $k_J$ nodes are informed.
    We denote by the random variable $T_j$ the number of rounds spent in the $j$th phase.
    By Lemma~\ref{lem:exp-growth-ETj-upper}, $T_j$ is stochastically dominated by $1 + \Geom(1-Q_j)$.
    With Lemma~\ref{lem:exp-growth-sum-Qj}, we compute
    \begin{align}
        \Expect[T(1,k_J)]
        & \le \sum_{j=0}^{J-1} \Expect[T_j] \le \sum_{j=0}^{J-1} (1+\tfrac{Q_j}{1-Q_j}) \notag \\
        & = J + \sum_{j=0}^{J-1} \tfrac{Q_j}{1-Q_j} \le J + \tfrac1{1-Q_0}\sum_{j=0}^{J-1} Q_j \notag \\
        & = J + O(1). \notag
    \end{align}
    Since $Q_j$ is bounded by a geometric sequence, Lemma~\ref{lemma:sum geometrical-2} claims that there exist $A'_1, \alpha'_1$ such that $$\Pr[T(1,k_J) > J + r/2] \le A'_1e^{-\alpha'_1r}.$$
    If $k_J < fn$, then we observe that for all $k \in [k_J, fn[$, $p_k$ satisfies the conditions of Lemma~\ref{lem:general-connect}. Therefore, $T(k_J, fn) = O(1)$ and there exist $A'_2, \alpha'_2$ such that $\Pr[T(k_J, fn) > r/2] \le A'_2e^{-\alpha'_2r}$.
    Combining bounds for $T(1,k_J)$ and $T(k_J,fn)$ we obtain the following.
    \begin{eqnarray*}
    	&& \Expect[T(1,fn)] \le \Expect[T(1,k_J)] + \Expect[T(k_J,fn)]
    			\le \log_{1+\gamma_n} n + O(1), \\
    	&& \Pr[T(1,fn) > \log_{1+\gamma_n}n + r] \le A' e^{-\alpha' r},\,
    	\mbox{where $A' := (A'_1+A'_2)e^{\alpha' \tau}$ and $\alpha' := \min\{\alpha'_1,\alpha'_2\}$}.
    \end{eqnarray*}
\end{proof}

\subsection{Exponential Growth Regime. Lower Bound}

In this section, we prove a lower bound for an exponential growth regime. We formulate a condition matching the upper bound condition and show that this leads to a lower bound on the rumor spreading time that matches the upper bound apart from a constant number of rounds. We use again the target-phase method.

This is the first time that the target-phase argument is used to prove a lower bound. In the work closest to ours,~\cite{DoerrK14}, only the classic push protocol was regarded. Consequently, there, the simple argument that the number of nodes can at most double each round was sufficient to obtain a lower bound for the growth regime. Such an argument, e.g., is not possible for the classic pull protocol.

The main difference to the upper bound proof lies in the final argument. In the upper bound proof, the failure to reach a round target simply resulted in that we had to try again to reach this target. For the lower bound, a failure is that the process gains more than one phase in one round, resulting in that the time usually spent in these now skipped phases is spared. Arguing that the total time spared by such events is only $O(1)$ needs a slightly more complicated book-keeping of the failure events and a slightly more complicated final argument.

\subsubsection{Exponential Growth Conditions}

We formulate the lower exponential growth condition in an analoguous way as the upper one. In particular, the covariance condition is identical.

\begin{defn} [lower exponential growth conditions] \label{def:lower-exp-growth-conditions}
    Let $\gamma_n$ be bounded between two positive constants and let $a, b, c \ge 0$ and $0 < f < 1$.
    We say that a homogeneous epidemic protocol satisfies the \emph{lower exponential growth conditions} in $[1,fn[$ if for any $n \in \N$ big enough, the following properties are satisfied for any $k < fn$.
	\renewcommand{\theenumi}{(\roman{enumi})}%
    \begin{enumerate}
        \item $p_k \le \gamma_n \tfrac{k}{n} \cdot \left(1 + a\tfrac{k}{n} + \tfrac{b}{\ln n}\right)$.
        \item $c_k \le c\tfrac{k}{n^2}$.
    \end{enumerate}
\end{defn}

These conditions imply the following lower bounds on the rumor spreading time.
\begin{theorem}
\label{th:exp-growth-lower}
	Consider a homogeneous epidemic protocol satisfying the lower exponential growth conditions in $[1,fn[$. Then there are constant $A', \alpha'>0$ such that 
	\begin{eqnarray*}
    &&\Expect[T(1, fn)] \ge \log_{1+\gamma_n} n - O(1),\\
    &&\Pr[T(1, fn) \le \log_{1+\gamma_n} n - r] \le A' \exp(-\alpha'r)\, \mbox{ for all $r \in \N$}.
  \end{eqnarray*}
  In addition there exists $f' \in ]f,1[$ such that with probability $1-O\left(\tfrac1n\right)$ there are at most $f'n$ informed nodes after $T(1,fn)$ rounds.
\end{theorem}

\subsubsection{Round Targets and Failure Probabilities}

As above, we consider a round with $k$ informed nodes initially.
We define $X(k)$ to be the number of newly informed nodes in this round.
Since  $\Expect[X(k)] = \Pk (n-k)$, the exponential growth conditions give $\Expect[X(k)] \le E(k)$ with
\begin{equation}
	E(k) := \gamma_n k \left(1 + a\tfrac{k}{n} + \tfrac{b}{\ln n}\right). \notag
\end{equation}
Note that we could replace the $a$ above by $a-1$, giving an expression closer resembling the corresponding one from the previous section. Since all these constants do not matter, we preferred the simpler version without the extra~$-1$.

Like in the previous section we introduce
\begin{equation}
	E_0(k) := E(k) + Ak^B, \label{eq:expgrowth-def-E0-lower}
\end{equation}
where $A > 0$ and $B \in ]0.5, 1[$ are some constants chosen uniformly for all values of $k$ and $n$.
Unlike in Section~\ref{section:exp-growth-upper}, it is obvious that $E(k)$ and $E_0(k)$ are increasing.

Note that we can freely replace $f$ in the definition of the lower exponential growth conditions by a smaller constant $f'$, since showing $\E[T(1,f'n)] \ge \log_{1+\gamma_n}(n) - O(1)$ in Theorem~\ref{th:exp-growth-lower} would immediately imply $\E[T(1,fn)] \ge \log_{1+\gamma_n}(n) - O(1)$. Consequently, let us assume that $f$ is small enough such that for any $n$ sufficiently large and $k<fn$,
\begin{equation}
    E(k) \le 2\gamma_n k. \label{eq:exp-growth-lower-42}
\end{equation}

The following lemma will later be used to argue that an unexpectedly fast progress is unlikely. Different from the upper bound analysis in the previous section, we now need a failure probability for different excessive progresses (quantified by the parameter $h$ below).

\begin{lemma}\label{lem:exp-growth-failure-lower}
    For any $k < fn$ and $h = 0, 1, 2, \ldots$,
    \begin{equation}
    	\Pr[X(k) \ge E(k) + Ak^B(1+\gamma_n)^h]
    	\le q_h(k) := \tfrac{2\gamma_n + c}{A^2} \cdot \tfrac{k^{-2B+1}}{(1+\gamma_n)^{2h}} \notag.
    \end{equation}
\end{lemma}
\begin{proof}
    By the exponential growth conditions, $\Expect[X(k)] \le E(k)$.
    By the covariance condition and~\eqref{eq:exp-growth-lower-42},
    \begin{equation}
        \Var[X(k)] \le E(k) + n^2 c_k \le k(2\gamma_n+c) \notag.
    \end{equation}
    Applying Chebyshev's inequality, we obtain
    \begin{align}
    	\Pr&[X(k) \ge E(k) + Ak^B(1+\gamma_n)^h] \notag \\
    	& \le \Pr[X(k) \ge \Expect[X(k)] + Ak^B(1+\gamma_n)^h] \notag \\
    	& \le \tfrac{\Var[X(k)]}{(Ak^B)^2(1+\gamma_n)^{2h}} \notag \\
    	& \le \tfrac{2\gamma_n+c}{A^2} \cdot k^{-2B+1} \cdot \tfrac1{(1+\gamma_n)^{2h}} \notag.
    \end{align}
\end{proof}




\subsubsection{The Phase Calculus}
Like in Section~\ref{section:exp-growth-upper}, we define the sequence $k_j$ recursively by
\begin{equation}
	k_0 = 1, \quad k_{j+1} := k_j + E_0(k_j) \notag,
\end{equation}
and obtain the following exponential growth behavior.

\begin{lemma}\label{lem:exp-growth-expk-lower}
	By taking $A$ small enough in~\eqref{eq:expgrowth-def-E0-lower}, there exist $\alpha > 0$ and $J = \log_{1+\gamma_n}n - O(1)$ such that for all $j < J$
	\begin{equation}
		(1+\gamma_n)^j \le k_j \le \alpha(1+\gamma_n)^j \text{ and } \; \text k_j < fn \notag.
	\end{equation}
\end{lemma}
\begin{proof}
    Note that 
    $k_j \ge (1+\gamma_n)^j$ is immediate from the definitions and a simple induction.
    So it remains to show the upper bound on the $k_j$.
    Clearly, by definition of $k_j$,
    \begin{equation}
        k_j \le (1+\gamma_n) (1+\tfrac{b}{\ln n}) k_{j-1} \left(1 + a\tfrac{k_{j-1}}{n}\right)
        	\left(1 + Ak_{j-1}^{-1+B}\right) \notag.
	\end{equation}
    Since $k_0 = 1$, by induction we obtain
	\begin{equation}
        k_j \le (1+\gamma_n)^j (1+\tfrac{b}{\ln n})^j \prod_{i=0}^{j-1}\left(1 + a\tfrac{k_i}{n}\right)
        	\prod_{i=0}^{j-1}\left(1 + Ak_i^{-1+B}\right). \notag 
	\end{equation}
	Let $J := \log_{1+\gamma_n}n - \Delta r$ for some $\Delta r = O(1)$ determined later.
    If $j < J$, then $(1-\tfrac{b}{\ln n})^j = \Theta(1)$.
    In particular, it is at most $\tfrac\alpha4$ for some $\alpha > 0$ and any $n$ big enough.
	By the fact that $1+x \le e^x$ for any $x > 0$, we have
	\begin{equation}
		k_j \le \tfrac\alpha4 (1+\gamma_n)^j
		\exp\left(\sum_{i=0}^{j-1}a\tfrac{k_i}{n}\right)
		\cdot \exp\left(\sum_{i=0}^{j-1}Ak_i^{-1+B}\right) \label{eq:exp-growth-upper-eq1}.
	\end{equation}
	We prove the claim of lemma by induction on $j$. Assume that for some $j < J$ we have $k_i \le \alpha(1+\gamma_n)^i$ for any $i < j$.
	Since $k_i \ge (1+\gamma_n)^i$ for all $i$, both sums in~\eqref{eq:exp-growth-upper-eq1} can be bounded by geometric series.
	Therefore,
	\begin{equation}
		k_j \le \tfrac\alpha4 (1+\gamma_n)^j
		\exp\left(\sum_{i=0}^{j-1}\tfrac{a}{n} \cdot \alpha(1+\gamma_n)^i\right)
		\cdot \exp\left(\sum_{i=0}^{j-1}A(1+\gamma_n)^{i(-1+B)}\right) \notag.
	\end{equation}
	Since $j < J$, by choosing $\Delta r$ large enough and $A$ small enough, we can bound both sums by any positive constant, in particular by $\ln 2$.
	Therefore, for any $j < J$,
	\begin{equation}
		k_j \le \tfrac\alpha4(1+\gamma_n)^j \exp(\ln2) \cdot \exp(\ln 2) = \alpha(1+\gamma_n)^j\notag.
	\end{equation}
\end{proof}

By definition, the $k_j$ form a non-decreasing sequence.
Like in Section~\ref{section:exp-growth-upper}, we say that the rumor spreading process is in phase $j$ for $j = 0, \ldots, J-1$, if the number of informed nodes is in $[k_j, k_{j+1}[$.

\begin{lemma}\label{lem:exp-growth-jump-prob}
	Let $h \ge 2$.
	If the process is in phase $j < J$ at the beginning of one round, then the probability that the number of informed nodes is at least $k_{j+h}$ at the end of the round, is at most $q_{h-2}(k_j)$.
\end{lemma}
\begin{proof}
	For $1 \le k \le k_{j+1}$, we have $$k + E(k) + Ak^B \le k_{j+1} + E(k_{j+1}) + Ak_{j+1}^B = k_{j+2}.$$
	Since $k_{j+h} \ge (1+\gamma_n)^{h-2} k_{j+2}$, we have
	\begin{equation}
		k_{j+h}
		\ge (1+\gamma_n)^{h-2} \left(E(k) + Ak^B + k\right)
		\ge k + E(k) + Ak^B(1+\gamma_n)^{h-2} \notag.
	\end{equation}

    By Lemma~\ref{lem:exp-growth-failure-lower}, the maximum probability to have at least $k_{j+h}$ informed nodes at the end of the round is
    \begin{align}
        & \max_{k\in [k_j, k_{j+1}[} \Pr[k+X(k) \ge k_{j+h}] \notag \\
    	& \le \max_{k\in [k_j,k_{j+1}[} \Pr[k+X(k)\ge k+E(k)+Ak^B(1+\gamma_n)^{h-2}] \notag \\
    	& \le \max_{k\in [k_j,k_{j+1}[} q_{h-2}(k) \le q_{h-2}(k_j) \notag.
    \end{align}
    The last inequality follows from the fact that since $B > 1/2$, $q_{h-2}(\cdot)$ decreases.
\end{proof}


With Lemma~\ref{lem:exp-growth-expk-lower}~and~\ref{lem:exp-growth-jump-prob}, we can now prove Theorem~\ref{th:exp-growth-upper}.

\begin{proof}[Proof of Theorem~\ref{th:exp-growth-lower}]
	Let $S$ be the set of visited phases, e.g., if the process does not jump over any phase, then $S = \{0,\ldots,J-1\}$.
	By $\tau_j$ we denote the number of rounds spent in the $j$th phase.
	So the spreading time $T(k_0, k_J) = \sum_{j\in S} \tau_j$.
	We do not know the size of $S$, so in order to bound the spreading time below, let us introduce the random variable $\Delta_j$ which is equal to the length of the jump from the $j$th phase when the process leaves it.
	Let also $d_j := \Delta_j - \tau_j$.
	Since $\sum_{j\in S} \Delta_j = J$, we have $T(k_0, k_J) = J - \sum_{j\in S} d_j$.
	By definition, for $j \in S$ and $h > 0$, we have $\Pr[d_j \ge h] \le \Pr[\Delta_j \ge h+1]$.
	Then, by Lemma~\ref{lem:exp-growth-failure-lower}~and~\ref{lem:exp-growth-jump-prob},
	\begin{equation*}
		\Pr[d_j \ge h]
		\le q_{h-1}(k_j)
		\le \tfrac{2\gamma_n+c}{A^2} \tfrac{k_j^{-2B+1}}{(1+\gamma_n)^{2h-2}}.
	\end{equation*}
  The above argument shows that $T(k_0, k_J)$ stochastically dominates $J - D$, where $D$ is the sum of independent non-negative integer random variables $D = \sum_{j=0}^{J-1} D_j$ satisfying $\Pr[D_j \ge h] \le \tfrac{2\gamma_n+c}{A^2} \tfrac{k_j^{-2B+1}}{(1+\gamma_n)^{2h-2}}$ for all $h \ge 1$. Let $R_h := \{(r_0, \dots, r_{J-1}) \in \Z^J_{\ge 0} \mid \sum_{j=0}^{J-1} r_i = h\}$ for all $h \ge 1$. We compute
  \begin{align*}
  \Pr[D \ge h] & \le \sum_{r \in R_h} \prod_{j = 0}^{J-1} \Pr[D_j \ge r_j] \\
  &\le (1+\gamma_n)^{-2h} \sum_{r \in R_h} \prod_{j \in [0..J-1], r_j >0} \tfrac{2\gamma_n+c}{A^2} \tfrac{k_j^{-2B+1}}{(1+\gamma_n)^{-2}}\\
  &\le (1+\gamma_n)^{-2h} \sum_{M \subseteq [0..J-1]} \prod_{j \in M} \tfrac{2\gamma_n+c}{A^2} \tfrac{k_j^{-2B+1}}{(1+\gamma_n)^{-2}}\\
  &\le (1+\gamma_n)^{-2h} \prod_{j \in [0..J-1]} \bigg(1+\tfrac{2\gamma_n+c}{A^2} \tfrac{k_j^{-2B+1}}{(1+\gamma_n)^{-2}}\bigg)\\
  &\le (1+\gamma_n)^{-2h} \exp\bigg(\sum_{j \in [0..J-1]} \tfrac{2\gamma_n+c}{A^2} \tfrac{k_j^{-2B+1}}{(1+\gamma_n)^{-2}}\bigg)\\
  &\le (1+\gamma_n)^{-2h} O(1),
\end{align*}
	where the last estimate uses  Lemma~\ref{lem:exp-growth-expk-lower}. This proves that tail bound statement. For the claim on the expected rumor spreading time, we compute
	\begin{align*}
		\Expect[D] \le \sum_{h \ge 1} \Pr[D \ge h] \le \sum_{h \ge 1} (1+\gamma_n)^{-2h} O(1) = O(1).
	\end{align*}
  Finally, by Lemma~\ref{lem:general-connect-lower}, there exists $f' \in ]f,1[$ such that with probability $1-O\left(\tfrac1n\right)$ there are at most $f'n$ informed nodes after $T(1,fn)$ rounds.
\end{proof}

%% file: exp_shrinking_new.tex

\subsection{Exponential Shrinking Regime. Upper Bound}

We now regard the regime that at most $gn$, $g$ a small constant, nodes are not informed, and that in each round each of these nodes has an approximately constant chance of becoming informed. From a very distant point of view, this part of the process vaguely resembles the exponential growth regime with time running backwards, but the details are too different to simply transfer our previous results to this setting.

We start in this section with the upper bound on the runtime. Throughout this section, we assume that our homogeneous epidemic protocol satisfies the following \emph{upper exponential shrinking conditions} including the covariance condition.

\begin{defn}[upper exponential shrinking conditions] \label{def:upper-exp-shrinking-conditions}
    Let $\rho_n$ be bounded between two positive constants.
	Let $0 < g < 1$ and $a, c \in \R_{\ge0}$ such that $e^{-{\rho_n}} + ag < 1$.
	We say that a homogeneous epidemic protocol satisfies \emph{the upper exponential shrinking conditions} if for any $n \in \N$ big enough, the following properties are satisfied, for all $u = n-k \le gn$.
	\renewcommand{\theenumi}{(\roman{enumi})}%
	\begin{enumerate}
		\item
			$1-p_k = 1-p_{n-u} \le e^{-{\rho_n}} + a\frac{u}{n}$;
		\item
			$c_k = c_{n-u} \le \frac{c}{u}$.
	\end{enumerate}
\end{defn}
Let us note that in this section we study the number of uninformed nodes $u := n-k$ instead of $k$, i.e., the number of informed ones.
We will show that $u$ shrinks by almost a constant factor each round.
So the main result of the section is the following theorem.

\begin{theorem}[upper bound for spreading time] \label{th:exp-shrinking-upper}
	Consider a homogeneous epidemic protocol satisfying the upper exponential shrinking conditions.
  Then there are constant $A', \alpha' > 0$ such that
	\begin{eqnarray*}
		&& \Expect[T(n-\lfloor g n \rfloor, n)] \le \tfrac1{\rho_n}\ln n + O(1), \\
		&& \Pr[T(n-\lfloor g n \rfloor, n) > \tfrac1{\rho_n}\ln n + r] \le A' e^{-\alpha r}\,
			\mbox{ for all $r \in \N$}.
	\end{eqnarray*}
\end{theorem}

We first note that the upper exponential shrinking conditions imply that nodes remain uninformed with at most a constant probability. Hence Lemma~\ref{lem:general-connect} shows that we reach any constant fraction of uninformed nodes in expected constant time. For this reason, we may conveniently assume that \emph{$g$ is an arbitrarily small constant} in the following. We shall also always assume that \emph{$n$ is large enough}.

The proof below follows the general principle established in this work, that is, we define for each number $u$ of uninformed nodes a suitable target $E_0(u)$ such that with sufficiently high probability $1-q(u)$ (following from the covariance condition and Chebyshev's inequality), one round started with at most $u$ uninformed nodes ends with at most $E_0(u)$ uninformed nodes. The choice of $E_0(u)$ is such that the sequence $u_0 = gn, u_1 = E_0(u_0), u_2 = E_0(u_1), \dots$ within $J = \frac 1{\rho_n} \ln(n) + O(1)$ steps reaches a constant $u_J$ and such that failure probabilities $q(u_i)$, $i = 0, \dots, J-1$, imply that only an expected constant number of rounds in addition to $J$ are needed to reach at most $u_J$ nodes. For the constant number of  $u_J$ or less remaining uninformed nodes, we use the simple waiting time argument that each of them needs an expected constant number of rounds to be informed, adding another constant number of rounds to the expected spreading time.

\subsubsection{Round Targets and Failure Probabilities}

Let us introduce the random variable $Y(u)$ being equal to the number of uninformed nodes at the end of a round started with $u$ uninformed ones.
Since $\Expect[Y(u)] = u(1-p_{n-u})$, the exponential shrinking conditions imply that
\begin{equation}
	\Expect[Y(u)] \le E(u) := u\left(e^{-{\rho_n}} + a\tfrac{u}{n}\right) \notag.
\end{equation}

As before, the Lemma~\ref{lem:exp-shrinking-failure-upper} shows that with good probability, $Y(u)$ is less than the $\emph{target value}$
\begin{equation}
	E_0(u) := E(u) + A u^{1-B} \label{eq:def-E0-upper-expshr},
\end{equation}
where $A > 0$ and $0 < B < 1/2$ are some constants chosen uniformly for all values of $u$ and $n$.
In addition we will choose $g$ and $A$ small enough (relative to $g$) to ensure that for all $u \le gn$, the target value $E_0(u)$ is less than $u$ (see Lemma~\ref{lem:exp-shrinking-E0-decrease}) and that the "chain" of consequent target values forms an exponentially decreasing sequence (see Lemma~\ref{lem:exp-shrinking-expk-upper}).

\begin{lemma}\label{lem:exp-shrinking-E0-decrease}
	Assume that $g$ and $A$ are sufficiently small constants. Then for all $u \in [1,gn]$, we have $E_0(u) < u$.
\end{lemma}
\begin{proof}
	Indeed, it suffices to show that
	\begin{equation}
		\tfrac{E_0(u)}{u} = e^{-{\rho_n}} + a\tfrac{u}{n} + Au^{-B} < 1 \notag.
	\end{equation}
	Since $u \in [1, gn]$, we have
	\begin{equation}
		\tfrac{E_0(u)}{u} \le e^{-{\rho_n}} + ag + A \notag.
	\end{equation}
	Clearly there exist positive $A$ and $g$ small enough such that the expression above is less than 1.
\end{proof}

We assume in the following that $g$ and $A$ are small enough to make the assertion of the lemma above true. We compute the target failure probabilities as follows.

\begin{lemma} \label{lem:exp-shrinking-failure-upper}
	For any $1 \le u < gn$,
	\begin{equation}
		\Pr[Y(u) \ge E_0(u)] \le q(u) := \tfrac{(1+a)e^{-{\rho_n}}+c}{A^2} \cdot \tfrac1{u^{1-2B}}\notag.
	\end{equation}
\end{lemma}
\begin{proof}
	Like in the proofs of Lemma~\ref{lem:exp-growth-failure}~and~\ref{lem:exp-growth-failure-lower}, using Chebyshev's inequality and taking into account $E(u) \ge \Expect[Y(u)]$, we compute
	\begin{align}
		\Pr&[Y(u) \ge E_0(u)]
		\le \Pr\left[Y(u) \ge \Expect[Y(u)] + Au^{1-B}\right]
		\le \tfrac{\Var[Y(u)]}{(Au^{1-B})^2} \notag.
	\end{align}
	From Lemma~\ref{lem:prelim:variance} and the covariance condition it follows that
	\begin{equation}
		\Var[Y(u)] \le \Expect[Y(u)] + cu \le \E[Y(u)] + cu \notag.
	\end{equation}
	Therefore,
	\begin{align}
		\Pr[Y(u) \ge E_0(u)]
			\le \tfrac{E(u) + cu}{A^2u^{2-2B}}
			\le \tfrac{(1+a)e^{-{\rho_n}}+c}{A^2} \cdot \tfrac1{u^{1-2B}} \notag.
	\end{align}
\end{proof}

%

\subsubsection{The Phase Calculus}

Let us define the sequence $u_j$ recursively by
\begin{equation}
	u_0 = gn, \quad u_{j+1} := E_0(u_j) \notag.
\end{equation}
The next observation follows from the definition.

\begin{observation}\label{obs:exp-shrinking-expk-upper}
	For any $j \ge 1$ we have $u_j \ge u_0e^{-j{\rho_n}}$.
	In particular, for any $j \le \tfrac1{\rho_n} \ln n$ we have $u_j \ge \tfrac{u_0}{n}$.
\end{observation}

\begin{lemma}\label{lem:exp-shrinking-expk-upper}
	By choosing $A$ in~\eqref{eq:def-E0-upper-expshr} and $g$ sufficiently small, we can assume that for all $j \le \tfrac1{\rho_n} \ln n$, we have $u_j \le 2u_0 e^{-j{\rho_n}}$.
\end{lemma}
\begin{proof}
	For $j=0$, there is nothing to prove.
	Consider $1 \le j \le \tfrac1{\rho_n} \ln n$ and assume that for all $i < j$ we have $u_i \le 2u_0 e^{-i{\rho_n}}$.
	We will show that $u_j \le 2u_0e^{-j{\rho_n}}$.
	By definition,
	\begin{align}
		u_j & = u_0 e^{-j{\rho_n}}
			\prod_{i=0}^{j-1}\left(1 + ae^{\rho_n} \tfrac{u_i}{n} + Ae^{\rho_n} u_i^{-B}\right) \notag\\
		& \le u_0 e^{-j{\rho_n}}
			\prod_{i=0}^{j-1} \exp\left(ae^{\rho_n} \tfrac{u_i}{n} + Ae^{\rho_n} u_i^{-B}\right) \notag\\
		& \le u_0 e^{-j{\rho_n}} \exp\left(\sum_{i=0}^{j-1}ae^{\rho_n} \tfrac{u_i}{n}
			+ \sum_{i=0}^{j-1}Ae^{\rho_n} u_i^{-B}\right) \label{eq:exp-shrinking-upper-eq1}.
	\end{align}
	We estimate separately the two sums.
	Since $u_i \le 2u_0 e^{-i{\rho_n}}$ for $i<j$, the first sum can be bounded by a geometric series:
	\begin{equation}
		\sum_{i=0}^{j-1}ae^{\rho_n} \tfrac{u_i}{n}
		\le \tfrac{ae^{\rho_n}}{n} \sum_{i=0}^{j-1} 2u_0e^{-i{\rho_n}}
		\le ae^{\rho_n} \cdot \tfrac{2u_0}{n} \cdot \tfrac1{1-e^{-{\rho_n}}} \notag.
	\end{equation}
	This expression is proportional to $\tfrac{u_0}{n} = g$, so by choosing $g$ small enough, we can bound it by $\tfrac{\ln2}{2}$.
	For the second sum we use Observation~\ref{obs:exp-shrinking-expk-upper} and obtain
	\begin{align}
		\sum_{i=0}^{j-1}Ae^{\rho_n} u_i^{-B}
		& \le Ae^{\rho_n} \sum_{i=0}^{j-1} u_0^{-B}e^{i{\rho_n} B}
			\le Ae^{\rho_n} u_0^{-B} \frac{e^{j{\rho_n} B}}{e^{{\rho_n} B}-1} \notag \\
		& \le Ae^{\rho_n} \left(\tfrac{n}{u_0}\right)^B \tfrac1{e^{{\rho_n} B}-1}
			\le Ae^{\rho_n} g^{-B} \tfrac1{e^{{\rho_n} B}-1} \label{eq:exp-shrinking-upper-eq42}.
	\end{align}
	By taking $A$ small enough, the result is also at most $\tfrac{\ln2}{2}$.
	Substituting the sums in~\eqref{eq:exp-shrinking-upper-eq1} by their bounds of $\tfrac{\ln 2}{2}$, we obtain
	\begin{equation}
		u_j \le u_0 e^{-j{\rho_n}} \exp\left(\tfrac{\ln2}{2} + \tfrac{\ln2}{2}\right)
		= 2u_0 e^{-j{\rho_n}} \notag.
	\end{equation}
\end{proof}
We assume in the following that $A$ and $g$ are as in Lemma~\ref{lem:exp-shrinking-expk-upper}.
Combining the lemma above with the definition of $q(u)$ in Lemma~\ref{lem:exp-shrinking-failure-upper}, one can easily see the following.
\begin{corollary}\label{cor:exp-shrinking-nphases-upper}
	There exists $J \le \tfrac1{\rho_n} \ln n$ such that (i) $q(u_J) < \tfrac12$ and (ii) $u_J = O(1)$.
\end{corollary}

By Lemma~\ref{lem:exp-shrinking-E0-decrease}, $u_j$ form a decreasing sequence.
We say that the rumor spreading process is in \emph{phase} $j$, $j \in \{0, \ldots, J-1\}$, if the number of informed nodes is in $[u_{j+1}, u_j[$.


\begin{lemma}\label{lem:exp-shrinking-ETj-upper}
	If the process is in phase $j < J$, then the number of rounds to leave phase $j$ is stochastically dominated by $1 + \Geom(1-Q_j)$, where $Q_j := q(u_{j+1})$.
\end{lemma}
\begin{proof}
	Consider a round with $u$ uninformed nodes.
	By definition, the process leaves the phase $j$ if $Y(u) < u_{j+1} = E_0(u_j)$.
	Since $E_0(u)$ is an increasing function, the upper bound for the probability to stay in phase $j$ in current round is the following.
	\begin{equation}
		\max_{u \in [u_{j+1},u_j[} \Pr[Y(u) \ge E_0(u_j)]
		\le \max_{u \in [u_{j+1},u_j[} \Pr[Y(u) \ge E_0(u)]
		\le q(u_{j+1}) \notag.
	\end{equation}
	So the number of rounds to leave phase $j$ is stochastically dominated by $1 + \Geom(1-Q_j)$.
\end{proof}

\begin{lemma} \label{lem:exp-shrinking-sum-Qj-upper}
	$\sum_{j=0}^{J-1} Q_j = O(1)$.
\end{lemma}
\begin{proof}
	By Lemma~\ref{lem:exp-shrinking-failure-upper}, we have
	\begin{equation}
		\sum_{j=0}^{J-1} Q_j \le O(1) \cdot \sum_{j=1}^J \tfrac1{u_j^{1-2B}} = O(1) \notag,
	\end{equation}
	where the last equality follows as in~\eqref{eq:exp-shrinking-upper-eq42}, using that $J \le \tfrac1{\rho_n} \ln n$.
\end{proof}

Now we can proof the main result of this section, i.e., Theorem~\ref{th:exp-shrinking-upper}.
\begin{proof}[Proof of Theorem~\ref{th:exp-shrinking-upper}]
    First, let $g' > 0$ be smaller than $g$.
    Then,
    \begin{equation}
        \Expect[T(n-\lfloor gn\rfloor, n)]
        \le \Expect[T(n-\lfloor gn \rfloor, n-\lceil g'n \rceil)]
        + \Expect[T(n- \lfloor g'n \rfloor, n)] \notag.
    \end{equation}
    By Lemma~\ref{lem:general-connect}, the exponential shrinking conditions imply that $\Expect[T(n-\lfloor gn \rfloor, n-\lceil g'n \rceil)]$ is at most a constant.
    In addition there exist $A'_0, \alpha'_0 > 0$ such that $\Pr[T(n-\lfloor gn \rfloor, n-\lceil g'n \rceil) > r/3] \le A'_0 e^{-\alpha'_0 r}$.
    We can hence assume that $g$ is small enough so that all Lemma~\ref{lem:exp-shrinking-E0-decrease}~and~\ref{lem:exp-shrinking-expk-upper} are satisfied.

	We denote by the random variable $T_j$ the number of rounds spent in phase $j$.
	With Corollary~\ref{cor:exp-shrinking-nphases-upper} and Lemma~\ref{lem:exp-shrinking-sum-Qj-upper}, we compute
	\begin{align}
		\Expect[T(n-\lfloor gn \rfloor, n-\lceil u_J \rceil)]
		& \le \sum_{j=0}^{J-1} \Expect[T_j]
			\le \sum_{j=0}^{J-1} \left(1 + \tfrac{Q_j}{1-Q_j}\right) \notag \\
		& = J + \sum_{j=0}^{J-1} \tfrac{Q_j}{1-Q_j}
			\le J + \tfrac1{1-Q_J} \cdot \sum_{j=0}^{J-1} Q_j \notag \\
		& = J + O(1) \notag.
	\end{align}
	Since $Q_j$ form a geometrical sequence, it follows from Lemma~\ref{lemma:sum geometrical-2} that there exist $A', \alpha' > 0$ such that
	\begin{equation}
		\Pr[T(n-\lfloor gn \rfloor, \lceil u_J \rceil) > J+r/2] \le A' e^{-\alpha' r}. \label{eq:1/th-30}
	\end{equation}
	For the last at most $u_J$ uninformed nodes, we argue as follows.
	Consider one uninformed node.
	From the exponential shrinking conditions it follows that the expected number of rounds until this node is informed is at most $O(1)$.
	So, $\Expect[T(n-\lfloor u_J\rfloor, n)] \le u_J \cdot O(1) = O(1)$.
	Finally,
	\begin{equation}
		\Expect[T(n-\lfloor gn \rfloor, n)]
		\le \Expect[T(n-\lfloor gn \rfloor, n-\lceil u_J \rceil)]
			+ \Expect[T(n-\lfloor u_J\rfloor, n)]
		\le \tfrac1{\rho_n} \ln n + O(1) \notag.
	\end{equation}
	
	To prove the tail bound statement, let $q = 1-\min_{k\in[n-u_J,n]}p_k$.
	Now we consider the epidemic protocol with $m = O(1)$ uninformed nodes.
	Since an uninformed node stays uninformed for $r/2$ rounds with probability at most $q^{r/2}$, we have $\Pr[T(n-m,n) > r/2] \le m \cdot q^{r/2}$.
	Combining the last inequation with~\eqref{eq:1/th-30}, we obtain
	\[
		\Pr[T(n-\lfloor gn \rfloor, n) > J + r]
		\le (u_J+A') \exp\left(-r\cdot\min\right\{\alpha', \tfrac{\ln q}2\left\}\right) \,.
	\]
	Since $u_J = O(1)$, the tail bound statement directly follows as in the proof of Theorem~\ref{th:exp-growth-upper}.
\end{proof}


\subsection{Exponential Shrinking Regime. Lower Bound}

\subsubsection{Exponential Shrinking Conditions}
\begin{defn}[lower exponential shrinking conditions] \label{def:upper-exp-shrinking-conditions}
    Let $\rho_n$ be bounded between two positive constants.
	Let $0 < g < 1$ and $a, c \in \R_{\ge0}$.
	We say that a homogeneous epidemic protocol satisfies \emph{the lower exponential shrinking conditions} if for any $n \in \N$ big enough, the following properties are satisfied, for all $u \le gn$ (resp. $k \in [n-\lfloor gn \rfloor, n]$).
	\renewcommand{\theenumi}{(\roman{enumi})}%
	\begin{enumerate}
		\item
			$1 - p_k = 1-p_{n-u} \ge e^{-{\rho_n}} - a\frac{u}{n}$;
		\item
			$c_k = c_{n-u} \le \frac{c}{u}$.
	\end{enumerate}
\end{defn}

\begin{theorem}[lower bound of spreading time] \label{th:exp-shrinking-lower}
	Consider a homogeneous epidemic protocol satisfying the lower exponential shrinking conditions (see definition above). There is a constant $g' \in ]0, 1[$ and further constants $A', \alpha'>0$ such that for any positive $g < g'$,
	\begin{align*}
		&\Expect[T(n-\lfloor gn \rfloor, n)] \ge \tfrac1{\rho_n}\ln n + O(1),\\
		&\Pr[T(n-\lfloor gn \rfloor, n) \le \tfrac1{\rho_n}\ln n - r] \le A' \exp(-\alpha'r)\, \mbox{ for all $r \in \N$}.\\
	\end{align*}
\end{theorem}

\subsubsection{Round Targets and Failure Probabilities}

Let $Y(u)$ be the number of uninformed nodes at the end of the round with $u$ uninformed ones.
From the exponential shrinking conditions it follows that
\begin{equation}
	\Expect[Y(u)] \ge E(u) := u \left(e^{-{\rho_n}} - a\tfrac{u}{n}\right) \notag.
\end{equation}
We define the \emph{target value} in the same way as for the upper bound.
\begin{equation}
	E_0(u) := E(u) - Au^{1-B}, \label{eq:def-E0-lower-expshr}
\end{equation}
where $A>0$ and $B \in ]0,1/2[$ are some constants chosen uniformly for all values of $u$ and $n$.
In addition $A$ is required to be small enough to satisfy Lemma~\ref{lem:exp-shrinking-expk-lower}.


\begin{lemma}~\label{lem:exp-shrinking-failure-lower}
	For any $u > gn$ and $u \in \N$,
	\begin{equation}
		\Pr[Y(u) \le E_0(u)] \le q(u) := \tfrac{e^{-{\rho_n}}+c}{A^2} \cdot \tfrac1{u^{1-2B}} \notag.
	\end{equation}
\end{lemma}
\begin{proof}
	As before, using Chebyshev's inequality and taking into account that $E(u) \le \Expect[Y(u)]$, we compute
	\begin{align}
		\Pr&[Y(u) \le E_0(u)]
			= \Pr\left[Y(u) \le E(u) \cdot \left(1-\tfrac{Au^{1-B}}{E(u)}\right) \right] \notag \\
		& \le \Pr\left[Y(u) \le \Expect[Y(u)]-Au^{1-B}\cdot\tfrac{\Expect[Y(u)]}{E(u)}\right] \notag \\
		& \le \tfrac{\Var[Y(u)]}{(Au^{1-B})^2} \cdot \tfrac{E(u)^2}{\Expect[Y(u)]^2} \notag.
	\end{align}
    From covariance condition, it follows that $\Var[Y(u)] \le \Expect[Y(u)] + cu$.
    Therefore,
    \begin{align}
        \Pr[Y(u) \le E_0(u)]
        & \le \left(1 + \tfrac{cu}{\Expect[Y(u)]}\right) \cdot \tfrac{E(u)}{\Expect[Y(u)]}
    	   		\cdot \tfrac{E(u)}{(Au^{1-B})^2} \notag \\
    	& \le \left(1 + \tfrac{cu}{E(u)}\right) \cdot \tfrac{E(u)}{(Au^{1-B})^2} \notag \\
        & = (E(u)+cu) \cdot \tfrac1{(Au^{1-B})^2}
        	\le \tfrac{e^{-{\rho_n}}+c}{A^2} \cdot \tfrac1{u^{1-2B}} \notag.
    \end{align}
\end{proof}

\subsubsection{The Phase Calculus}

We define the sequence $u_j$ recursively by
\begin{equation}
	u_0 := gn, \qquad u_{j+1} := E_0(u_j) \notag.
\end{equation}
The next observation follows from the definition.
\begin{observation} \label{obs:exp-shrinking-expk-lower}
	For any $j \ge 0$ we have $u_j \le u_0 e^{-j{\rho_n}}$.
\end{observation}

\begin{lemma}\label{lem:exp-shrinking-expk-lower}
	By choosing $A$ in~\eqref{eq:def-E0-lower-expshr} and $g$ sufficiently small, we can assume that for all $j \le \tfrac1{\rho_n} n$, we have $u_j \le \tfrac12 u_0 e^{-j{\rho_n}}$.
\end{lemma}
\begin{proof}
	For $j=0$, there is nothing to prove.
	Consider $1 \le j \le \tfrac1{\rho_n} \ln n$ and assume that for all $i<j$ we have $u_i \ge \tfrac12 u_0 e^{-i{\rho_n}}$.
	We will show that $u_j \ge \tfrac12 u_0 e^{-j{\rho_n}}$.
	By definition,
	\begin{align}
		u_j & = u_0e^{-j{\rho_n}} \prod_{i=0}^{j-1}
			\left(1-e^{\rho_n} a\tfrac{u_i}{n} - A\tfrac1{u_i^B}\right) \notag \\
		& \ge u_0e^{-j{\rho_n}} \left(1
			-\tfrac{e^{\rho_n} a}{n}\sum_{i=0}^{j-1}u_i - A\sum_{i=0}^{j-1}\tfrac1{u_i^B}\right) \notag
	\end{align}
	Like in the proof of Lemma~\ref{lem:exp-shrinking-expk-upper}, we estimate separately the two sums.
	Using Observation~\ref{obs:exp-shrinking-expk-lower}, we obtain for the first sum that
	\begin{equation}
		\tfrac{e^{\rho_n} a}{n}\sum_{i=0}^{j-1}u_i
		\le e^{\rho_n} a \tfrac{u_0}{n} \sum_{i \ge 0} e^{-i{\rho_n}}
		= \tfrac{e^{\rho_n} a}{1-e^{-{\rho_n}}} \cdot \tfrac{u_0}{n}
		= g \cdot O(1) \notag.
	\end{equation}
	By the hypothesis of induction, for any $i < j$, $u_i \ge \tfrac12 u_0 e^{-i{\rho_n}}$.
	Since $j < \tfrac1{\rho_n} \ln n$,
	\begin{align}
		A \sum_{i=0}^{j-1}\tfrac1{u_i^B}
		& \le \tfrac{A}{2^B u_0^B} \sum_{i=0}^{j-1} e^{-i{\rho_n} B}
			\le \tfrac{A}{2^B u_0^B} \cdot \tfrac{e^{j{\rho_n} B}}{e^{{\rho_n} B} - 1} \notag \\
		& = \tfrac{A}{2^B(e^{{\rho_n} B}-1)} \cdot \tfrac{n^B}{u_0^B}
			= \tfrac{A}{2^B(e^{{\rho_n} B}-1)} \cdot g^{-B}
			= Ag^{-B} \cdot O(1) \notag.
	\end{align}
	Then, by choosing $A$ and $g$ small enough, we can bound both sums by 1/4, so that
	\begin{equation}
		u_j \ge u_0 e^{-j{\rho_n}} \left(1 - \tfrac14 - \tfrac14\right)
		\ge \tfrac12 u_0 e^{j-{\rho_n}} \notag.
	\end{equation}
\end{proof}

Having $u_j$ bounded from above and below, one can easily see the following.
\begin{corollary} \label{cor:exp-shrinking-nphases-lower}
	There exists $J = \tfrac1{\rho_n} \ln n + O(1)$ such that $u_J > 1$ for any $n$ big enough.
\end{corollary}

By definition, the $u_j$ form a non-decreasing sequence.
We say that the rumor spreading process is in phase $j$, $j \in \{0,\ldots,J-1\}$, if the number of informed nodes is in $[u_{j+1}, u_j[$.

\begin{lemma}
	If the process is in phase $j < J-1$, then the probability that it "leapfrogs" phase $j+1$ (i.e., proceeds to phase $j+2$ or further in current round) is at most $q(u_j)$.
\end{lemma}
\begin{proof}
    Consider a round with $u \in [u_{j+1}, u_j[$ uninformed nodes.
    The protocol jumps over the phase $j+1$, if at the end of current round $Y(u) < u_{j+2} = E_0(u_{j+1})$.
    Since $E_0$ is increasing,
	\begin{equation}
		\Pr[u < u_{j+2}]
		\le Pr[u < E_0(u)]
		\le q(u) \notag.
	\end{equation}
	Since $q(u)$ is a decreasing function, the upper bound for the probability to jump over phase $j+1$ is  the following.
	\begin{equation}
		\max_{u \in [u_{j+1}, u_j[} \Pr[u < u_{j+2}]
		\le q(u_{j+1}) \notag.
	\end{equation}
\end{proof}

Now we can proof the main result of this section, i.e., Theorem~\ref{th:exp-shrinking-lower}.

\begin{proof}[Proof of Theorem~\ref{th:exp-shrinking-lower}]
    Let $\tau$ be the first round $t$ (of this shrinking phase) in which the process leapfrogs a phase. Let $\tau = \infty$ if such an event does not occur. 
    By Corollary~\ref{cor:exp-shrinking-nphases-lower}, the interval $[1,gn]$ is cut into at least $J = \tfrac1{\rho_n} \ln n + O(1)$ phases.
    Clearly, if $\tau < J$, then $T(n-\lfloor gn \rfloor, n) \ge \tau$, and if $\tau \ge J$, then $T(n-\lfloor gn \rfloor, n) \ge J$.

    If $\tau = J-t$, then the process in phase $J-t$, that is, from some number $u$ of uninformed nodes belonging to phase $J-t$, makes an exceptionally large progress from. Since $q(u)$ is a decreasing function, we have $\Pr[\tau = J-t] \le q(u_{J-t})$. Consequently, using the fact that $q(u_j)$ forms a decreasing geometric sequence, we obtain 
    \begin{align*}
    \Pr[T(n-\lfloor gn \rfloor, n) \le J-t] \le \Pr[\tau \le J-t] \le q(u_0) + q(u_1) + \ldots + q(u_{J-t}) = O(q(u_{J-t})).
    \end{align*}
    Then, using $u_{J-t} \ge O(1) \cdot u_J \cdot e^{{\rho_n} t}$, we compute
    \begin{align*}
    \Pr[T(n-\lfloor gn \rfloor, n) \le J-t]
    & \le O(q(u_{J-t})) \le O(1) u_{J-t}^{-2B+1} \\
    & \le O(1) (u_J e^{\rho_n t})^{-2B+1} \le O(1) \exp(-\Omega(t)).
		\end{align*}
    Applying Lemma~\ref{lem:exp-shrinking-failure-lower}, we obtain
    \begin{align*}
        \Expect[T(n-\lfloor gn \rfloor, n)]
        &\ge J \Pr[\tau > J] + \sum_{t=1}^{J-1} t\cdot\Pr[\tau = t]
        = J - \sum_{t=1}^{J-1} t \Pr[\tau = J-t] \\
        &\ge J - \sum_{t=1}^{J-1} t q(u_{J-t})
        \ge J - \tfrac{e^{\rho_n}+a+c}{A^2} \cdot \sum_{t=1}^{J-1} \tfrac{t}{u_{J-t}^{1-2B}}.
    \end{align*}
    Since $B < 1/2$ and $u_{J-t} \ge O(1) \cdot u_J \cdot e^{{\rho_n} t}$, the sum above converges.
    Therefore,
    \begin{align}
    	\Expect[T(n-\lfloor gn \rfloor, n)] \ge J + O(1) \notag.
    \end{align}
\end{proof} 

%% file: double_exp_shrinking_new.tex
\subsection{Double Exponential Shrinking Regime. Upper Bound.}

In the following two sections we consider the regime in which uninformed nodes remain uninformed with probability proportional to the fraction uninformed nodes, or, more generally, some positive power $\ell-1$ there of.
Such a regime often occurs in protocols using pull operations.
We show that the \emph{fraction} of uninformed nodes is raised to the $\ell$-th power each round and that such a regime informs the last $gn$ nodes ($g$ a small constant) in a double logarithmic number of rounds.

We discuss the upper bound on the runtime first.
Throughout this section, we assume that our homogeneous epidemic protocol satisfies the following \emph{upper double exponential shrinking conditions} including the covariance condition.

\begin{defn}[upper double exponential shrinking conditions] \label{def:dexp}
	Let $g, \alpha \in [0,1]$, $\ell > 1$, and $a, c \in \R_{\ge0}$ such that $ag^{\ell-1} < 1$.
	We say that a homogeneous epidemic protocol satisfies \emph{the upper double exponential shrinking conditions}
	if for any $n$ big enough, the following properties are satisfied for all $u \in [n^{1-\alpha}, gn]$.
	\renewcommand{\theenumi}{(\roman{enumi})}%
	\begin{enumerate}
		\item $1-p_{n-u} \le a\left(\tfrac{u}{n}\right)^{\ell-1}$.
		\item $c_{n-u} \le c \tfrac{n}{u^2}$.
	\end{enumerate}
\end{defn}

Similarly to the exponential shrinking regime we argue with the number $u$ of uninformed nodes rather than the number $k$ of informed ones.
To ease the notation in the double exponential shrinking regime we use the \emph{fraction} $\eps := \tfrac{u}{n}$ of uninformed nodes instead of the absolute number $u$.
Thus, the double exponential shrinking conditions turns into the following bounds, valid for all $\eps \in [n^{-\alpha},g]$ with $\eps n \in \N$.
\renewcommand{\theenumi}{(\roman{enumi})}%
\begin{enumerate}
	\item $1-p_{n(1-\eps)} \le a\eps^{\ell-1}$.
	\item $c_{n(1-\eps)} \le \eps^{-2}\tfrac{c}{n}$.
\end{enumerate}

In the definition above, we cover the rounds starting with a number of uninformed nodes between $n^{1-\alpha}$ and $gn$. While, by taking $\alpha=1$ this would allow to analyze the process until all nodes are informed, it turns out that the crucial part is reduce the number of uninformed nodes from $\Theta(n)$ to $n^{1-\alpha}$ for an arbitrarily small constant $\alpha$. For $u \in [1,n^{1-\alpha}]$, the double exponential shrinking conditions can be relaxed: the covariance condition is no longer needed and it is sufficient to bound uniformly the probability of a node to stay uninformed by $n^{-\tau}$, for some $\tau < 1$.

The main result of the section is the following theorem.

\begin{theorem}\label{th:double-exp-shrinking-upper}
	Consider a homogeneous epidemic protocol satisfying the upper double exponential shrinking conditions in $[n^{-\alpha},g]$.
	Suppose further that there exists $\tau > 0$ such that $1-p_{n-u} \le n^{-\tau}$ for all $u \le n^{1-\alpha}$.
	
	Then there exist constant $A', \alpha' > 0$ such that
	\begin{eqnarray*}
		&& \Expect[T(\lceil (1-g)n \rceil,n)] \le \log_\ell \ln n + O(1), \\
		&& \Pr[T(\lceil (1-g)n \rceil,n) \ge \log_\ell \ln n + r] \le O(n^{-\alpha' r+A'})\, \mbox{ for all $r \in \N$}.
	\end{eqnarray*}
\end{theorem}


\subsubsection{Round Targets and Failure Probabilities}

Let the random variable $y(\eps)$ denote to the fraction of uninformed nodes at the end of a round started with $\eps n$ uninformed ones.
The double exponential shrinking conditions state that
\begin{equation}
    \Expect[y(\eps)] \le E(\eps) := a\eps^\ell \notag.
\end{equation}

\begin{lemma}\label{lem:double-exp-shrinking-variance}
	$\Var[y(\eps)] \le \tfrac{1+c}{n}$.
\end{lemma}
\begin{proof}
	Indeed, $\Var[y(\eps)] = \tfrac1{n^2}\Var[Y(\eps)]$, where $Y(\eps) := ny(\eps)$ is the number of uninformed nodes at the end of the round.
	By Lemma~\ref{lem:prelim:variance},
	\begin{equation}
		\Var[Y(\eps)] \le \Expect[Y(\eps)] + (n\eps)^2 c_{n(1-\eps)}
		\le n + cn\notag.
	\end{equation}
\end{proof}

The next lemma states that with good probability, $y(\eps)$ is less than the \emph{target value} $2E(\eps)$.

\begin{lemma}\label{lem:double-exp-shrinking-failure-upper}
    For any fraction of uninformed nodes $\eps \in [n^{-\alpha}, g]$,
    \begin{equation}
        \Pr[y(\eps) \ge 2E(\eps)] \le q := \tfrac{(1+c)}{a^2}n^{2\alpha\ell-1} \notag.
    \end{equation}
\end{lemma}
\begin{proof}
    Applying Chebyshev's inequality and taking into account that $E(\eps) \ge \Expect[y(\eps)]$, we compute
    \begin{equation}
        \Pr[y(\eps) \ge 2E(\eps)]
        \le \Pr[y(\eps) \ge \Expect[y(\eps)] + E(\eps)]
        \le \tfrac{\Var[y(\eps)]}{E(\eps)^2} \notag.
    \end{equation}
    By Lemma~\ref{lem:double-exp-shrinking-variance} and since $\eps \ge n^{-\alpha}$,
    \begin{equation}
    	\Pr[y(\eps) \ge 2E(\eps)]
    	\le \tfrac{1+c}{n} \cdot \tfrac1{(a\eps^\ell)^2}
    	\le \tfrac{1+c}{a^2} n^{2\alpha\ell-1}\notag.
    \end{equation}
\end{proof}
Our choice to analyze the double exponential shrinking regime only up to $n^{1-\alpha}$ uninformed nodes allows us to define $q$ independent of $\eps$.
Since the double exponential shrinking conditions imply the second assumption of Theorem~\ref{th:double-exp-shrinking-upper}, without loss of generality we may assume that $\alpha< \tfrac1{2\ell}$, and that consequently $q=n^{-\Theta(1)}$.


\subsubsection{The Phase Calculus}

Let us define the sequence $\eps_j$ recursively by
\begin{equation}
    \eps_0 := g, \quad \eps_{j+1} := 2E(\eps_j) \notag.
\end{equation}
The following observation can be obtained by a simple induction.
\begin{observation}\label{obs:double-exp-shrinking-epsj}
    For all $j \ge 0$, $\eps_j = (2a)^{\frac{\ell^j-1}{\ell-1}} g^{\ell^j}$.
    In particular, the $\eps_j$ form a decreasing sequence if $g < (2a)^{-\tfrac1{\ell-1}}$.
\end{observation}
In the following we assume that $g$ is small enough to ensure that the $\eps_j$ decrease.
Applying logarithm twice to the previous equation one can also see the following.
\begin{corollary} \label{cor:double-exp-shrinking-nphases}
There exists $J = \log_\ell \ln n + O(1)$ such that for any $n$ big enough
\begin{equation}
    n^{-\alpha} < \eps_J \le \left(\tfrac{n^{-\alpha}}{2a}\right)^{1/\ell} \notag.
\end{equation}
\end{corollary}
\begin{proof}
	From Observation~\ref{obs:double-exp-shrinking-epsj} we see that the biggest $J$ such that $\eps_J > n^{-\alpha}$ is equal to $\log_\ell \ln n + O(1)$.
	Since $\eps_{J+1} < n^{-\alpha}$, we have $\eps_J < \left(\tfrac{n^{-\alpha}}{2a}\right)^{1/\ell}$.
\end{proof}

We say that the process is in phase $j$ if the fraction $\eps$ of uninformed nodes is in $]\eps_{j+1}, \eps_j]$.
\begin{lemma} \label{lem:double-exp-shrinking-ETj-upper}
    If the process is in phase $j$, $j < J$, then the number of rounds to leave phase $j$ is stochastically dominated by $1 + \Geom(1-q)$.
\end{lemma}
\begin{proof}
    Consider a round starting with $\eps n$ uninformed nodes.
    By construction, the process leaves the phase $j$ if $y(\eps) \le \eps_{j+1} = 2E(\eps_j)$.
    Since $E(\cdot)$ is an increasing function, an upper bound for the probability to stay in phase $j$ in the current round is
    \begin{equation}
        \max_{\eps\in]\eps_{j+1},\eps_j]} \Pr[y(\eps) > 2E(\eps_j)]
        \le \max_{\eps\in]\eps_{j+1},\eps_j]} \Pr[y(\eps) \ge 2E(\eps)]
        \le q \notag.
    \end{equation}
    Hence, the number of rounds the process spends in phase $j$ is stochastically dominated by a random variable with distribution $1+\Geom(1-q)$.
\end{proof}

Let us now prove the main theorem of the section.

\begin{proof}[Proof of Theorem~\ref{th:double-exp-shrinking-upper}]
	From Lemma~\ref{lem:general-connect} it follows that for any $g' < g$ we have $\Expect[T(n-\lfloor gn \rfloor, n - \lceil g'n \rceil)] = O(1)$.
	So without loss of generality we can assume that $g < (2a)^{-\tfrac1{\ell-1}}$ that is required by Observation~\ref{obs:double-exp-shrinking-epsj} and, thus, by Corollary~\ref{cor:double-exp-shrinking-nphases}.
	Let the random variable $T_j$ denote the number of rounds spent in phase $j$.
    With Corollary~\ref{cor:double-exp-shrinking-nphases} as well as Lemma~\ref{lem:double-exp-shrinking-failure-upper}~and~\ref{lem:double-exp-shrinking-ETj-upper}, we compute
    \begin{eqnarray}
        && \Expect[T(n-\lfloor gn \rfloor, n - \lceil\eps_Jn\rceil)]
        	\le \sum_{j=0}^{J-1} \Expect[T_j]
	        \le J \left(1 + \tfrac{q}{1-q}\right)
        	= \log_\ell \ln n + O(1) \label{eq:2/th-42} \\
        && \Pr\left[T(n-\lfloor gn \rfloor, n - \lceil \eps_Jn \rceil) > J+r\right]
        	\le J q^{-r} = n^{-\Omega(r)} \,. \label{eq:1/th-42}
    \end{eqnarray}
    By Corollary~\ref{cor:double-exp-shrinking-nphases}, $\eps_J < \left(\tfrac{n^{-\alpha}}{2a}\right)^{1/\ell}$.
    Consequently, there exists $\alpha' \in ]0,\alpha[$ such that $\eps_J < n^{-\alpha'}$ for any $n$ large enough.
    Without loss of generality we can assume that for any $u \le n^{1-\alpha'}$ we have $1-p_{n-u} \le n^{-\tau}$ (for $u \in [n^{1-\alpha}, n^{1-\alpha'}]$ it follows from the double exponential shrinking condition).
    Now suppose $u_0 \le n^{1-\alpha'}$ and consider $T(n-u_0,n)$.
	By the argument above, any of the $u_0$ uninformed nodes stays uninformed for $r \ge 1$ rounds with probability at most $n^{-\tau r}$.
	Then by the union bound, we have $\Pr[T(n-u_0, n) > r] \le P_r := \min\{1,n^{-\tau r + 1-\alpha'}\}$, that together with~\eqref{eq:1/th-42} proves the tail bound statement.
	
	Finally, $\E[T(n-u_0,n)] \le 1+\sum_{r\ge1} P_r = O(1)$, for any $u_0 \le n^{1-\alpha}$.
	Then, together with~\eqref{eq:2/th-42} it proves that $\E[T(n-\lfloor gn \rfloor, n)] \le \log_\ell \ln n + O(1)$.
\end{proof}

\subsection{Double Exponential Shrinking Regime. Lower Bound.}

We now prove that under lower bound conditions comparable to the upper bound conditions of the previous section, we obtain a lower bound on the runtime equaling our upper bound apart from an additive constant.

\subsubsection{Double Exponential Shrinking Conditions}

Throughout this section, we assume that the following \emph{lower double exponential shrinking conditions} are satisfied.

\begin{defn}[lower double exponential shrinking conditions] \label{def:dexp}
	Let $g, \alpha \in ]0,1]$ and $\ell > 1$.
	Let $a, c \in \R_{\ge0}$.
	We say that a homogeneous epidemic protocol satisfies \emph{the lower double exponential shrinking conditions} if for any $n$ big enough, the following properties are satisfied for all $u \in [n^{1-\alpha},gn]$.
	\renewcommand{\theenumi}{(\roman{enumi})}%
	\begin{enumerate}
		\item $1-p_{n-u} \ge a\left(\tfrac{u}{n}\right)^{\ell-1}$.
		\item $c_{n-u} \le c\tfrac{n}{u^2}$.
	\end{enumerate}
\end{defn}

Similarly to the upper double exponential shrinking conditions, we work mostly with the fraction $\eps := \tfrac{u}{n}$ of uninformed nodes instead of the absolute number~$u$.
Thus, the double exponential shrinking conditions turns into the following bounds, valid for all $\eps \in [n^{-\alpha},g]$ with $\eps n \in \N$.
\renewcommand{\theenumi}{(\roman{enumi})}%
\begin{enumerate}
	\item $1-p_{n(1-\eps)} \ge a\eps^{\ell-1}$.
	\item $c_{n(1-\eps)} \le \eps^{-2}\tfrac{c}{n}$.
\end{enumerate}

The main result of this section is the following theorem.

\begin{theorem}\label{th:double-exp-shrinking-lower}
	Consider a homogeneous epidemic protocol satisfying the lower double exponential shrinking conditions in the interval $[n^{1-\alpha},gn]$. Let $r$ be a sufficiently large constant (possibly depending on $\alpha$). Then,
	\begin{align*}
		&\Expect[T(n-\lceil gn \rceil, n-\lfloor n^{1-\alpha} \rfloor)] \ge \log_\ell \ln n + O(1),\\
		&\Pr[T(n-\lceil gn \rceil, n-\lfloor n^{1-\alpha} \rfloor) \le \log_\ell \ln n - r] \le O(n^{-1+2\alpha\ell}),
	\end{align*}
\end{theorem}

\subsubsection{Round Targets and Failure Probabilities}
Let again $y(\eps)$ denote the fraction of uninformed nodes at the end of a round started with $\eps n$ uninformed ones.
The double exponential shrinking conditions state that
\begin{equation}
    \Expect[y(\eps)] \ge E(\eps) := a\eps^\ell \notag.
\end{equation}

The next lemma gives that with good probability, $y(\eps)$ is at least the \emph{target} value $E(\eps)/2$.

\begin{lemma}\label{lem:double-exp-shrinking-failure-lower}
    For any fraction of uninformed nodes $\eps \in [n^{-\alpha}, g]$,
    \begin{equation}
        \Pr\left[y(\eps) \le \tfrac12E(\eps)\right] \le \tfrac{4+4c}{a^2 \eps^2 n} \le q := \tfrac{4+4c}{a^2}n^{2\alpha\ell-1} \notag.
    \end{equation}
\end{lemma}
\begin{proof}
    Applying Chebyshev's inequality and taking into account that $\Expect[y(\eps)] \ge E(\eps)$, we compute
    \begin{equation}
        \Pr[y(\eps) \le \tfrac12E(\eps)]
        \le \Pr\left[y(\eps) \le \Expect[y(\eps)] - \tfrac12E(\eps)\right]
        \le 4 \cdot \tfrac{\Var[y(\eps)]}{E(\eps)^2} \notag.
    \end{equation}
    By the same arguments like in Lemma~\ref{lem:double-exp-shrinking-variance}, $\Var[y(\eps)] \le \tfrac{1+c}{n}$.
    Since $\eps \ge n^{-\alpha}$, we have $E(\eps) \ge an^{-\alpha\ell}$, and the claim of the lemma directly follows.
\end{proof}

Similarly to the upper bound, our choice to analyze the double exponential shrinking regime only up to $n^{1-\alpha}$ uninformed nodes allows us to define $q$ independent of $\eps$.
We also assume that $\alpha< \tfrac1{2\ell}$ so that $q=n^{-\Theta(1)}$.

\subsubsection{The Phase Calculus}

Let us define the sequence $\eps_j$ recursively by
\begin{equation}
    \eps_0 := g, \quad \eps_{j+1} := \tfrac12E(\eps_j) \notag.
\end{equation}
The next observation follows from the definition by a simple induction.
The $\eps_j$ are decreasing simply because $\eps_{j+1}=\tfrac12E(\eps_j) < \Expect[y(\eps_j)] \le \eps_j$.
Note that $y(\eps) \le \eps$ with probability one for any homogeneous protocol.
\begin{observation}
	For all $j \ge 1$,
    $\eps_j = (a/2)^{\frac{\ell^j-1}{\ell-1}} g^{\ell^j}$.
    The $\eps_j$ form a decreasing sequence.
\end{observation}
In the rest of the section we assume that $g < (a/2)^{-\tfrac1{\ell-1}}$.
Applying logarithm twice to the previous equation one can also see the following.
\begin{observation}\label{obs:double-exp-shrinking-nphases-lower}
    There exists $J = \log_\ell \ln n + O(1)$ such that $\eps_J > n^{-\alpha}$.
\end{observation}

As before, we say that the process is in phase $j$ if the fraction $\eps$ of uninformed nodes is in $]\eps_{j+1}, \eps_j]$.
\begin{lemma}\label{lem:double-exp-shrinking-ETj-lower}
    If the process starts in phase $j$, $j < J$, then the probability that after one round it is in phase $j+2$ or higher is at most $q$.
\end{lemma}
\begin{proof}
    Consider a round starting with $\eps n$ uninformed nodes, where $\eps \in ]\eps_{j+1},\eps_j]$.
    By construction, the process leapfrogs phase $j+1$ if $y(\eps) \le \eps_{j+2} = \tfrac12E(\eps_{j+1})$.
    Since $E(\cdot)$ is an increasing function, an upper bound for the probability to jump over phase $j+1$ is
    \begin{equation}
        \max_{\eps\in]\eps_{j+1},\eps_j]} \Pr[y(\eps) \le \tfrac12E(\eps_{j+1})]
        \le \max_{\eps\in]\eps_{j+1},\eps_j]} \Pr[y(\eps) \le \tfrac12E(\eps)]
        \le q \notag.
    \end{equation}
\end{proof}

\begin{proof}[Proof of Theorem~\ref{th:double-exp-shrinking-lower}]
  Consider the rumor spreading process starting with $\eps_0 n = gn$ uninformed nodes. By Lemma~\ref{lem:double-exp-shrinking-ETj-lower}, with probability at least $(1-q)^J \ge 1 - Jq$, the process visits each phase $j \in [0..J-1]$, which naturally takes at least $J-1$ rounds. Consequently, by definition of $J$ in Observation~\ref{obs:double-exp-shrinking-nphases-lower}, we have
	\begin{align*}
		\Expect[T(n-\lceil gn \rceil, n-\lfloor n^{1-\alpha} \rfloor)]
			&\ge \Expect[T(n-\lceil n\eps_0 \rceil, n-\lfloor n\eps_J\rfloor)] \\
			&\ge (J-1)(1-Jq) = \log_\ell \ln n + O(1) \notag.
	\end{align*}
	The large-deviation statement follows immediately from adding the failure probabilities $\frac{4+4c}{a^2\eps_j^2 n}$, $j = 0, \dots, J-1$, from Lemma~\ref{lem:double-exp-shrinking-failure-lower}.
\end{proof}

%% file: basic_push_pull.tex
\section{Application of our Method to the Classic Protocols}\label{sec:classics}

In this section, we define the classic push, pull, and push-pull protocols, give some background information on them, and show how the methods developed above easily give very sharp (tight apart from additive constants) rumor spreading times. For this, we easily convince ourselves that all three protocols satisfy the exponential growth conditions. The push protocol satisfies the exponential shrinking conditions, whereas the pull and push-pull protocols both satisfy the double exponential shrinking conditions. For all these conditions, we can show for the upper and lower bound part of the conditions the same value for the critical parameter $\gamma_n$, ${\rho_n}$, and~$\ell$), which is why we then obtain sharp estimates for the rumor spreading times.

We stick to the usual convention that for rumor spreading in complete graphs we allow that nodes call themselves, that is, the random communication partner is chosen uniformly at random from all nodes. By replacing all $(1-\tfrac 1n)$ terms with $(1-\tfrac 1 {n-1})$, the elementary proofs below can easily be transformed to the setting where nodes only call random neighbors in the complete graph.


\subsection{Push Protocol}

The push protocol appeared in the computer science literature first in the works of Frieze and Grimmett~\cite{FriezeG85} (as a technical tool to analyze the all-pairs shortest path problem on complete digraphs with random edge weights) and, under the name \emph{rumor mongering}, Demers et al.~\cite{Demers87}, the first work that proposed rumor spreading as a robust and scalable method to maintain consistency in replicated databases. In the push protocol, in each round each node knowing the rumor calls a random neighbor and gossips the rumor to it.

The push protocol is the most intensively studied rumor spreading process. It has been proven that with high probability it disseminates a rumor known to a single node to all others in time logarithmic in the number $n$ of nodes when the communication networks is a complete graph (see below), a random graph in the $G(n,p)$ model with $p \ge (1+\eps) \ln(n)/n$, that is, only very slightly above the connectivity threshold, or a hypercube~\cite{FeigePRU90}, or a random regular graph~\cite{FountoulakisP10} (and this list is not complete).

For the complete graph, Frieze and Grimmett~\cite{FriezeG85} show (among other results) that with high probability, the rumor spreading time is $\log_2 n + \ln n \pm o(\log n)$. This estimate was sharpened by Pittel~\cite{Pittel87}, who proved that for any $h = \omega(1)$, the rumor spreading time with high probability is $\log_2 n + \ln n \pm h(n)$. The first explicit bound for the expected runtime, $\lfloor \log_2 n \rfloor + \ln n - 1.116 \le E[S_n] \le \lceil \log_2 n \rceil + \ln n + 2.765 + o(1)$ was shown in~\cite{DoerrK14}. All these works are relatively technical (see, e.g., the 9-pages proof of~\cite{Pittel87}) and heavily exploit particular properties of the push process (e.g., a birthday paradox argument for the first $\log_2(o(\sqrt n))$ calls and a reduction to the coupon collector process for the last roughly $\ln n$ rounds in~\cite{DoerrK14}).

With the methods developed in this work, we only need to show that the push protocol satisfies the exponential growth and shrinking conditions (with $\gamma_n = 1$ and ${\rho_n}=1$), which is very easy. This reproves the bound of~\cite{DoerrK14} cited above apart from the additive constants, but with a, as we believe, much simpler proof.


\begin{theorem}\label{th:example-push}
	The expected rumor spreading time of the push protocol on the complete graph with $n$ vertices is $\log_2 n + \ln n \pm O(1)$.
\end{theorem}

\begin{proof}
	Consider one round of the protocol. Let $x_1, x_2$ be two different uninformed nodes. Let $X_1$ and $X_2$ be the indicator random variables for events that $x_1$ resp.~$x_2$ become informed. Clearly, if we condition on that $x_1$ becomes informed, then it is slightly less likely that $x_2$ becomes informed. Consequently, $\Cov[X_1,X_2] < 0$ and the covariance part of the exponential growth and shrinking conditions is satisfied.

	Therefore, it remains to analyze the probability $p_k$ of an uninformed node to become informed.
	
	For the exponential growth regime, suppose that $k$ nodes are informed.	An uninformed node remains uninformed when all informed nodes fail to call it.
	Consequently, it becomes informed with probability $p_k = 1 - \left(1-\tfrac1n\right)^k$. With the estimates
	\[\tfrac{k}{n} - \tfrac{k^2}{2n^2} \le p_k \le \tfrac{k}{n}\]
	we see that the protocol satisfies the exponential growth conditions with parameter $\gamma_n = 1$. More precisely, we can take $\gamma_n=1$, $f=1$, $b=0$ and $c=0$ is both the upper and lower bound exponential growth condition. Taking $a=1$ satisfies the upper exponential growth condition, taking $a=0$ suffices for the lower exponential growth condition.
	
	For the exponential shrinking conditions, suppose that there are $u$ uninformed nodes.	Again, the probability for a node to stay uninformed is $1-p_{n-u} = \left(1-\tfrac1n\right)^{n-u}$.
By Corollary~\ref{cor:prelim:(1-1/n)^(n-u)}, for any $u<n$ we have the following estimate.
\[\tfrac1e \le 1-p_{n-u} \le \tfrac1e + \tfrac2e\cdot\tfrac{u}{n}\]
	The push protocol hence satisfies the exponential shrinking conditions (from $gn := \tfrac12 n$ uninformed nodes on) with parameter ${\rho_n}=1$.
	
	By Theorems~\ref{th:exp-growth-upper},~\ref{th:exp-growth-lower},~\ref{th:exp-shrinking-upper},~and~\ref{th:exp-shrinking-lower}, the expected rumor spreading time of the push protocol is $\log_2 n + \ln n \pm O(1)$.
\end{proof}

\subsection{Pull Protocol}

The pull protocol is dual to the push protocol in the sense that now in each round, each uninformed node calls a random neighbor and becomes informed if the latter was informed. We are not aware of a convincing practical motivation for this protocol, however, it has been very helpful in proving performance guarantees for other protocols, e.g., in~\cite{Giakkoupis11}. Note that the duality between the two protocols immediately shows that the probability that the push protocol in $t$ rounds moves a rumor initially present at a node $u$ to a node $v$ equals the probability that the pull protocol gets the rumor from $v$ to $u$ in $t$ rounds, but this does not imply that both protocols have the same rumor spreading times (as also Theorems~\ref{th:example-push} and~\ref{th:example-pull} show).

We are not aware of any performance guarantees proven for the pull protocol. Some existing results for the push protocol obviously can be transformed into results for the pull protocol via the duality and union bounds. For complete graphs, we do not see how this would give bounds stronger than $\Theta(\log n)$.

Interestingly, the expansion phase of the pull protocol (when viewed from a distance) resembles the expansion phase of the push protocol---the probability that an uninformed node becomes informed in a round starting with $k$ informed nodes is $p_k = \tfrac kn$ and thus, for small $k$, very close to the $\tfrac kn - \Theta(\tfrac{k^2}{n^2})$ probability of the push protocol. Nevertheless, the precise processes are very different. For example, in the push protocol we almost surely observe a perfect doubling of the number of informed nodes as long as $o(\sqrt n)$ nodes are informed. For the pull protocol, the number of newly informed nodes in the first round is binomially distributed with parameters $n-1$ and $\frac 1n$, so the probability for a perfect doubling is asymptotically equal to $\tfrac 1e$. For this reason, the existing analyses of the push protocol cannot easily be transferred to the pull protocol. This is different for our method, which ignored many details of the process and only relies on the rough characteristics $p_k$ and $c_k$ of the process. We show below that the similar values of $p_k$ lead to the same $\log_2 n \pm O(1)$ time it takes to inform a constant fraction of the nodes. From that point on, the double exponential shrinking conditions are obvious, leading to a double logarithmic remaining time.

%

\begin{theorem}\label{th:example-pull}
	The expected rumor spreading time of the pull protocol on the complete graph with $n$ vertices is $\log_2 n + \log_2 \ln n \pm O(1)$.
\end{theorem}

\begin{proof}
  Clearly, the events that uniformed nodes become informed are mutually independent. Hence the covariance conditions are exponential growth and double exponential shrinking regimes are satisfied.
	
	An uninformed node becomes informed if its call reaches an informed node. Hence for all $k \in [1..n-1]$, we have  $p_k = k/n$. This shows that both the upper and lower exponential growth conditions are satisfied with parameter $\gamma_n=1$ (and $f=1$, $a=0$, $b=0$, $c=0$).
	
	For the same reason, the probability $1-p_{n-u}$ that an uninformed node remains uninformed when $u$ nodes are uninformed, is $1- p_{n-u} = 1 - \tfrac{n-u}{n} = \tfrac un$. Consequently, the upper and lower double exponential shrinking conditions are satisfied with $\ell = 2$ (and $g=1$, $\alpha = 0$, $a=1$, and $c=0$).
	

	By Theorems~\ref{th:exp-growth-upper},~\ref{th:exp-growth-lower},~\ref{th:double-exp-shrinking-upper},~and~\ref{th:double-exp-shrinking-lower}, the expected rumor spreading time is $\log_2 n + \log_2 \ln n \pm O(1)$.
\end{proof}

\subsection{Push-Pull Protocol}\label{sec:push-pull}

In the push-pull protocol, both informed and uninformed nodes contact a random neighbor in each round. If one of the two partners of such a conversation is informed, then also the other one becomes informed. The push-pull protocol is popular for a number of reasons.

The push-pull protocol (called \emph{anti-entropy} there) was found to be very reliable in the first experimental work on epidemic algorithms~\cite{Demers87}. The seminal paper by Karp et al.~\cite{KarpSSV00} proved that the push-pull protocol disseminates a rumor in a complete graph in $\log_3 n \pm O(\log\log n)$ rounds with high probability. This not only is faster than the push and pull protocols, but it allows implementations using only few messages per node. The just mentioned rumor spreading time stems from an exponential growths phase of length roughly $\log_3 n$ and a double exponential shrinking phase. Hence by making informed nodes stop their activity after the exponential growth phase, the total number of messages can be reduced massively.

The push-pull protocol was also investigated in models for social networks. Clearly, when modeling human communication, say people randomly meeting at parties and chatting, a push-pull spreading mechanism makes sense. However, also from the algorithmic viewpoint, it was observed that in graphs with a non-concentrated degree distribution the push-pull protocol greatly outperforms the push and pull protocols. This was first made precise by Chierichetti, Latanzi, and Panconesi~\cite{ChierichettiLP09}, who showed that the push-pull protocol spreads a rumor in a preferential attachment graph~\cite{BarabasiA99,BollobasR03} in time $O(\log^2 n)$, whereas both the push and the pull protocols need time $\Omega(n^\alpha)$ for some constant $\alpha > 0$ to inform all nodes. The precise rumor spreading time of $\Theta(\log n)$ of the push-pull protocol was shown in~\cite{DoerrFF11} (see also~\cite{DoerrFF12acm}). There is was also proven that the rumor spreading time reduces to $\Theta(\frac{\log n}{\log\log n})$ when the communication partners are chosen randomly but with the previous partner excluded. This first sublogarithmic rumor spreading time was quickly followed up by other fast rumor spreading times in networks modeling social networks, e.g.,~\cite{FountoulakisPS12,DoerrFF12,MehrabianP14}.

The push-pull protocol also performs well and admits strong theoretical analyses when the network has certain general expansion properties like a good vertex expansion~\cite{GiakkoupisS12,Giakkoupis14} or a low conductance~\cite{MoskAoyamaS06,ChierichettiLP10stoc,Giakkoupis11}.

\begin{theorem}\label{th:example-push-n-pull}
	The expected rumor spreading time of the push-pull protocol on the complete graph with $n$ vertices is $\log_3n + \log_2\ln n \pm O(1)$.
\end{theorem}

\begin{proof}
  We again discuss the covariance condition first. Consider one round of the protocol. Let $x_1$, $x_2$ be two different uninformed nodes. For $i = 1,2$, let $X_i$ be the indicator random variable for the event that $x_i$ becomes informed in this round, $Y_i$ the indicator random variable for the event that $x_i$ is called by an informed node, and $Z_i$ the indicator random variable for event that $x_i$ calls an informed node. Clearly, $X_i = \max\{Z_i,Y_i\}$.

  We show $\Cov[X_1,X_2] \le 0$, and thus all covariance conditions, by showing that $\Pr[X_1=1\mid X_2=1] \le \Pr[X_1=1]$. We have
    \begin{align}
        \Pr[X_1=&1 \mid X_2=1]
        = \Pr[X_1=1 \mid X_2=1 \AND Z_2=1] \cdot \Pr[Z_2=1 \mid X_2=1] \notag \\
        & + \Pr[X_1=1 \mid X_2=1 \AND Z_2=0] \cdot \Pr[Z_2=0 \mid X_2=1] \label{eq:push-pull-1}.
    \end{align}
    Since the intersection of events $Z_2=1 \AND X_2=1$ is equivalent to the single event $Z_2=1$ and the outgoing call of the uninformed node cannot inform any node, we have
    \begin{equation}
        \Pr[X_1=1 \mid X_2=1  \AND Z_2=1] = \Pr[X_1=1 \mid Z_2=1] = \Pr[X_1=1]. \label{eq:push-pull-2}
    \end{equation}
    When $Z_2=0 \AND X_2=1$ holds, then $x_2$ becomes informed via a push call, which is not available anymore to inform $x_1$. Hence
    \begin{equation}
        \Pr[X_1=1 \mid Z_2=0 \AND X_2=1] \le \Pr[X_1=1]. \label{eq:push-pull-3}
    \end{equation}
    From~\eqref{eq:push-pull-1} to~\eqref{eq:push-pull-3} we obtain $\Pr[X_1=1 \mid X_2=1] \le \Pr[X_1=1]$.

	An uninformed node remains uninformed if it is not called by any informed node and it calls an uninformed node itself. Hence $p_k = 1 - \left(1-\tfrac1n\right)^k\cdot\tfrac{n-k}{n}$.
    Using the estimates from Lemma~\ref{lem:prelim:(1-1/n)^k} we obtain
    \[
    	2\tfrac{k}{n} - \tfrac{3k^2}{2n^2} \le p_k \le 2\tfrac{k}{n}
    \]
    and see that the protocol satisfies the exponential growth conditions with $\gamma_n = 2$.

    Likewise, the probability $1-p_{n-u}$ that an uninformed node stays uninformed in a round starting with $u$ uninformed nodes is equal to $\tfrac{u}{n} \left(1-\tfrac1n\right)^{n-u}$.
    With Corollary~\ref{cor:prelim:(1-1/n)^(n-u)}, we estimate \[\tfrac1e\cdot\tfrac{u}{n} \le 1-p_{n-u} \le \tfrac{u}{n}.\]
    Therefore, the protocol satisfies the double exponential shrinking conditions with $\ell = 2$.

    By Theorems~\ref{th:exp-growth-upper},~\ref{th:exp-growth-lower},~\ref{th:double-exp-shrinking-upper},~and~\ref{th:double-exp-shrinking-lower}, the expected rumor spreading time is $\log_3 n + \log_2 \ln n \pm O(1)$.
\end{proof}

\section{Robustness, Multiple Calls, and Dynamic Graphs}\label{sec:more examples}

In this section, we apply our analysis method to settings (i)~in which calls fail independently with constant probability, (ii)~in which nodes are allowed to call a random number of other nodes instead of one as proposed in~\cite{PanagiotouPS15}, and (iii)~to a simple dynamic graph setting.

\subsection{Transmission Failures}

One key selling point for randomized rumor spreading, and more generally gossip-based algorithms, is that all these algorithms due to the intensive use of independent randomness are highly robust against all types of failures. In this subsection, we analyze the performance of the three classic protocols in the presence of independent transmission failures, that is, when calls are successful only with probability $p < 1$. Not unexpectedly, we can show that the rumor spreading times only increase by constant factors. However, we also observe a structural change, namely that the extremely fast double exponential shrinking previously seen with the pull and push-pull protocols is replaces by the slower single exponential shrinking regime. This has the important implication that the message complexity of the simple push-pull protocol (where messages are counted as in~\cite{KarpSSV00} and the protocol is assumed to stop when a suitable time limit is reached) increases from the theoretically optimal value of $\Theta(n \log\log n)$ to $\Theta(n \log n)$, see the remark following the proof of Theorem~\ref{th:example-push-n-pull-w-failures}.

While the robustness of randomized rumor spreading is consistently emphasized in the literature, only relatively few proven guarantees for this phenomenon exist. All results model communication failures by assuming that each call independently with probability $1-p$ fails to reach its target. The usual assumption is that the protocol does not take notice of such events. Els\"asser and Sauerwald~\cite{ElsasserS09} show for any graph $G$ that if the push protocol spreads a rumor with probability $1-O(1/n)$ to all nodes in time $T$, then the push protocol with failures succeeds in informing all nodes with probability $1 - O(1/n)$ in time $\tfrac 6p T$. This was made more precise for complete graphs in~\cite{DoerrHL13}, for which a rumor spreading time of $\log_{1+p} + \tfrac 1p n \pm o(\log n)$ was shown to hold with high probability. The same result also holds for random graphs in the $G(n,p')$ model when the edge probability $p'$ is $\omega(\log(n)/n)$, that is, asymptotically larger than the connectivity threshold~\cite{FountoulakisHP10}. To the best of our knowledge, these few results are all that is known in terms of proven guarantees for the classic rumor spreading protocols in the presence of failures.

We now use the methods developed in this work to obtain very sharp estimates for the runtimes of the classic protocols on complete graphs when calls fail independently with probability $1-p$, $p < 1$. As in Sections~\ref{sec:classics}, the growth or shrinking conditions valid in each case are easily proven, showing again the versatility of our approach.


\begin{theorem}\label{th:example-push-w-failures}
    The expected rumor spreading time for the push protocol with success probability $p$ on the complete graph of size $n$ is equal to
    \[
        \log_{1+p} n + \tfrac1p\ln n \pm O(1).
    \]
\end{theorem}
\begin{proof}
	With the same argument as in the proof of Theorem~\ref{th:example-push}, we see that the covariances regarded in the covariance conditions are all negative.
	
	Consider an uninformed node in a round started with $k$ informed nodes. The probability that it becomes informed in this round is $p_k = 1 - (1-\tfrac{p}{n})^k$. By Lemma~\ref{lem:prelim:(1-1/n)^k}, we estimate
	\[
		\tfrac{pk}{n} - \tfrac{p^2k^2}{2n^2} \le p_k \le \tfrac{pk}{n}
	\]
	for all $k < n$ and see that the protocol satisfies the exponential growth conditions in $[1,n[$ with $\gamma_n = p$.
	
	Similarly, the probability that an uninformed node in a round starting with $u := n-k$ uninformed nodes stays uninformed, is $1-p_{n-u} = \left(1-\tfrac{p}{n}\right)^{n-u}$. By Corollary~\ref{cor:prelim:(1-p/n)^(n-u)}, we estimate
	\[
		e^{-p} \le 1-p_{n-u} \le e^{-p} (1+\tfrac{2pu}{n})
	\]
	for all $u < n$ and thus have the exponential shrinking conditions with ${\rho_n} = p$ for all $u \le n/2$.
	
	By Theorems~\ref{th:exp-growth-upper},~\ref{th:exp-growth-lower},~\ref{th:exp-shrinking-upper},~and~\ref{th:exp-shrinking-lower}, the expected rumor spreading time is $\log_{1+p} n + \tfrac1p\log n \pm O(1)$.
\end{proof}

The result above and its proof are valid for $p=1$ and then coincide with Theorem~\ref{th:example-push}. For the pull protocol and the push-pull protocol, we observe a substantial change of the process when transmission errors occur. In this case, an uninformed node stays uninformed with probability at least $1-p$, so the double exponential shrinking conditions cannot be satisfied. Instead, we observe that the single exponential shrinking conditions are satisfied.

\begin{theorem}\label{th:example-pull-w-failures}
	The expected rumor spreading time of the pull protocol with success probability $p<1$ on the complete graph of size $n$ is equal to
	\begin{equation}
		\log_{1+p}n + \tfrac1{\ln \frac{1}{1-p}} \ln n \pm O(1) \notag.
	\end{equation}
\end{theorem}

\begin{proof}
	As in the proof of Theorem~\ref{th:example-pull}, the events that uninformed nodes become informed are mutually independent. Hence all covariance conditions are satisfied with $c=0$. The probability that an uninformed node becomes informed in a round starting with $k$ informed nodes is $p_k = p\frac kn$, hence the protocol satisfies the exponential growth conditions in $[1,n[$ with $\gamma_n = p$.
	
	Similarly, the probability that an uninformed node remains uninformed in a round starting with $u$ uninformed nodes is \[1 - p_{n-u} = 1 - p\tfrac{n-u}{n} = 1 - p + p \tfrac un = \exp(-\ln \tfrac 1 {1-p}) + p \tfrac un.\] Consequently, the protocol satisfies the exponential shrinking conditions with ${\rho_n} = \ln\tfrac1{1-p}$ for all $u \le gn$, $g$ any constant smaller than $1$.
	
	By Theorems~\ref{th:exp-growth-upper},~\ref{th:exp-growth-lower},~\ref{th:exp-shrinking-upper},~and~\ref{th:exp-shrinking-lower}, the expected rumor spreading time is $\log_{1+p} n + \tfrac1{\ln(1/(1-p))} \ln n \pm O(1)$.
\end{proof}

\begin{theorem}\label{th:example-push-n-pull-w-failures}
	The expected rumor spreading time for the push-pull protocol with success probability $p<1$ on the complete graph of size $n$ is equal to
	\begin{equation}
		\log_{2p+1}n + \tfrac1{p + \ln \frac{1}{1-p}} \ln n \pm O(1) \notag.
	\end{equation}
\end{theorem}
\begin{proof}
	Using the same arguments as for the push-pull protocol without failures, we observe that the covariances are at most zero, so all covariance conditions are satisfied.
	Consider an uninformed node in a round starting with $k$ informed nodes.
	The probability that this node does not inform itself via its pull call is $1-p\tfrac{k}{n}$.
	The probability that it is not successfully called by an informed node is $\left(1-\tfrac{p}{n}\right)^k$.
	Hence $p_k = 1 - \left(1-p\tfrac{k}{n}\right) \left(1-\tfrac{p}{n}\right)^k$ and Corollary~\ref{cor:prelim:(1-p/n)^k} gives
	\[
		2p\tfrac{k}{n} - \tfrac{3p^2k^2}{2n^2} \le p_k \le 2p\tfrac{k}{n}.
	\]
	Thus the protocol satisfies the exponential growth conditions in $[1,\tfrac 23n[$ with $\gamma_n = 2p$.

    Likewise, the probability $1-p_{n-u}$ that an uninformed node stays uninformed in a round starting with $u$ uninformed nodes is equal to
    $\left(1-p\tfrac{n-u}{n}\right) \left(1-\tfrac{p}{n}\right)^{n-u}$.
    With Corollary~\ref{cor:prelim:(1-p/n)^(n-u)} we estimate
    \[
        (1-p)e^{-p} + pe^{-p}\cdot\tfrac{u}{n}
        \le 1-p_{n-u}
        \le (1-p)e^{-p} + 3pe^{-p}\cdot\tfrac{u}{n}.
    \]
	Therefore, the protocol satisfies the exponential growth conditions with ${\rho_n} = p + \ln\tfrac1{1-p}$.
    Thus by Theorems~\ref{th:exp-growth-upper},~\ref{th:exp-growth-lower},~\ref{th:exp-shrinking-upper},~and~\ref{th:exp-shrinking-lower}, the expected spreading time is equal to
    $\log_{p+1}n + \tfrac1{p + \ln(1/(1-p))} \ln n \pm O(1)$.
\end{proof}

The fact that in the presence of transmission failures the double exponential shrinking regime ceases to exist has an important implication on the message complexity. In their seminal paper~\cite{KarpSSV00}, Karp et al.\ show that any address-oblivious rumor spreading algorithm that informs all nodes of the complete graph with at least constant probability needs $\Omega(n \log\log n)$ message transmissions in expectation (we refer to that paper for a discussion of the tricky question how to count messages in algorithms performing pull calls).

This optimal order of magnitude is attained by the push-pull protocol when nodes stop sending a rumor that is older than $\log_3 n + O(\log\log n)$ rounds. As Karp et al.\ remark, relying on such a time stamp is risky. A mild underestimate of the true rumor spreading time leaves a constant fraction of the nodes uninformed. A mild overestimate of the rumor spreading time by $\eps \log n$ rounds leads to the situation that for $\eps \log n$ rounds a constant fraction of the nodes knows and pushes the rumor, which implies a message complexity of $\Omega(n \log n)$. For this reason, Karp et al.\ propose the more complicated median-counter algorithms which is robust against a moderate number of adversarial node failures and against moderate deviations from the uniform choice of the nodes to contact.

Our above analysis of the push-pull protocol in the presences of transmission faults shows that not only an unexpected deviation from the ideal fault-free push-pull protocol leads to an increased message complexity, but even a perfectly anticipated faulty behavior. While we know the expected rumor spreading time very precisely (and we could with the same arguments also show a tail bound stating that our upper bound for the expectation is exceeded by $\lambda$ with probability $\exp(-\Omega(\lambda))$ only), the ``transmit until time limit reached'' approach still leads to a message complexity of $\Omega(n \log n)$ due to the missing double exponential shrinking phase. As our analysis shows, after an expected number of $\log_{2p+1} n$ iterations, a constant fraction of the nodes are informed. However, it takes another $\tfrac1{p + \ln \frac{1}{1-p}} \ln n + O(1)$ rounds in the exponential shrinking regime until all nodes are informed. Hence when using the simple ``transmit until time limit reached'' approach to limit the number of messages, the exponential shrinking regime alone would see $\Omega(n \log n)$ push calls by the $\Omega(n)$ informed nodes.

It is not clear how to overcome this difficulty. The median-counter algorithm of Karp et al.\ for constant-probability transmission failures also seems to require $\Omega(n \log n)$ messages (see the comment right before Theorem~3.1 in~\cite{KarpSSV00}).

\subsection{Multiple Calls}

In this section, we analyze rumor spreading protocols in which in each round each node when active calls a random number $R$ of nodes. This was proposed by~\cite{PanagiotouPS15} to model different data processing speeds of nodes. Unlike in~\cite{PanagiotouPS15}, we assume that each node in each round resamples the number of nodes it may call. This allows to model changing data processing speed as opposed to nodes having generally different speeds.

Consider a random integer variable $R$ taking values in $[0,n[$.
We say that a rumor spreading protocol is an $R$-protocol if in each round it respects the following call procedure.
Each node which can make calls in current round samples independently a new value $r$ from $R$.
Then it calls $r$ different neighbors chosen uniformly at random.

In this section we consider the $R$-push protocol and the $R$-push-pull protocol and prove the statements similar to Theorem~1.1, 1.2, and 1.3 from~\cite{PanagiotouPS15}.
Note that by putting $R\equiv1$, we obtain the classic push and push-pull protocols.
\begin{theorem}\label{th:multiple-push}
    Assume that $R$ is a distribution with $\E[R] = \Theta(1)$ and $\Var[R]=O(1)$.
    Then the expected spreading time for the $R$-push protocol on the complete graph of size $n$ is equal to
    \[
        \log_{1+\E[R]}n + \tfrac1{\E[R]}\ln n \pm O(1).
    \]
\end{theorem}
\begin{proof}
    Consider a round of the protocol started from $k$ informed nodes.
    Let $x_1$ and $x_2$ be two different uninformed nodes and let $X_1$ and $X_2$ be the indicator random variables for events that $x_1$ resp. $x_2$ become informed.
    Suppose that node $y$ is informed.
    The probability that $x_1$ and $x_2$ are both called by $y$ is at most
    \[
        \sum_{j\ge2} \Pr[R=j] \cdot \binom{j}{2} \cdot \tfrac1{n(n-1)}
        \le \tfrac1{n^2} \sum_{j\ge2} j^2 \cdot \Pr[R=j]
        \le (\Var[R] + \E[R]^2) \cdot \tfrac1{n^2}
        = O\left(\tfrac1{n^2}\right).
    \]
    Since there are $k$ informed nodes, the probability that $x_1$, $x_2$ are both called by the same node (not necessary $y$) is $k\cdot O\left(\tfrac1{n^2}\right)$.
    In addition, if we condition on the event that $x_1$ and $x_2$ are not called by the same node, then the probability that they both get informed is slightly less than $p_k^2 = \Pr[X_1=1]^2$.
    Therefore, $\Cov[X_1,X_2] \le k\cdot O\left(\tfrac1{n^2}\right)$ for any $k < n$ which corresponds to the covariance condition for both exponential growth and exponential shrinking.

    Now let us study the probability $p_k$.
    Since the probability that $x$ does not belong to a random set of $j$ nodes is equal to
    $$\left(1-\tfrac1n\right)\left(1-\tfrac1{n-1}\right)\ldots\left(1-\tfrac1{n-j+1}\right) = \tfrac{n-j}{n},$$
    the probability that $y$ does not call $x$ is equal to $\sum_{j\ge0} \Pr[R=j] \cdot \tfrac{n-j}{n} = 1 - \tfrac{\E[R]}{n}$.
    Therefore the probability $p_k$ that $x$ gets informed in current round is equal to
    \begin{equation}
        1 - \left(1-\tfrac{\E[R]}{n}\right)^k. \label{eq:multiple-1}
    \end{equation}
    With Corollary~\ref{cor:prelim:(1-p/n)^k} we estimate
    \begin{equation}
        \E[R]\cdot\tfrac{k}{n} - \E[R]^2\cdot\tfrac{k^2}{2n^2}
        \le p_k
        \le \E[R]\cdot\tfrac{k}{n}, \label{eq:multiple-1}
    \end{equation}
    for any $k \le n/\E[R]$.
    Therefore, the protocol satisfies the exponential growth conditions in $[1,n/\E[R]]$ with $\gamma_n = \E[R]$.

    Similarly, the probability that an uninformed node stays uninformed in a round starting with $u:=n-k$ uninformed nodes, is $1-p_{n-u} = \left(1-\tfrac{\E[R]}{n}\right)^{n-u}$.
    By Corollary~\ref{cor:prelim:(1-p/n)^(n-u)}, for all $u \le n/\E[R]$ we estimate
    \begin{equation}
        e^{-\E[R]}
        \le 1-p_{n-u}
        \le e^{-\E[R]} \left(1+2\E[R]\tfrac{u}{n}\right). \label{eq:multiple-4}
    \end{equation}
    Therefore, the protocol satisfies the exponential shrinking conditions in $[n(1-1/\E[R]),n]$ with ${\rho_n} = \E[R]$.

    We note that the intervals for the exponential growth and shrinking regime does not intersect if $\E[R] > 2$.
    However, we still be able to bound the expected spreading time.
    From~\eqref{eq:multiple-1} it follows that $p_{n/\E[R]} = 1-\tfrac1e+o(1)$ and $p_{n(1-1/\E[R])} = 1 - e^{1-\E[R]} + o(1)$.
    Since $p_k$ increases, it is bounded uniformly for any $k \in \left[\tfrac{n}{\E[R]}, n-\tfrac{n}{\E[R]}\right]$.
    Hence, by Lemma~\ref{lem:general-connect}, we have
    $\E\left[T\left(\tfrac{\E[R]}{n},n-\tfrac{\E[R]}{n}\right)\right] = O(1)$.
    So by Theorems~\ref{th:exp-growth-upper}~and~\ref{th:exp-shrinking-upper}, the expected rumor spreading time is at most $\log_{1+\E[R]} n + \tfrac1{\E[R]}\log n \pm O(1)$.

    Similarly, by Lemma~\ref{lem:general-connect-lower}, there exists some $f' \in \left]1-\tfrac1{\E[R]},1\right[$ such that with probability $1-O\left(\tfrac1n\right)$ the number of informed nodes after some round will belong to $\left[n-\tfrac{n}{\E[R]},f'n\right]$.
    Then by Theorems~\ref{th:exp-growth-lower}~and~\ref{th:exp-shrinking-lower}, the expected rumor spreading time is at least $\log_{1+\E[R]} n + \tfrac1{\E[R]}\log n \pm O(1)$.
\end{proof}

\begin{theorem}\label{th:multiple-push-pull}
    Assume that $R$ is a distribution with $\E[R] = \Theta(1)$ and $\Var[R]=O(1)$.
    Let $\ell$ be the smallest nonnegative integer such that $\Pr[R=\ell] > 0$ and we suppose that $\Pr[R=\ell] = \Theta(1)$.
    Then the expected spreading time for the $R$-push-pull protocol on the complete graph of size $n$ is at most
    \begin{align*}
        & \log_{1+2\E[R]}n + \tfrac1{\E[R]-\ln\Pr[R=0]} \cdot \ln n \pm O(1), & \ell = 0; \\
        & \log_{1+2\E[R]}n + \log_{1+\ell}\ln n \pm O(1), & \ell > 0.
    \end{align*}
\end{theorem}
\begin{proof}
	As usual, we discuss the covariance condition first.
	Consider one round of the protocol started from $k$ informed nodes.
	Let $x_1$, $x_2$ be two different uninformed nodes.
	For $i = 1,2$, let $X_i$ be the indicator random variables for event that $x_i$ becomes informed in this round, $Y_i$ the indicator random variable for the event that $x_i$ is called by an informed node, and $Z_i$ the indicator random variable for event that $x_i$ calls an informed node.
	Since $Y_i$ coincides with $X_i$ for the push protocol from the proof of Theorem~\ref{th:multiple-push}, we have $\Cov[Y_1,Y_2] \le k \cdot O\left(\tfrac1{n^2}\right)$.
	In addition $Z_i$ are pairwise independent and also independent from $Y_i$.
	Since $X_i = \max\{Z_i,Y_i\}$ we have also $\Cov[X_1,X_2] \le k \cdot O\left(\tfrac1{n^2}\right)$ for any $k < n$.
	Therefore, the covariance condition is satisfied for exponential growth and both exponential and double exponential shrinking conditions.
	
	Let us study $\Pr[Z_1=0]$.
	If node $x_1$ calls $j$ different nodes in current round, then the probability that it does not hit informed node is $\left(1-\tfrac{k}{n}\right)\ldots\left(1-\tfrac{k}{n-j+1}\right)$.
	Summing over all possible values of $j$ we obtain the following.
	\begin{equation}
		\Pr[Z_1=0] = \sum_{j=0}^{n-k} \Pr[R=j] \cdot \left(1-\tfrac{k}{n}\right)\ldots\left(1-\tfrac{k}{n-j+1}\right). \label{eq:multiple-3}
	\end{equation}
	Recall that that $\sum_{j=0}^n j\cdot\Pr[R=j] = \E[R]$ and $\sum_{j=0}^n j^2\cdot\Pr[R=j] = \Var[R] + \E[R]^2 = O(1)$.
	Using estimate from Corollary~\ref{cor:prelim:(1-p/n)^k}, we compute for any $k\le\tfrac{n}{2}$
	\begin{align*}
		\Pr[Z_1=0]
		& \le \sum_{j=0}^{n-k} \Pr[R=j] \cdot \left(1-\tfrac{k}{n}\right)^j \\
		& \le \sum_{j=0}^{n/k}\Pr[R=j] \cdot \left(1-j\tfrac{k}{n}+j^2\tfrac{k^2}{2n^2}\right)
			+ \sum_{j=n/k+1}^{n-k} \Pr[R=j] \\
		& = \sum_{j=0}^{n/k}\Pr[R=j] - \tfrac{k}{n}\sum_{j=0}^{n/k}j\cdot\Pr[R=j]
			+ \tfrac{k^2}{2n^2}\sum_{j=0}^{n/k}j^2\cdot\Pr[R=j]
			+ \sum_{j=n/k-1}^{n-k}\Pr[R=j] \\
		& \le 1 - \tfrac{k}{n}\left(\E[R]-\sum_{j=n/k-1}^n j\cdot\Pr[R=j]\right)
			+ \tfrac{k^2}{n^2}\sum_{j=0}^{n-k}j^2\cdot\Pr[R=j] \\
		& \le 1 - \E[R] \cdot \tfrac{k}{n} + \tfrac{k^2}{n^2} \sum_{j=n/k-1}^n j^2\cdot\Pr[R=j]
			+ \tfrac{k^2}{n^2}\sum_{j=0}^{n-k}j^2\cdot\Pr[R=j] \\
		& \le 1 - \E[R] \cdot \tfrac{k}{n} + 2(\Var[R]+\E[R]^2) \cdot \tfrac{k^2}{n^2}.
	\end{align*}
	For any $k\le\tfrac{n}{2}$ we can similarly bound $\Pr[Z_i=0]$ from below using Bernoulli's inequality.
	\begin{align*}
		\Pr[Z_1=0]
		& \ge \sum_{j=0}^{n-k} \Pr[R=j]\left(1-k\cdot\tfrac{j}{n-j}\right) \\
		& \ge \sum_{j=0}^{n-k} \Pr[R=j]\left(1-\tfrac{jk}{n}\left(1+2\tfrac{j}{n}\right)\right) \\
		& = 1 - \E[R]\cdot \tfrac{k}{n} + O(1) \cdot \tfrac{k^2}{n^2}
	\end{align*}
	By \eqref{eq:multiple-1}, we estimate $\Pr[Y_1=0] = 1-\E[R]\cdot\tfrac{k}{n} \pm O(1) \cdot \tfrac{k^2}{n^2}$.
	Since $Y_1$ and $Z_1$ are independent, we have
	$$\Pr[X_1=1] = 1-\Pr[Y_1=0]\cdot\Pr[Z_1=0].$$
	Therefore, $p_k = 2\E[R]\cdot\tfrac{k}{n} \pm O(1) \cdot \tfrac{k^2}{n^2}$ for any $k \le \min\left\{\tfrac{n}{2}, \tfrac{n}{\E[R]}\right\}$.
	Hence the protocol satisfies the exponential growth conditions with $\gamma_n = 2\E[R]$ for any $k \le \min\left\{\tfrac{n}{2}, \tfrac{n}{\E[R]}\right\}$.
	
	Now we discuss the shrinking conditions.
	We consider a round started from $u:=n-k$ uninformed nodes.
	Similarly to~\eqref{eq:multiple-3}, we have
	\[
		\Pr[Z_1=0] = \sum_{j\ge0} \Pr[R=j] \cdot \tfrac{u}{n} \cdot \tfrac{u-1}{n-1} \cdot \ldots \cdot \tfrac{u-j+1}{n-j+1}.
	\]
	Assume first that $\Pr[R=0] > 0$, i.e., $\ell=0$.
	Since $x_1$ might not call in current round, there is at least a constant probability, that it stays uninformed.
	With~\eqref{eq:multiple-4} and estimate
	$$\Pr[R=0] \le \Pr[Z_1=0] \le \Pr[R=0] + \Pr[R\ge1] \cdot \tfrac{u}{n},$$
	we see that $\Pr[X_1=0] = \Pr[R=0]\cdot e^{-\E[R]} \pm O(1)\cdot\tfrac{u}{n}$ for any $u \le \min\left\{\tfrac{n}{2}, \tfrac{n}{\E[R]}\right\}$.
	In this case the protocol satisfies the exponential shrinking conditions with ${\rho_n} = \E[R] - \ln\Pr[R=0]$.
	Applying Lemma~\ref{lem:general-connect}~and~\ref{lem:general-connect-lower} in the similar way as in the proof of Theorem~\ref{th:multiple-push}, one can see that by Theorems~\ref{th:exp-growth-upper},~\ref{th:exp-growth-lower},~\ref{th:exp-shrinking-upper},~and~\ref{th:exp-shrinking-lower}, the expected rumor spreading time is $\log_{1+2\E[R]}n + \tfrac1{\E[R]-\ln\Pr[R=0]}\ln n \pm O(1)$.
	
	Finally, suppose that $\Pr[R=0] = 0$, and let $\ell$ be the smallest integer such that $\Pr[R=\ell] > 0$.
	In this case we can easily estimate the probability that $x_1$ stays uninformed.
	From below we have
	\[
		\Pr[X_1=0] \ge \Pr[Y_1=0] \cdot \Pr[R=\ell] \cdot \tfrac{u^\ell}{n^\ell}
		\ge e^{-\E[R]} \cdot \Pr[R=\ell] \cdot \tfrac{u^\ell}{n^\ell}.
	\]
	From above, $\Pr[X_1=0] \le \Pr[Z_1=0] \le \tfrac{u^\ell}{n^\ell}$.
	Hence the protocol satisfies the double exponential shrinking conditions with parameter $1+\ell$.
	Again, by Theorems~\ref{th:exp-growth-upper},~\ref{th:exp-growth-lower},~\ref{th:double-exp-shrinking-upper},~and~\ref{th:double-exp-shrinking-lower} and Lemmas~\ref{lem:general-connect}~and~\ref{lem:general-connect-lower}, the expected rumor spreading time is $\log_{1+2\E[R]}n + \log_{1+\ell}\ln n \pm O(1)$.
\end{proof}

\input{Erdos_Renyi.tex}

%% file: Erdos_Renyi.tex
\subsection{Dynamic Graphs}\label{sec:dynamic graphs}

We now show that our method can also be applied to certain dynamic graph settings, that is, when the network structure may be different in each round. While it is generally agreed upon that dynamic problem settings are highly relevant for practical applications, it is still not so clear what is a good theoretical model for dynamicity. For rumor spreading problems, the only work regarding dynamic graphs~\cite{ClementiCDFPS16} considers the two models (i) that in each round independently the network is a $G(n,p)$ random graph and (ii)~that each possible edge has its own independent two-state Markov chain describing how it changes between being present and not (edge-Markovian dynamic graphs). For both models, it is proven that the push protocol informs all nodes in logarithmic time with high probability (when the parameters are chosen reasonably).

It is clear that the edge-Markovian model due to the time-dependence cannot be analyzed with our methods. For the other result, we now show that our method quite easily gives a very precise analysis. We only treat the case of $\Theta(1/n)$ edge probabilities, as this seems to be the most interesting one (the graph is not connected, but has nodes with degrees varying between $0$ and $\Theta(\log(n)/\log\log(n))$; when $p \ge (1+\eps)/n$, a giant component encompassing a linear number of nodes exists). 

To make the model precise, we assume that in each round independently, before the communication starts, the communication graph is sampled as $G(n,p)$ random graph, where $p = a/n$ for some positive constant $a$. That is, between any two nodes there is an edge, independently, with probability $a/n$. In the communication part of the round, each informed node chooses a communication partner uniformly at random from its neighbors in the communication graph and sends a copy of the rumor to it. Isolated informed nodes, naturally, do not communicate in this round.

We introduce the following notation. We consider one round and aim at showing the exponential growth and shrinking conditions. Let $E$ be the set of edges of the communication graph $G(n,\tfrac an)$ of this round. We write $xy \in E$ as shorthand for $\{x,y\} \in E$. We write $x \to y$ to denote the event that $x$ calls $y$. By $\deg_{\inf} x$ we denote the number of informed neighbors of $x$.



\begin{lemma}\label{lemma: Erdos-Renyi - call probability}
	Consider an uninformed node $x$ and an informed node $y$.
	Let $\ell \le n/2$ and let $A_\ell$ be the event that $\{y_1y, \ldots, y_\ell y\} \cap E = \emptyset$.
	Then
	$$\Pr[y \to x \mid xy \in E \AND A_\ell] = \tfrac{1-e^{-a}}{a} + (\ell+1) \cdot O\left(\tfrac1n\right).$$
\end{lemma}
\begin{proof}
    Assume that $xy \in E$. Then the number of other neighbors of $y$, that is,  the random variable $\deg y - 1$, has a binomial distribution with parameters $n-2-\ell$ and $\tfrac{a}{n}$.
    The probability that $y$ calls $x$ is equal to $\tfrac1{\deg y}$.
	Using the fact that $\binom{m+1}{k+1} = \tfrac{k+1}{m+1} \binom{m}{k}$, we compute
	\begin{align*}
		\Pr[y \to x & \mid xy \in E \AND A_\ell]
		= \sum_{i=0}^{n-2-\ell} \tfrac1{i+1} \binom{n-2-\ell}{i} \left(\tfrac an\right)^i \left(1-\tfrac an\right)^{n-2-\ell-i} \\
		& = \tfrac na \cdot \tfrac1{n-2-\ell+1}
			\cdot \sum_{i=0}^{n-2-\ell} \binom{n-2-\ell+1}{i+1} \left(\tfrac an\right)^{i+1} \left(1-\tfrac an\right)^{n-2-\ell+1-(i+1)} \\
		& = \tfrac1a \cdot \left(1-\tfrac{\ell+1}{n-\ell-1}\right)
			\cdot \left(1 - \Pr[\Bin(n-2-\ell+1,\tfrac an)=0]\right) \\
		& = \tfrac1a \cdot \left(1-\tfrac{\ell+1}{n-\ell-1}\right)
			\cdot \left(1 - \left(1-\tfrac an\right)^{n-\ell-1}\right) \\
		& = \tfrac{1-e^{-a}}{a} + (\ell+1) \cdot O\left(\tfrac 1n\right),
	\end{align*}
	where above we denoted by $\Bin(m,p)$ a random variable having a binomial distribution with parameters $m$ and $p$.
\end{proof}

\begin{lemma} \label{lemma: Erdos-Renyi - growth and shrinking}
	Consider one round starting with $k<n$ informed nodes.
	The probability $1-p_k$ that an uninformed node $x$ stays uninformed in this round is at most
	$(1 - \tfrac{1-e^{-a}}{n})^k + k \cdot O(\tfrac1{n^2})$.
\end{lemma}

\begin{proof}
	Let $A$ be the event that $G\left(n,\tfrac an\right)$ contains no triangle formed by $x$ and two other informed nodes.
	By the first moment method, $\Pr[A] \ge 1 - k^2\cdot\tfrac{a^3}{n^3}$. Let $X$ be the indicator random variable for the event that $x$ is called by an informed node. Then
	\begin{align*}
		\Pr[X=0] \le \Pr[\NOT A] + \Pr[X=0 \AND A] \le k^2\tfrac{a^3}{n^3} + \Pr[X=0 \AND A].
	\end{align*}
	We compute $\Pr[X=0 \AND A]$ by conditioning on $\deg_{\inf} x$, which has a binomial distribution with parameters $k$ and $\tfrac an$.
	In addition, we observe that the conditioning on $A$ makes the actions of the informed neighbors of $x$ independent (in the probability space composed of the random actions of the nodes and the not yet determined random edges). Hence
	\[
		\Pr[X=0 \mid \deg_{\inf}x = \ell \AND A]
		= \left(1-\Pr[y \to x \mid xy \in E \AND A_{\ell-1}]\right)^\ell
		\le \left(1-\tfrac{1-e^{-a}}{a} + O\left(\tfrac1n\right)\right)^\ell
	\]
	by Lemma~\ref{lemma: Erdos-Renyi - call probability}.
We compute.
	\begin{align*}
		\Pr[X=0 \AND A]
		& = \sum_{\ell=0}^k \Pr[\deg_{inf} x = \ell] \cdot \Pr[A \mid \deg_{inf} x = \ell] \cdot \Pr[X=0 \mid \deg_{\inf}x = \ell \AND A] \\
		& \le \sum_{\ell=0}^k \binom kl \left(\tfrac an\right)^\ell \left(1-\tfrac an\right)^{k-\ell}
			\cdot 1 
			\cdot \left(1-\tfrac{1-e^{-a}}{a} + O\left(\tfrac1n\right)\right)^\ell \\
		& \le \left[\tfrac an\left(1-\tfrac{1-e^{-a}}{a} + O\left(\tfrac1n\right)\right)
			+ 1 - \tfrac an\right]^k \\
		& = \left(1 - \tfrac{1-e^{-a}}{n}\right)^k + k \cdot O\left(\tfrac1{n^2}\right).
	\end{align*}
%
%
\end{proof}

\begin{lemma} \label{lemma: Erdos-Renyi - growth}
	Consider one round starting with $k < n$ informed nodes. The probability $p_k$ that an uninformed node $x$ becomes informed in the current round is at most
	$\tfrac kn \cdot \left(1-e^{-a} + O\left(\tfrac1n\right)\right)$.
\end{lemma}
\begin{proof}
	Consider an uninformed node $x$ and an informed node $y$.
	Applying Lemma~\ref{lemma: Erdos-Renyi - call probability} with $\ell = 0$, we compute
	\[
		\Pr[y \to x] = \Pr[xy \in E] \cdot \Pr[y \to x \mid xy \in E]
		= \tfrac an \cdot \left(\tfrac{1-e^{-a}}{a} + O\left(\tfrac1n\right)\right).
	\]
	A union bound over the $k$ informed nodes proves the claim.
\end{proof}

\begin{lemma}\label{lemma: Erdos-Renyi - shrinking}
	Consider one round starting with $k = \Omega(n)$ informed nodes.
	The probability $1-p_k$ that an uninformed node $x$ stays uninformed in current round is at least
	$\left(1 - \tfrac{1-e^{-a}}{n}\right)^k - O\left(\tfrac{\log^2 n}{n}\right)$.
\end{lemma}
\begin{proof}
	Let again $A$ denote the event that $G\left(n,\tfrac an\right)$ contains no cycle of length 3 formed by $x$ and two other informed nodes, and let $X$ be the indicator random variable for the event that $x$ becomes informed.
	Then $\Pr[X=0] \ge \Pr[X=0 \AND A]$.
	Similar to the proof of Lemma~\ref{lemma: Erdos-Renyi - growth and shrinking}, we compute $\Pr[X=0]$ by conditioning on the number $\deg_{\inf} x$ of its informed neighbors.
	\begin{align*}
		\Pr[X=0 \AND A]
		& = \sum_{\ell=0}^k \Pr[\deg_{\inf} x = \ell]
			\cdot \Pr[A \mid \deg_{\inf} x = \ell] \cdot \Pr[X=0 \mid \deg_{\inf}x =\ell \AND A] \\
		& = \sum_{\ell=0}^{k} \binom kl \left(\tfrac an\right)^\ell \left(1-\tfrac an\right)^{k-\ell}
			\cdot \left(1-\tfrac an\right)^{\ell^2}
			\cdot \left(1 - \tfrac{1-e^{-a}}{a} - (\ell+1) \cdot O\left(\tfrac 1n\right)\right)^\ell
	\end{align*}
	To simplify the notation, we denote
	$x_\ell := \binom kl \left(\tfrac an\right)^\ell \left(1-\tfrac an\right)^{k-\ell}$ and $q := 1 - \tfrac{1-e^{-a}}{a}$.
	Then
		\begin{align*}
		\Pr[X=0 \AND A]
		& \ge \sum_{\ell=0}^{c\log n} x_\ell
			\cdot \left(1-\tfrac an\right)^{\ell^2}
			\cdot \left(q - \ell \cdot O\left(\tfrac 1n\right)\right)^\ell \\
		& \ge \sum_{\ell=0}^{c\log n} x_\ell \cdot \left(1-\tfrac an\right)^{c^2\log^2n} \left(q-O\left(\tfrac{\log n}{n}\right)\right)^\ell \\
		& \ge \left(1-O\left(\tfrac{\log^2n}{n}\right)\right)\sum_{\ell=0}^{c\log n} x_\ell q^\ell.
	\end{align*}
	By Lemma~\ref{lemma: log degree}, there exists $c > 0$ such that $\sum_{\ell=c\log n}^{k} x_\ell q^\ell \le \tfrac1n$.
	Since $\sum_{\ell=0}^{k} x_\ell q^\ell = \left(1-\tfrac{1-e^{-a}}{n}\right)^k$, we have 
	\begin{align*}
		\Pr[X=0 \AND A] \ge \left(1-O\left(\tfrac{\log^2n}{n}\right)\right) \left(1-\tfrac{1-e^{-a}}{n}\right)^k.
	\end{align*}
\end{proof}

\begin{lemma} \label{lemma: Erdos-Renyi - covariance}
	Consider a round starting with $k$ informed nodes. Let $x_1$ and $x_2$ be two uninformed nodes.
	Then the corresponding random indicator variables $X_1$ and $X_2$ for the events of these becoming informed are negatively correlated.
\end{lemma}
\begin{proof}
	By symmetry, we can assume that in this round we first generate the random communication graph, then we let each node choose a potential communication partner (uniformly among its neighbors), and then we decide randomly which $k$ nodes are informed, and finally those nodes which are informed actually call the potential partner chosen before. In this joint probability space, let $x_1$ and $x_2$ be two nodes. We condition in the following on (i) the outcome of the random graph, (ii) the outcome of the potential communication partners, and (iii) $x_1$ and $x_2$ being uninformed. In other words, all randomness is already decided except which set $I$ of $k$ nodes different from $x_1$ and $x_2$ is informed. 

	Let $S_1$ and $S_2$ be the sets of nodes having chosen $x_1$ and $x_2$ as potential partner. Now we have $X_1=1$ if and only if $S_1 \cap I \ne \emptyset$.	Similarly, $X_2 = 1$ is equivalent to $S_2 \cap I \ne \emptyset$.
	Since $S_1 \cap S_2 = \emptyset$ by construction, $X_1$ and $X_2$ are negatively correlated.
\end{proof}

\begin{theorem}
	The expected rumor spreading time is $\log_{2-e^{-a}}n + \tfrac1{1-e^{-a}}\ln n \pm O(1)$.
	In addition, there are constant $A' \alpha' > 0$ such that for any $r \in \N$ we have $\Pr[|T - \E[T]| \ge r] \le A' e^{-\alpha' r}$.
\end{theorem}
\begin{proof}
	By Lemma~\ref{lemma: Erdos-Renyi - covariance}, the covariance conditions are satisfied for both exponential growth and exponential shrinking.

	From Lemma~\ref{lemma: Erdos-Renyi - growth and shrinking} together with Corollary~\ref{cor:prelim:(1-p/n)^k} it follows that for any $k < n$ we have
	$$p_k \ge \tfrac kn \left(1-e^{-a}\right) - \tfrac{k^2}{2n^2}\left(1-e^{-a}\right)^2 - k\cdot O\left(\tfrac1{n^2}\right).$$
	Combining this with Lemma~\ref{lemma: Erdos-Renyi - growth}, we see that the process satisfies the exponential growth conditions with $\gamma_n = 1-e^{-a}$ in interval $[1,fn]$ for any constant $0 < f < 1$.
    
    For $k = \Theta(n)$, Lemma~\ref{lemma: Erdos-Renyi - growth and shrinking} and Lemma~\ref{lemma: Erdos-Renyi - shrinking} yield that
    \[
    	\left(1 - \tfrac{1-e^{-a}}{n}\right)^k - O\left(\tfrac{\log^2 n}{n}\right)
    	\le 1-p_k
    	\le \left(1 - \tfrac{1-e^{-a}}{n}\right)^k + k \cdot O\left(\tfrac1{n^2}\right).
    \]
    Substituting $k$ by $n-u$ and applying Corollary~\ref{cor:prelim:(1-p/n)^(n-u)}, we obtain for any $u < n$ that
    \[
    	\exp\left(-1+e^{-a}\right) - O\left(\tfrac{\log^2 n}{n}\right)
    	\le 1-p_{n-u}
    	\le \exp\left(-1+e^{-a}\right) \left(1 + 2\left(1-e^{-a}\right)\tfrac un\right) + O\left(\tfrac1n\right).
    \]
    Therefore, the protocol satisfies the upper exponential shrinking conditions with $\rho_n = 1-e^{-a}$ and the lower exponential shrinking conditions with $\rho_n = 1-e^{-a} + O\left(\tfrac{\log^2n}{n}\right)$ in the interval $[n-gn,n]$ for any $0 < g < 1$.

	Since the intervals of exponential growth and exponential shrinking overlap, it follows from Theorems~\ref{th:exp-growth-upper},~\ref{th:exp-growth-lower},~\ref{th:exp-shrinking-upper},~and~\ref{th:exp-shrinking-lower} that the expected spreading time $\E[T]$ is equal to $\log_{1-e^{-a}}n + \tfrac1{1-e^{-a}}\ln n \pm O(1)$ and $\Pr[|T-\E[T]| \ge r] \le A' e^{-\alpha' r}$ for suitable constants $A', \alpha' > 0$.
\end{proof}

%
%

%% file: push_pull_ignorance.tex
\section{Limited Incoming Calls Capacity} \label{sec:single incoming call}

For all the protocols discussed above the nodes are allowed to be called several times in one round.
For some processes such as protocols considered in Section~\ref{sec:dynamic graphs}, the number of calls received by each node is at most constant.
However in most of rumor spreading processes such number can be unbounded.
For example, consider the basic push-pull protocol from Section~\ref{sec:push-pull} on the complete graph with $n$ vertices.
Since each round all nodes make calls, the maximum number of incoming calls received by the same node in one round is the same as the maximum load of a bin in the well-known problem of throwing uniformly and independently at random $n$ balls into $n$ bills, i.e., $\frac{\log n}{\log\log n} \cdot (1+o(1))$.
Such phenomenon can impact the scalability of the rumor spreading process: typically the time gap between rounds is bounded, but each round with high probability there is at least one node which have to finish $\omega(1)$ transactions.

The simplest solution is to limit the incoming ``capacity'' of nodes, i.e., the number of calls they can reply in one round.
In this section we propose a \emph{single incoming call} setting -- any node can reply to only one incoming call per round chosen uniformly at random among all received calls in current round.
All other calls are considered ``dropped", i.e., they cannot transfer the rumor.
Therefore, each node participates in at most two rumor transactions per round, whatever is the size of the network.

On the other hand, we expect the noticeable slowdown for the protocols based on the single incoming call setting compared to the usual unlimited ``capacity'' setting.
Thus we will show in Section~\ref{sec:single-push-pull} that the single incoming call push-pull protocol satisfies the single exponential shrinking conditions instead of double exponential shrinking and the corresponding expected rumor spreading time is equal to $\log_{3-2/e}n + \tfrac12\ln n \pm O(1)$.
In Section~\ref{sec:single-pull} we argue that since $\Theta(n)$ nodes are informed, the push calls of informed nodes becomes inefficient and they are responsible for such considerable slowdown.
Finally, in Section~\ref{sec:single call-fast} we combine a single incoming call push-pull protocol with pull protocol and provide a not memoryless process with spreading time
$\log_{3-2/e}n + \log_2\ln n + O(1)$.

Before proceeding to the computations, we observe that the following setting is equivalent to the single incoming call model.
In each round we choose uniformly at random a permutation $\sigma \in S_n$.
The element $\sigma_n$ is the \emph{order} of the outgoing call of node $x_i$, we write $ord_i = \sigma_i$.
Each node accepts the call with the lowest order among its received incoming calls.
We call such construction the \emph{ordered calls} setting.

\subsection{Single Incoming Call Push-Pull Protocol}\label{sec:single-push-pull}


%

\begin{theorem} \label{th:single push-pull}
	The expected spreading time for the single incoming call push-pull protocol is $\log_{3-2/e} n + \tfrac12\ln n + O(1)$.
\end{theorem}

In this section we keep the notation from the previous ones, i.e. $X_i$ is the random indicator variable corresponding to the event ``uninformed node $x_i$ gets informed in considered round''.
Since all considered protocols are uniform, we denote by $p_k$ the probability $\Pr[X_i=1]$ for the round started with $k$ informed nodes and any $i$.
In addition we denote by $Y_i, Z_i$ the indicator random variables for the following events.
\begin{itemize}
    \item [$Y_i$] ``Node $i$ is called and the first incoming call comes from an informed node.''
    \item [$Z_i$] ``The outgoing call of node $i$ is accepted by an informed node.''
\end{itemize}


\begin{lemma}
    Suppose that the fraction $f$ of nodes is informed.
    Suppose node $i$ is uninformed.
    Then
    \begin{equation}
        p_{fn} = 2f\left(1-\tfrac1e\right) - f^2\left(1-\tfrac1e\right)^2 + f\cdot O\left(\tfrac1n\right) \label{eq:single-prob-informed}.
    \end{equation}
\end{lemma}
\begin{proof}
    First, we compute the probabilities of the events corresponding to $Y_i$ and $Z_i$.
    Since each node makes a call in the round, the probability that node $x_i$ is not called is equal to $(1-\tfrac1n)^n$.
    Therefore,
    \[
        \Pr[Y_i=1]
        = f \left(1-\left(1-\tfrac1n\right)^n\right)
        = f \left(1-\tfrac1e\right) + f\cdot O\left(\tfrac1n\right).
    \]

    To compute $\Pr[Z_i=1]$ we will use the ordered call model.
    Suppose that $ord_i = \ell$.
    Then, the outgoing call of node $x_i$ is accepted if all calls with orders less than $\ell$ do not call the same node.
    Since the probability that the outgoing call of node $x_i$ has order $\ell$ is equal to $\tfrac1n$, we compute
    \[
        \Pr[Z_i=1]
        = f \sum_{\ell=1}^n \tfrac1n \left(1-\tfrac1n\right)^{\ell-1}
        = f \left(1-\left(1-\tfrac1n\right)^n\right)
        = f \left(1-\tfrac1e\right) + f\cdot O\left(\tfrac1n\right).
    \]

    Since $X_i = \max\left\{Y_i,Z_i\right\}$, it remains to compute the probability of the event $Y_i=Z_i=1$.
    Suppose that $ord_i = \ell$.
    Since the outgoing call of node $x_i$ is accepted, all calls with order less than $\ell$ should go away from the $x_i$'s target, i.e., they can have only $n-1$ possible targets.
    We also remark that node $x_i$ calls informed node, so it cannot call itself.
    Thus the probability that nobody calls node $x_i$ is equal to $\left(1-\tfrac1{n-1}\right)^{i-1}\left(1-\tfrac1n\right)^{n-i}$.
    Therefore,
    \begin{align*}
        \Pr[Z_i=1|Y_i=1, \; ord_i = \ell]
        & = f\left(1 - \left(1-\tfrac1{n-1}\right)^{i-1} \left(1-\tfrac1n\right)^{n-i}\right) \\
        & = f\left(1 - \left(1-\tfrac1n\right)^n + O\left(\tfrac1n\right)\right).
    \end{align*}
    Since the probability above is independent of $\ell$, we obtain immediately that node
    \begin{align*}
        \Pr[Y_i=Z_i=1]
        & = f^2\left(1-\left(1-\tfrac1n\right)^n\right)^2 + f^2 \cdot O\left(\tfrac1n\right) \\
        & = f^2\left(1-\tfrac1e\right)^2 + f^2 \cdot O\left(\tfrac1n\right).
    \end{align*}
    The claim of lemma follows by including-excluding formula.
\end{proof}

\begin{lemma}\label{lem:single-conditional-prob}
    There exists $c \ge 0$ such that for any uninformed nodes $x_i \ne x_j$ we have
    \begin{equation}
        \Pr[X_i=1|X_j=1] \le \Pr[X_i=1] + \tfrac{c}{n}. \label{eq:single-prob-conditional}
    \end{equation}
\end{lemma}
\begin{proof}
    We say that nodes $x_i$ and $x_j$ \emph{interact} if one calls another or if they both call the same node.
    Clearly, $\Pr[x_i, x_j \text{ interact}|X_j=1] = O\left(\tfrac1n\right)$.
    Since we need to bound $\Pr[X_i=1|X_j=1]$ up to $O(\tfrac1n)$, without loss of generality we assume for the rest of the proof that nodes $x_i$ and $x_j$ do not interact.
    We say that a call interacts with a node $x_j$ if its target coincides with $x_j$ or with $x_j$'s target (by convention a call does not interact with it source).
    Denote by $I_j$ the number of calls interacting with node $x_j$ and observe that since $x_i$ and $x_j$ don't interact, no node can interact with both $x_i$ and $x_j$.
    We split the probability $\Pr[X_i=1|X_j=1]$ conditioning on the values of $I_j$ as follows.
    \begin{align*}
        \Pr[X_i=1|X_j=1] = \sum_{k=1}^n \Pr[X_i=1|X_j=1,I_j=k] \cdot \Pr[I_j=k|X_j=1].
    \end{align*}
    Our goal is to study $\Pr[X_i=1|X_j=1,I_j=k]$.
    Since $k$ nodes interact with $x_j$, there are $n-k-1$ independent calls going uniformly to $n-2$ remaining targets (except $x_j$ and $x_j$'s target).
    In addition at least $n(f - \tfrac{k+1}{n})$ of calls are made by informed nodes.
    By these two observations we deduce
    \begin{align*}
        \Pr[Y_i=1|X_j=1,I_j=k]
        & = \left(f - \tfrac{k+1}{n}\right) \left(1-(1-\tfrac1{n-2})^{n-k-1}\right) \\
        & = f \left(1-\left(1-\tfrac1n\right)^n\right) + kO\left(\tfrac1n\right)
        = f \left(1-\tfrac1e\right) + k\cdot O\left(\tfrac1n\right).
    \end{align*}
    By the similar analysis we obtain that
    \begin{align*}
        \Pr[Z_i=1|X_j=1,I_j=k] & = f \left(1-\tfrac1e\right) + k\cdot O\left(\tfrac1n\right); \\
        \Pr[Y_i=Z_i=1|X_j=1,I_j=k] & = f^2 \left(1-\tfrac1e\right)^2 + k\cdot O\left(\tfrac1n\right).
    \end{align*}
    Therefore, $\Pr[X_i=1|X_j=1,I_j=k] = \Pr[X_i=1] + k \cdot O\left(\tfrac1n\right)$.
    Since $\Expect[I_j|X_j=1] = O(1)$, we sum up by $k$ and obtain
    \begin{align*}
        \Pr[X_i=1|X_j=1]
        & = \Pr[X_i=1] + \sum_{k=1}^n kO\left(\tfrac1n\right) \cdot \Pr[I_j=k|X_j=1] \\
        & = \Pr[X_i=1] + O\left(\tfrac1n\right) \Expect[I_j|X_j=1]
        = \Pr[X_i=1] + O\left(\tfrac1n\right).
    \end{align*}
\end{proof}

\begin{proof}[Proof of Theorem~\ref{th:single push-pull}]
    Consider a round started with $k$ informed nodes.
    Substituting $f$ by $k/n$ in~\eqref{eq:single-prob-informed}, we obtain the probability part of the exponential growth conditions.
    \[
        p_k = 2\left(1-\tfrac1e\right)\cdot\tfrac{k}{n} + k^2 \cdot O\left(\tfrac1{n^2}\right).
    \]
    Multiplying~\eqref{eq:single-prob-conditional} by $p_k$ we get the covariance condition.
    Therefore the protocol satisfies the exponential growth conditions with $\gamma_n = 2(1-\tfrac1e)$.

    Denote by $u:=n-k$ the number of uninformed nodes.
    Substituting $f$ by $1-\tfrac{u}{n}$ in~\eqref{eq:single-prob-informed}, we compute
    \[
        \Pr[X_i=0] = 1-\Pr[X_i=1]
        = \tfrac1{e^2} + O\left(\tfrac1n\right).
    \]
    Since the covariance condition follows from Lemma~\ref{lem:single-conditional-prob}, the protocol satisfies the exponential shrinking conditions with ${\rho_n} = 2$.
    Therefore the expected spreading time is equal to $\log_{3-2/e} n + \tfrac12\ln n + O(1)$.
\end{proof}

\subsection{Single Incoming Call Pull-Only Protocol}\label{sec:single-pull}
We showed that the the single call push-pull protocol is significantly slower than the classic push-pull protocol.
Although protocol based on the single incoming call setting cannot be faster than the classic independent call model, we can make it noticeably faster using the following trick.
Let us consider one round of the exponential shrinking phase with $u$ uninformed nodes.
In such round there are $n-u$ push calls, each one hits uninformed node with small probability $\fun$.
On the other hand, each of $u$ pull calls touches some informed node with probability $1-\fun$.
One can conclude that push calls ``spam'' the network: they ``occupy'' other informed nodes making them inaccessible for pull calls of uninformed nodes.
This observation is verified in the following theorem.


\begin{theorem} \label{th:single pull}
	The spreading time for the single incoming call pull protocol is $\log_{2-1/e} n + \log_2\ln n + O(1)$.
\end{theorem}
\begin{proof}
	Consider one round of the protocol.
    Clearly, if $x_1$ becomes informed it ``occupies'' one informed node which cannot inform any other node in current round.
    Thus, if we condition on that $X_1=1$, then it is slightly less likely that $x_2$ becomes informed. Consequently, $\Cov[X_1,X_2] < 0$ and the covariance part of the exponential growth and double exponential shrinking conditions is satisfied.

    Again, the call with order $\ell$ is accepted with probability $\left(1-\tfrac1n\right)^{\ell-1}$.
	Since in the round started with $k$ informed nodes only $n-k$ nodes perform calls, $ord_i$ is uniformly distributed in $\{1,\ldots,n-k\}$.
    Since the probability to call an informed node is $\tfrac{k}{n}$, we compute
    \begin{equation}
        p_k = \tfrac{k}{n}\sum_{\ell=1}^{n-k} \tfrac1{n-k}\left(1-\tfrac1n\right)^{\ell-1}
        = \tfrac{k}{n-k}\left(1-\left(1-\tfrac1n\right)^{n-k}\right) \label{eq:single-pull-prob}.
    \end{equation}
    By Corollary~\ref{cor:prelim:(1-1/n)^(n-u)}, we have
    \[
        \left(1-\tfrac1e\right)\tfrac{k}{n} - 4\tfrac{k^2}{n^2}
        \le p_k
        \le \left(1-\tfrac1e\right)\tfrac{k}{n} + 2\left(1-\tfrac1e\right)\tfrac{k^2}{n^2}.
    \]
    So the protocol satisfies the exponential growth conditions with parameter $\gamma_n = 1-\tfrac1e$.

    If we denote by $u$ the number of uninformed nodes, from~\eqref{eq:single-pull-prob} follows the following expression.
    \[
        1-p_{n-u} = \tfrac{n-u}{u} \left(1-\left(1-\tfrac1n\right)^u\right).
    \]
    With Lemma~\ref{lem:prelim:(1-1/n)^k}, we estimate $\tfrac{u}{n} \le 1-p_{n-u} \le \tfrac{3u}{2n}$.
    The protocol hence satisfies the double exponential shrinking conditions with $\ell=2$.

    Therefore, the expected spreading time is equal to $\log_{2-1/e} n + \log_2\ln n + O(1)$.
\end{proof}

\subsection{Push-Pull Protocol with Transition Time}\label{sec:single call-fast}

Comparing Theorems~\ref{th:single push-pull}~and~\ref{th:single pull} we see that push-pull protocol still be more efficient until $\Theta(n)$ nodes are informed.
Suppose now that we join to the rumor a counter which increases by one each round, so that each informed node knows the ``age'' of the rumor.
Then the single incoming call push-pull protocol \emph{with transition time $R>0$} acts as follows.
While the age of the rumor is at most $R$, it acts as a single incoming call push-pull protocol.
After $R$ rounds of rumor spreading, all informed nodes stop calling simultaneously, so the protocol acts as the single incoming call pull protocol until nodes are informed.

\begin{theorem}
	The expected rumor spreading time of the single incoming call push-pull protocol with the transition time $R = \lceil \log_{3-2/e} n \rceil$ on the complete graph with $n$ vertices is $\log_{3-2/e} n + \log_2\ln n + O(1)$.
\end{theorem}
\begin{proof}
    In the proof of Theorem~\ref{th:single pull} we showed that the single incoming call pull protocol satisfies the double exponential shrinking conditions for all $k \in [gn,n]$ for some $0 < g < 1$.
    Denote by $I_t$ the number of informed nodes after $t$ rounds.
    Let $t := \max\{R,t'\}$, where $t'$ is the smallest time such that $I_{t'} \ge gn$.
    By construction, after round $t$ the transition protocol acts as the pull protocol.
    Therefore,
    \[
        \E[T(1,n)] \le \E[t] + \E[T(fn,n)] \le \E[t] + \log_2\ln n + O(1).
    \]
    It is easy to see that the transition protocol satisfies the conditions of Lemma~\ref{lem:general-connect} with $\ell = fn$, $m = gn$ for any $0 < f < g < 1$.
    Thus, $\E[t] \le \E[T(1,fn)] + O(1)$ for any constant $0 < f < 1$, i.e., it suffices to analyse the spreading time until $fn$ informed nodes.

    Let us consider a single incoming call push-pull protocol.
    In the proof of Theorem~\ref{th:single push-pull} we showed that the single incoming call push-pull protocol satisfies the exponential growth conditions with $\gamma_n = 2-\tfrac2e$.
    In Section~\ref{subsection:exp-growth-upper - phases} we introduced a sequence $k_j$ splitting the interval $[1,fn]$ into phases such that most of the rounds the rumor spreading process moves to exactly the next phase.
    Lemma~\ref{lem:exp-growth-expk-upper} claims that the biggest number of phase $J = \log_{1+\gamma_n}n + O(1)$.
    Since $\gamma_n = 2-\tfrac2e$, we have $J = R + O(1)$.
    To simplify the proof we suppose that $R \le J$ and $fn \le k_R$.In the proof of Theorem~\ref{th:exp-growth-upper} we showed that $T(1,k_R) \le R + \Delta r$, where $\Delta r$ is stochastically dominated by a random variable with distribution $\Geom(1-q)$ for some constant $q<1$.
    By construction, $\Delta r$ is the number of rounds during which the process stayed it the same phase.
    Therefore, after at the end of round $R$ when the protocol switches from push-pull to pull-only, we have $I_R \ge k_{R-\Delta r}$.
    By Lemma~\ref{lem:exp-growth-expk-upper}, we have $k_{R-\Delta r} \ge \tfrac{fn}{(3-2/e)^{\Delta r}}$.

    Consider now the single incoming call pull protocol.
    Let a sequence $k'_j$ defines the phases for the single incoming call pull protocol.
    Suppose that $R'$ is such that $k_{R'} \ge fn$ and that $I_R$ belongs to the phase $i$ of the single incoming call pull protocol.
    Since the single incoming call pull protocol satisfies the exponential growth conditions with $\gamma_n = 1-1/e$, we have $R'-i = \tfrac{3-2/e}{2-1/e} \Delta r + O(1)$.
    Therefore,
    \[
        \E[T(I_R,fn)] \le \E[T(k_i', k'_{R'})] \le \tfrac{3-2/e}{2-1/e} \Delta r + O(1).
    \]
    Summing over all possible values of $\Delta r$ we compute
    \[
        \E[T(1,fn)] \le r + \sum_{s=0}^R \Pr[\Delta r = s] \cdot \left(\tfrac{3-2/e}{2-1/e} s + O(1) \right)
        = R + O(1)
    \]
    Since $\Delta r$ is dominated by a random variable with distribution $\Geom(1-q)$, we have $\E[T(1,fn)] \le R + O(1)$.
    Therefore, $\E[T(1,n)] \le \log_{3-2/e} n + \log_2\ln n + O(1)$.

    To prove the lower bound we consider the following protocol.
    Suppose that any node knows the total number of informed nodes.
    The protocol acts as the single incoming call push-pull protocol until there are at least $fn$ informed nodes for some $0 < f < 1$.
    Then the protocol acts as the single incoming call pull protocol.
    Since we proved Theorems~\ref{th:single push-pull}~and~\ref{th:single pull}, the expected spreading time of such protocol is at least $\log_{3-2/e}n + \log_2\ln n + O(1)$.
    It is also easy to see that such protocol spreads the rumor slightly quicker that the protocol with the fixed transition time, so the expected spreading time is bounded from below by the same expression.
\end{proof}

%% file: article.bbl
\begin{thebibliography}{10}

\bibitem{BarabasiA99}
Albert-L\'{a}szl\'{o} Barab\'{a}si and R\'{e}ka Albert.
\newblock Emergence of scaling in random networks.
\newblock {\em Science}, 286:509--512, 1999.

\bibitem{BergerBCS05}
Noam Berger, Christian Borgs, Jennifer~T. Chayes, and Amin Saberi.
\newblock On the spread of viruses on the internet.
\newblock In {\em Proceedings of the Sixteenth Annual ACM-SIAM Symposium on
  Discrete Algorithms (SODA)}, pages 301--310. SIAM, 2005.

\bibitem{BollobasR03}
B{\'{e}}la Bollob{\'{a}}s and Oliver Riordan.
\newblock Robustness and vulnerability of scale-free random graphs.
\newblock {\em Internet Mathematics}, 1:1--35, 2003.

\bibitem{CensorHillelHKM12}
Keren Censor{-}Hillel, Bernhard Haeupler, Jonathan~A. Kelner, and Petar
  Maymounkov.
\newblock Global computation in a poorly connected world: fast rumor spreading
  with no dependence on conductance.
\newblock In {\em Proceedings of the 44th Symposium on Theory of Computing
  Conference (STOC)}, pages 961--970, 2012.

\bibitem{HillelS10}
Keren Censor{-}Hillel and Hadas Shachnai.
\newblock Partial information spreading with application to distributed maximum
  coverage.
\newblock In {\em Proceedings of the 29th Annual {ACM} Symposium on Principles
  of Distributed Computing (PODC)}, pages 161--170, 2010.

\bibitem{CensorHillelS11}
Keren Censor{-}Hillel and Hadas Shachnai.
\newblock Fast information spreading in graphs with large weak conductance.
\newblock In {\em Proceedings of the Twenty-Second Annual {ACM-SIAM} Symposium
  on Discrete Algorithms (SODA)}, pages 440--448, 2011.

\bibitem{ChierichettiLP09}
Flavio Chierichetti, Silvio Lattanzi, and Alessandro Panconesi.
\newblock Rumor spreading in social networks.
\newblock In {\em Proceedings of the 36th International Colloquium on Automata,
  Languages and Programming (ICALP)}, pages 375--386. Springer, 2009.

\bibitem{ChierichettiLP10stoc}
Flavio Chierichetti, Silvio Lattanzi, and Alessandro Panconesi.
\newblock Almost tight bounds for rumour spreading with conductance.
\newblock In {\em Proceedings of the 42nd {ACM} Symposium on Theory of
  Computing (STOC)}, pages 399--408. {ACM}, 2010.

\bibitem{ClementiCDFPS16}
Andrea E.~F. Clementi, Pierluigi Crescenzi, Carola Doerr, Pierre Fraigniaud,
  Francesco Pasquale, and Riccardo Silvestri.
\newblock Rumor spreading in random evolving graphs.
\newblock {\em Random Structures and Algorithms}, 48:290--312, 2016.

\bibitem{DaumKM15}
Sebastian Daum, Fabian Kuhn, and Yannic Maus.
\newblock Rumor spreading with bounded in-degree.
\newblock {\em CoRR}, abs/1506.00828, 2015.

\bibitem{Demers87}
Alan~J. Demers, Daniel~H. Greene, Carl Hauser, Wes Irish, John Larson, Scott
  Shenker, Howard~E. Sturgis, Daniel~C. Swinehart, and Douglas~B. Terry.
\newblock Epidemic algorithms for replicated database maintenance.
\newblock In {\em Proceedings of the Sixth Annual {ACM} Symposium on Principles
  of Distributed Computing (PODC)}, pages 1--12. {ACM}, 1987.

\bibitem{DoerrFF11}
Benjamin Doerr, Mahmoud Fouz, and Tobias Friedrich.
\newblock Social networks spread rumors in sublogarithmic time.
\newblock In {\em Proceedings of the 43rd {ACM} Symposium on Theory of
  Computing (STOC)}, pages 21--30. {ACM}, 2011.

\bibitem{DoerrFF12}
Benjamin Doerr, Mahmoud Fouz, and Tobias Friedrich.
\newblock Asynchronous rumor spreading in preferential attachment graphs.
\newblock In {\em Proceedings of the 13th Scandinavian Symposium and Workshops
  on Algorithm Theory (SWAT)}, pages 307--315. Springer, 2012.

\bibitem{DoerrFF12acm}
Benjamin Doerr, Mahmoud Fouz, and Tobias Friedrich.
\newblock Why rumors spread so quickly in social networks.
\newblock {\em Commununications of the ACM}, 55:70--75, 2012.

\bibitem{DoerrHL13}
Benjamin Doerr, Anna Huber, and Ariel Levavi.
\newblock Strong robustness of randomized rumor spreading protocols.
\newblock {\em Discrete Applied Mathematics}, 161:778--793, 2013.

\bibitem{DoerrK14}
Benjamin Doerr and Marvin K{\"{u}}nnemann.
\newblock Tight analysis of randomized rumor spreading in complete graphs.
\newblock In {\em Proceedings of the Eleventh Workshop on Analytic Algorithmics
  and Combinatorics (ANALCO)}, pages 82--91. {SIAM}, 2014.

\bibitem{ElsasserS09}
Robert Els{\"{a}}sser and Thomas Sauerwald.
\newblock On the runtime and robustness of randomized broadcasting.
\newblock {\em Theoretical Computer Science}, 410:3414--3427, 2009.

\bibitem{FeigePRU90}
Uriel Feige, David Peleg, Prabhakar Raghavan, and Eli Upfal.
\newblock Randomized broadcast in networks.
\newblock {\em Random Structures and Algorithms}, 1:447--460, 1990.

\bibitem{FountoulakisHP10}
Nikolaos Fountoulakis, Anna Huber, and Konstantinos Panagiotou.
\newblock Reliable broadcasting in random networks and the effect of density.
\newblock In {\em Proceedings of the 29th International Conference on Computer
  Communications (INFOCOM)}, pages 2552--2560. {IEEE}, 2010.

\bibitem{FountoulakisP10}
Nikolaos Fountoulakis and Konstantinos Panagiotou.
\newblock Rumor spreading on random regular graphs and expanders.
\newblock In {\em 13th International Workshop on Approximation, Randomization,
  and Combinatorial Optimization. Algorithms and Techniques (APPROX)}, pages
  560--573. Springer, 2010.

\bibitem{FountoulakisPS12}
Nikolaos Fountoulakis, Konstantinos Panagiotou, and Thomas Sauerwald.
\newblock Ultra-fast rumor spreading in social networks.
\newblock In {\em Proceedings of the Twenty-Third Annual {ACM-SIAM} Symposium
  on Discrete Algorithms (SODA)}, pages 1642--1660. {SIAM}, 2012.

\bibitem{FriezeG85}
Alan~M. Frieze and Geoffrey~R. Grimmett.
\newblock The shortest-path problem for graphs with random arc-lengths.
\newblock {\em Discrete Applied Mathematics}, 10:57--77, 1985.

\bibitem{GhaffariN16}
Mohsen Ghaffari and Calvin Newport.
\newblock How to discreetly spread a rumor in a crowd.
\newblock {\em CoRR}, abs/1607.05697, 2016.

\bibitem{Giakkoupis11}
George Giakkoupis.
\newblock Tight bounds for rumor spreading in graphs of a given conductance.
\newblock In {\em 28th International Symposium on Theoretical Aspects of
  Computer Science (STACS)}, pages 57--68. Schloss Dagstuhl - Leibniz-Zentrum
  fuer Informatik, 2011.

\bibitem{Giakkoupis14}
George Giakkoupis.
\newblock Tight bounds for rumor spreading with vertex expansion.
\newblock In {\em Proceedings of the Twenty-Fifth Annual {ACM-SIAM} Symposium
  on Discrete Algorithms (SODA)}, pages 801--815. {SIAM}, 2014.

\bibitem{GiakkoupisS12}
George Giakkoupis and Thomas Sauerwald.
\newblock Rumor spreading and vertex expansion.
\newblock In {\em Proceedings of the Twenty-Third Annual {ACM-SIAM} Symposium
  on Discrete Algorithms (SODA)}, pages 1623--1641. {SIAM}, 2012.

\bibitem{Haeupler13}
Bernhard Haeupler.
\newblock Simple, fast and deterministic gossip and rumor spreading.
\newblock In {\em Proceedings of the Twenty-Fourth Annual {ACM-SIAM} Symposium
  on Discrete Algorithms (SODA)}, pages 705--716, 2013.

\bibitem{IwanickiS10}
Konrad Iwanicki and Maarten van Steen.
\newblock Gossip-based self-management of a recursive area hierarchy for large
  wireless sensornets.
\newblock {\em IEEE Transactions on Parallel and Distributed Systems},
  21:562--576, 2010.

\bibitem{KarpSSV00}
Richard~M. Karp, Christian Schindelhauer, Scott Shenker, and Berthold
  V{\"{o}}cking.
\newblock Randomized rumor spreading.
\newblock In {\em Proceedings of the Annual Symposium on Foundations of
  Computer Science (FOCS)}, pages 565--574. {IEEE}, 2000.

\bibitem{KiwiC17}
Marcos~A. Kiwi and Christopher~Thraves Caro.
\newblock {FIFO} queues are bad for rumor spreading.
\newblock {\em {IEEE} Trans. Information Theory}, 63:1159--1166, 2017.

\bibitem{Kleinberg08}
Jon~M. Kleinberg.
\newblock The convergence of social and technological networks.
\newblock {\em Communications of the ACM}, 51:66--72, 2008.

\bibitem{MatosSFOR12}
Miguel Matos, Valerio Schiavoni, Pascal Felber, Rui Oliveira, and Etienne
  Riviere.
\newblock {BRISA:} combining efficiency and reliability in epidemic data
  dissemination.
\newblock In {\em 26th {IEEE} International Parallel and Distributed Processing
  Symposium (IPDPS)}, pages 983--994, 2012.

\bibitem{MehrabianP14}
Abbas Mehrabian and Ali Pourmiri.
\newblock Randomized rumor spreading in poorly connected small-world networks.
\newblock In {\em Proceedings of the 28th International Symposium on
  Distributed Computing (DISC)}, pages 346--360. Springer, 2014.

\bibitem{MoskAoyamaS06}
Damon Mosk{-}Aoyama and Devavrat Shah.
\newblock Computing separable functions via gossip.
\newblock In {\em Proceedings of the Twenty-Fifth Annual {ACM} Symposium on
  Principles of Distributed Computing, (PODC)}, pages 113--122. {ACM}, 2006.

\bibitem{MoskAoyamaS08}
Damon Mosk{-}Aoyama and Devavrat Shah.
\newblock Fast distributed algorithms for computing separable functions.
\newblock {\em {IEEE} Trans. Information Theory}, 54:2997--3007, 2008.

\bibitem{PanagiotouPS15}
Konstantinos Panagiotou, Ali Pourmiri, and Thomas Sauerwald.
\newblock Faster rumor spreading with multiple calls.
\newblock {\em The Electronic Journal of Combinatorics}, 22:P1.23, 2015.

\bibitem{Pittel87}
Boris Pittel.
\newblock On spreading a rumor.
\newblock {\em SIAM Journal on Applied Mathematics}, 47:213--223, 1987.

\end{thebibliography}
